\documentclass[preprint,a4paper,11pt,onecolumn,numbers,sort&compress]{elsarticle}

\usepackage{graphicx}
\usepackage{epsfig}
\usepackage{natbib}

\usepackage{url}
\usepackage{amsmath,amsfonts,bm,mathrsfs,amssymb,amsthm}
\usepackage{subfig}
\usepackage{multirow}
\usepackage{array}
\usepackage{multicol}
\usepackage{color}
\usepackage[normalem]{ulem}

\DeclareMathOperator{\sgn}{sgn}
\newtheorem{theorem}{Theorem}
\newtheorem{claim}{Claim}
\newtheorem{observation}{\bf Observation}
\newtheorem{lemma}{Lemma}
\newtheorem{proposition}{Proposition}
\newtheorem{definition}{Definition}
\newtheorem{corollary}{Corollary}
\newdefinition{remark}{Remark}

\journal{Journal of Mathematical Economics}

\usepackage{bm}
\usepackage{pgfplots}
\usepackage{tikz}
\usepackage{lipsum}
\usetikzlibrary{arrows,calc,shapes, snakes, intersections}

\usepgfplotslibrary{fillbetween}
\usetikzlibrary{patterns}

\begin{document}

\begin{frontmatter}

\title{Optimal Mechanisms for Selling Two Items to a Single Buyer Having Uniformly Distributed Valuations\footnote{A preliminary version of this paper appeared in {\em Proceedings of the 12th International Conference on Web and Internet Economics} (WINE), 2016, Montreal, Canada \cite{TRN16}.}}
\author[a1]{D.~Thirumulanathan\corref{cor1}}
\ead{thirumulanathan@gmail.com}
\author[a2]{Rajesh Sundaresan}
\ead{rajeshs@iisc.ac.in}
\author[a3]{Y Narahari}
\ead{narahari@iisc.ac.in}

\cortext[cor1]{Corresponding author}

\address[a1]{Department of Electrical Communication Engineering, Indian Institute of Science, Bengaluru, India 560012}
\address[a2]{Department of Electrical Communication Engineering, and Robert Bosch Centre for Cyber-Physical Systems, Indian Institute of Science, Bengaluru, India 560012}
\address[a3]{Department of Computer Science and Automation, Indian Institute of Science, Bengaluru, India 560012}

\begin{abstract}
We consider the design of a revenue-optimal mechanism when two items are available to be sold to a single buyer whose valuation is uniformly distributed over an arbitrary rectangle $[c_1,c_1+b_1]\times[c_2,c_2+b_2]$ in the positive quadrant. We provide an explicit, complete solution for arbitrary nonnegative values of $(c_1,c_2,b_1,b_2)$. We identify eight simple structures, each with at most $4$ (possibly stochastic) menu items, and prove that the optimal mechanism has one of these eight structures. We also characterize the optimal mechanism as a function of $(c_1,c_2,b_1,b_2)$. The structures indicate that the optimal mechanism involves (a) an interplay of individual sale and a bundle sale when $c_1$ and $c_2$ are low, (b) a bundle sale when $c_1$ and $c_2$ are high, and (c) an individual sale when one of them is high and the other is low. To the best of our knowledge, our results are the first to show the existence of optimal mechanisms with no exclusion region. We further conjecture, based on promising preliminary results, that our methodology can be extended to a wider class of distributions.
\end{abstract}


\end{frontmatter}

\section{Introduction}
This paper studies the design of a revenue optimal mechanism for selling two items. While the solution to the problem of designing an optimal mechanism for selling a single item is well-known (\citet{Mye81}), optimal mechanism design for selling multiple items is a harder problem. Though there are many partial results available in the literature, finding a general solution, even in the simplest setting with two heterogeneous items and one buyer, remains open.

In this paper, we consider the problem of optimal mechanism design in the two-item one-buyer setting, when the valuations of the buyer are uniformly distributed in arbitrary rectangles in the positive quadrant. In the two-item one-buyer setting, \citet{DDT13,DDT15,DDT17} considered the most general class of distributions till date, and gave the optimal solution when the distribution gives rise to a so-called ``well-formed'' canonical partition (to be described in Section \ref{sec:zero}) of the support set. \citet{GK15} provided closed form solutions when $m=2$ and the distribution of $z$ satisfies some sufficient conditions. The papers \cite{DDT13,DDT15,DDT17} and \cite{GK15} rely on a result of \citet{Roc87} that transforms the search for an optimal mechanism into a search for the utility function of the buyer. This function represents the expected value of the lottery the buyer receives minus the payment to the seller. The expected payment to the seller is maximized, subject to the utility function being positive, increasing, convex, and $1$-Lipschitz. The above papers then identify a dual problem, solve it, and exploit this solution to identify a primal solution. We call this approach the {\em duality approach} in this paper.

The duality approach developed in these papers crucially uses the assumption that the support set $D$ of the distribution is $[0,b_1]\times[0,b_2]$. We are aware of only two examples where the support sets $[c,c+1]^2, c>0$, and $[4,16]\times[4,7]$, for which the exact solutions are known, are not bordered by the coordinate axes. These were considered by \citet{Pav11} and \citet{DDT17}, respectively. \citet{DDT13,DDT15,DDT17} do consider other distributions but these distributions must satisfy $f(z) (z\cdot n(z))=0$ on the boundaries of $D$, where $n(z)$ is the normal to the boundary at $z$. The uniform distribution on arbitrary rectangles (which we consider in this paper) has $f(z) (z\cdot n(z))<0$ in general on the left and bottom boundaries, and this requires additional nontrivial care in its handling.


\subsection{Our contributions}

\begin{figure}
\centering
\begin{tabular}{cccc}
\subfloat[]{\label{fig:a1}\begin{tikzpicture}[scale=0.2,font=\scriptsize,axis/.style={very thick, -}]
\node at (1.5,-1) {\tiny$(c_1,c_2)$};
\draw [axis,thick,-] (0,0)--(12,0);
\node at (9,-1) {\tiny$(c_1+b_1,c_2)$};
\draw [axis,thick,-] (0,0)--(0,12);
\node [above] at (3,12) {\tiny$(c_1,c_2+b_2)$};
\draw [axis,thick,-] (0,12)--(12,12);
\draw [axis,thick,-] (12,0)--(12,12);
\draw [axis,thick,-] (0,8)--(4,7);
\draw [axis,thick,-] (7,4)--(4,7);
\draw [axis,thick,-] (7,4)--(8,0);
\draw [axis,thick,dotted] (4,7)--(4,12);
\draw [axis,thick,dotted] (7,4)--(12,4);
\node at (2.5,2.5) {$(0,0)$};
\node at (9.75,1.75) {$(1,a_2)$};
\node at (2,9) {$(a_1,1)$};
\node at (8,8) {$(1,1)$};
\end{tikzpicture}}&
\subfloat[]{\label{fig:b1}\begin{tikzpicture}[scale=0.2,font=\scriptsize,axis/.style={very thick, -}]
\node at (1.5,-1) {\tiny$(c_1,c_2)$};
\draw [axis,thick,-] (0,0)--(12,0);
\node at (9,-1) {\tiny$(c_1+b_1,c_2)$};
\draw [axis,thick,-] (0,0)--(0,12);
\node [above] at (3,12) {\tiny$(c_1,c_2+b_2)$};
\draw [axis,thick,-] (0,12)--(12,12);
\draw [axis,thick,-] (12,0)--(12,12);
\draw [axis,thick,-] (0,6.5)--(4,5.5);
\draw [axis,thick,-] (9.5,0)--(4,5.5);
\draw [axis,thick,dotted] (4,5.5)--(4,12);
\node at (2.5,2.5) {$(0,0)$};
\node at (2,8.5) {$(a_1,1)$};
\node at (9,6) {$(1,1)$};
\end{tikzpicture}}&
\subfloat[]{\label{fig:f1}\begin{tikzpicture}[scale=0.2,font=\scriptsize,axis/.style={very thick, -}]
\node at (1.5,-1) {\tiny$(c_1,c_2)$};
\draw [axis,thick,-] (0,0)--(12,0);
\node at (9,-1) {\tiny$(c_1+b_1,c_2)$};
\draw [axis,thick,-] (0,0)--(0,12);
\node [above] at (3,12) {\tiny$(c_1,c_2+b_2)$};
\draw [axis,thick,-] (0,12)--(12,12);
\draw [axis,thick,-] (12,0)--(12,12);
\draw [axis,thick,-] (0,9.5)--(5.5,4);
\draw [axis,thick,-] (5.5,4)--(6.5,0);
\draw [axis,thick,dotted] (5.5,4)--(12,4);
\node at (2.5,2.5) {$(0,0)$};
\node at (9,2) {$(1,a_2)$};
\node at (6,9) {$(1,1)$};
\end{tikzpicture}}&
\subfloat[]{\label{fig:c1}\begin{tikzpicture}[scale=0.2,font=\scriptsize,axis/.style={very thick, -}]
\node at (1.5,-1) {\tiny$(c_1,c_2)$};
\draw [axis,thick,-] (0,0)--(12,0);
\node at (9,-1) {\tiny$(c_1+b_1,c_2)$};
\draw [axis,thick,-] (0,0)--(0,12);
\node [above] at (3,12) {\tiny$(c_1,c_2+b_2)$};
\draw [axis,thick,-] (0,12)--(12,12);
\draw [axis,thick,-] (12,0)--(12,12);
\draw [axis,thick,-] (0,8)--(8,0);
\node at (2.5,2.5) {$(0,0)$};
\node at (8,8) {$(1,1)$};
\end{tikzpicture}}\\
\subfloat[]{\label{fig:d1}\begin{tikzpicture}[scale=0.2,font=\scriptsize,axis/.style={very thick, -}]
\node at (1.5,-1) {\tiny$(c_1,c_2)$};
\draw [axis,thick,-] (0,0)--(12,0);
\node at (9,-1) {\tiny$(c_1+b_1,c_2)$};
\draw [axis,thick,-] (0,0)--(0,12);
\node [above] at (3,12) {\tiny$(c_1,c_2+b_2)$};
\draw [axis,thick,-] (0,12)--(12,12);
\draw [axis,thick,-] (12,0)--(12,12);
\draw [axis,thick,-] (0,4)--(5,0);
\draw [axis,thick,-] (6,0)--(6,12);
\node at (1.5,1) {\tiny$(0,0)$};
\node at (3,6) {$(a_1,1)$};
\node at (9,6) {$(1,1)$};
\end{tikzpicture}}&
\subfloat[]{\label{fig:e1}\begin{tikzpicture}[scale=0.2,font=\scriptsize,axis/.style={very thick, -}]
\node at (1.5,-1) {\tiny$(c_1,c_2)$};
\draw [axis,thick,-] (0,0)--(12,0);
\node at (9,-1) {\tiny$(c_1+b_1,c_2)$};
\draw [axis,thick,-] (0,0)--(0,12);
\node [above] at (3,12) {\tiny$(c_1,c_2+b_2)$};
\draw [axis,thick,-] (0,12)--(12,12);
\draw [axis,thick,-] (12,0)--(12,12);
\draw [axis,thick,-] (6,0)--(6,12);
\node at (3,6) {$(0,1)$};
\node at (9,6) {$(1,1)$};
\end{tikzpicture}}&
\subfloat[]{\label{fig:g1}\begin{tikzpicture}[scale=0.2,font=\scriptsize,axis/.style={very thick, -}]
\node at (1.5,-1) {\tiny$(c_1,c_2)$};
\draw [axis,thick,-] (0,0)--(12,0);
\node at (9,-1) {\tiny$(c_1+b_1,c_2)$};
\draw [axis,thick,-] (0,0)--(0,12);
\node [above] at (3,12) {\tiny$(c_1,c_2+b_2)$};
\draw [axis,thick,-] (0,12)--(12,12);
\draw [axis,thick,-] (12,0)--(12,12);
\draw [axis,thick,-] (0,5)--(4,0);
\draw [axis,thick,-] (0,6)--(12,6);
\node at (1.5,1) {\tiny$(0,0)$};
\node at (6,3) {$(1,a_2)$};
\node at (6,9) {$(1,1)$};
\end{tikzpicture}}&
\subfloat[]{\label{fig:h1}\begin{tikzpicture}[scale=0.2,font=\scriptsize,axis/.style={very thick, -}]
\node at (1.5,-1) {\tiny$(c_1,c_2)$};
\draw [axis,thick,-] (0,0)--(12,0);
\node at (9,-1) {\tiny$(c_1+b_1,c_2)$};
\draw [axis,thick,-] (0,0)--(0,12);
\node [above] at (3,12) {\tiny$(c_1,c_2+b_2)$};
\draw [axis,thick,-] (0,12)--(12,12);
\draw [axis,thick,-] (12,0)--(12,12);
\draw [axis,thick,-] (0,6)--(12,6);
\node at (6,3) {$(1,0)$};
\node at (6,9) {$(1,1)$};
\end{tikzpicture}}
\end{tabular}
\caption{An illustration of all possible structures that an optimal mechanism can have.}\label{fig:gen-structure}
\end{figure}
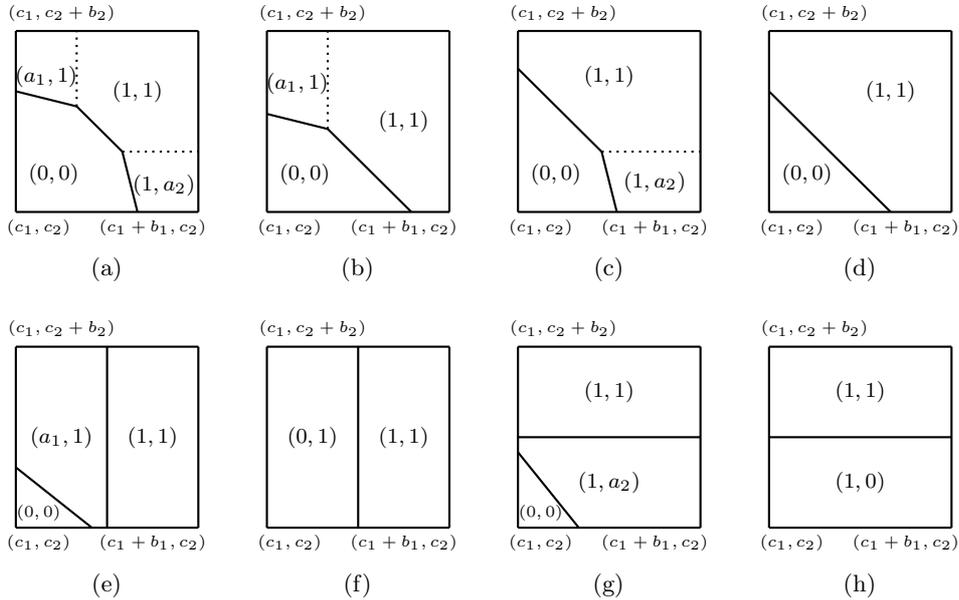

Our contributions can be summarized as follows.
\begin{enumerate}
\item[(i)] We solve the two-item single-buyer optimal mechanism design problem when $z\sim\mbox{Unif}[c_1,c_1+b_1]\times[c_2,c_2+b_2]$, for arbitrary nonnegative values of $(c_1,c_2,b_1,b_2)$. We compute the exact solution using a nontrivial extension of the duality approach designed in \cite{DDT17}. The extension takes care of cases where $f(z) (z\cdot n(z))<0$ on the left and bottom boundaries.
\item[(ii)] We prove that the structure of the optimal solution falls within a class of eight simple structures, each with at most four menu items\footnote{The phrase ``menu items'' refers to regions of constant allocation.} (see Figure \ref{fig:gen-structure}). This is in agreement with the result of \citet{WT14}.
\item[(iii)] To the best of our knowledge, our results are the first to show the existence of optimal mechanisms without an exclusion region, i.e., valuations that result in no sale (see Figures \ref{fig:e1} and \ref{fig:h1}). The results in \citet{Arm96} and \citet{PBBK14} assert that the optimal multi-dimensional mechanism has an exclusion region, but under some sufficient conditions on the distributions and the utility functions. \citet{Arm96} assumes that the support set of the distribution of valuations is strictly convex, and \citet{PBBK14} assume that the utility function is strictly concave in the allocations. Neither of these assumptions hold in our setting.
\item[(iv)] We also establish another interesting property of the optimal mechanism: given any value of $c_1$, we find a threshold value of $c_2$ beyond which the mechanism becomes deterministic (see Section \ref{sec:sol-space}).
\end{enumerate}

Some qualitative results on the structure of optimal mechanisms were already known when $z$'s distribution satisfied certain conditions. For instance, Pavlov \cite{Pav10} considers distributions with negative power rate\footnote{A distribution $f$ is said to have negative power rate when $\Delta(z_1,z_2)=-3-z_1f_1'(z_1)/f_1(z_1)-z_2f_2'(z_2)/f_2(z_2)\leq 0$ for every $(z_1,z_2)$ in the support set of $f$. It is said to have negative constant power rate if in addition to being negative, $\Delta$ is also a constant for every $(z_1,z_2)$ in the support set of $f$.}, while Wang and Tang \cite{WT14} consider distributions with negative constant power rate. Our work considers uniform distributions which is a special case of the settings in both \cite{Pav10} and \cite{WT14}. So the optimal mechanisms in our setting have allocations only of the form $(0,0)$, $(a_1,1)$ and $(1,a_2)$, in accordance with Pavlov's result, and have at most four menu items, in accordance with Wang and Tang's result. In our special case, we are able to prove stronger results. We show that the optimal mechanisms cannot have a structure other than one of the eight structures depicted in Figure \ref{fig:gen-structure}. Our results bring out some less expected structures like those in Figures \ref{fig:e1} and \ref{fig:h1}. Beyond qualitative descriptions, our results enable computation of the optimal mechanism for the uniform distribution on any arbitrary rectangle.

An alternative approach is to use the results by \citet{WT14} and \citet{Pav11} in the following way. First, the optimal mechanism can be parametrized to have at most four constant allocation regions with allocation probabilities $(0,0)$, $(a_1,1)$, $(1,a_2)$, and $(1,1)$, and then the revenue can be optimized over the allocation and payment variables. The approach appears to be simpler than the duality approach because it optimizes only over five variables as opposed to the infinite-dimensional optimization in the duality approach. We show in \ref{app:e} that this approach leads to non-concave objective functions whose global optima are generally harder to compute. We further show that the first order conditions of this optimization problem are structure-specific. Moreover, transitions between various structures (in Figure \ref{fig:gen-structure}, for example) are not captured using this method. See Proposition \ref{prop:opt-specific} and the discussions following the proposition in \ref{app:e}. On the other hand, the duality approach that we use in this paper gives a certificate of global optimality and captures transitions between various structures.

The optimal mechanism for various values of $(c_1,c_2,b_1,b_2)$ are detailed in Theorems \ref{thm:str-1}-\ref{thm:figc}. The phase diagram in Figure \ref{fig:partition_uni} represents how the optimal mechanism varies when the parameters $(c_1,c_2,b_1,b_2)$ vary. We interpret the results as follows.
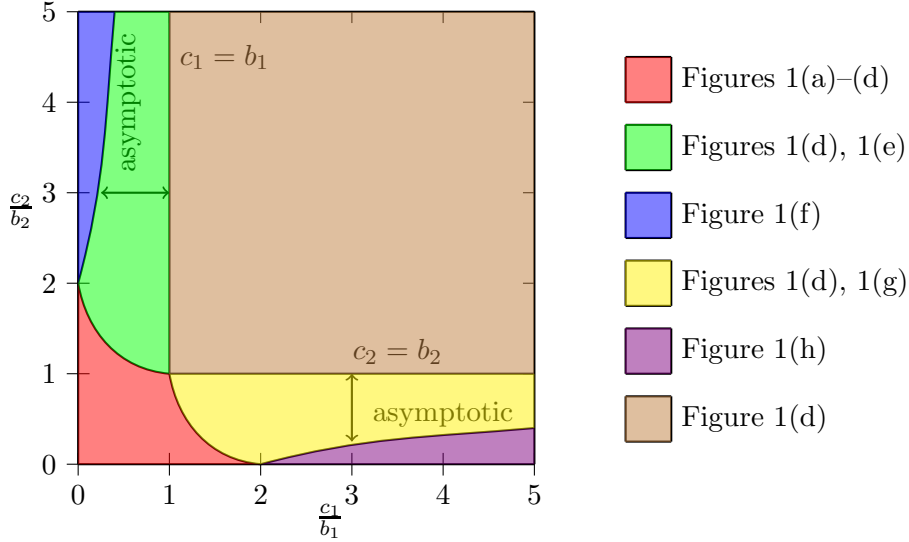
\begin{figure}[t]
\centering
\begin{tikzpicture}[scale=0.6,font=\normalsize,axis/.style={very thick, ->, >=stealth'}]
\draw [axis,thick,-] (0,0)--(0,10);
\node [right] at (5,-1.25) {$\frac{c_1}{b_1}$};
\draw [axis,thick,-] (0,0)--(10,0);
\node [above] at (-1.25,5) {$\frac{c_2}{b_2}$};
\draw [axis,thick,-] (10,0)--(10,10);
\draw [axis,thick,-] (0,10)--(10,10);
\draw [thin,-] (-0.25,0) -- (0.25,0);
\node [left] at (-0.25,0) {$0$};
\draw [thin,-] (-0.25,2) -- (0.25,2);
\node [left] at (-0.25,2) {$1$};
\draw [thin,-] (-0.25,4) -- (0.25,4);
\node [left] at (-0.25,4) {$2$};
\draw [thin,-] (-0.25,6) -- (0.25,6);
\node [left] at (-0.25,6) {$3$};
\draw [thin,-] (-0.25,8) -- (0.25,8);
\node [left] at (-0.25,8) {$4$};
\draw [thin,-] (-0.25,10) -- (0.25,10);
\node [left] at (-0.25,10) {$5$};
\draw [thin,-] (0,-0.25) -- (0,0.25);
\node [below] at (0,-0.25) {$0$};
\draw [thin,-] (2,-0.25) -- (2,0.25);
\node [below] at (2,-0.25) {$1$};
\draw [thin,-] (4,-0.25) -- (4,0.25);
\node [below] at (4,-0.25) {$2$};
\draw [thin,-] (6,-0.25) -- (6,0.25);
\node [below] at (6,-0.25) {$3$};
\draw [thin,-] (8,-0.25) -- (8,0.25);
\node [below] at (8,-0.25) {$4$};
\draw [thin,-] (10,-0.25) -- (10,0.25);
\node [below] at (10,-0.25) {$5$};
\draw [thin,-] (9.75,2) -- (10,2);
\draw [thin,-] (9.75,4) -- (10,4);
\draw [thin,-] (9.75,6) -- (10,6);
\draw [thin,-] (9.75,8) -- (10,8);
\draw [thin,-] (9.75,10) -- (10,10);
\draw [thin,-] (2,9.75) -- (2,10);
\draw [thin,-] (4,9.75) -- (4,10);
\draw [thin,-] (6,9.75) -- (6,10);
\draw [thin,-] (8,9.75) -- (8,10);
\draw [thin,-] (10,9.75) -- (10,10);
\draw [thick,-] (0,4) to[out=-80,in=175] (2,2);
\draw [thick,-] (4,0) to[out=175,in=-80] (2,2);
\draw [thick,-] (2,2)--(2,10);
\node [right] at (2,9) {$c_1=b_1$};
\draw [thick,-] (2,2)--(10,2);
\node [above] at (7,2) {$c_2=b_2$};
\draw [thick,-] (4,0) to[out=15,in=-175] (10,0.8);
\node [right] at (0.6,8) {\rotatebox{90}{asymptotic}};
\draw [thick,<->] (6,0.5) to (6,2);
\draw [thick,-] (0,4) to[out=75,in=-95] (0.8,10);
\node [above] at (8,0.6) {asymptotic};
\draw [thick,<->] (0.5,6) to (2,6);


\path[fill=red,opacity=.5] (0,4) to[out=-80,in=175] (2,2) to (2,2) to[out=-80,in=175] (4,0) to (0,0) -- (0,4) ;

\path[fill=green,opacity=.5] (0,4) to[out=-80,in=175] (2,2) to (2,2) to (2,10) to (0.8,10) to[out=-95,in=75] (0,4);
 
\path[fill=blue,opacity=.5] (0,4) to[out=75,in=-95] (0.8,10)--(0,10)--(0,4);

\path[fill=yellow,opacity=.5] (4,0) to[out=175,in=-80] (2,2) to (2,2) to (10,2) to (10,0.8) to[out=-175,in=15] (4,0);
 
\path[fill=violet,opacity=.5] (4,0) to[out=15,in=-175] (10,0.8)--(10,0)--(4,0);
 
\path[fill=brown,opacity=.5] (2,10)--(2,2)--(10,2)--(10,10)--(2,10);


\draw [thick,-] (12,8) to (12,9);
\draw [thick,-] (12,9) to (13,9);
\draw [thick,-] (13,8) to (13,9);
\draw [thick,-] (12,8) to (13,8);
\path[fill=red,opacity=.5] (12,8) to (13,8) to (13,9) to (12,9) to (12,8);
\node [right] at (13,8.5) {Figures \ref{fig:gen-structure}(a)--(d)};

\draw [thick,-] (12,6.5) to (12,7.5);
\draw [thick,-] (12,7.5) to (13,7.5);
\draw [thick,-] (13,6.5) to (13,7.5);
\draw [thick,-] (12,6.5) to (13,6.5);
\path[fill=green,opacity=.5] (12,6.5) to (13,6.5) to (13,7.5) to (12,7.5) to (12,7.5);
\node [right] at (13,7) {Figures \ref{fig:gen-structure}(d), \ref{fig:gen-structure}(e)};

\draw [thick,-] (12,5) to (12,6);
\draw [thick,-] (12,6) to (13,6);
\draw [thick,-] (13,5) to (13,6);
\draw [thick,-] (12,5) to (13,5);
\path[fill=blue,opacity=.5] (12,5) to (13,5) to (13,6) to (12,6) to (12,5);
\node [right] at (13,5.5) {Figure \ref{fig:gen-structure}(f)};

\draw [thick,-] (12,3.5) to (12,4.5);
\draw [thick,-] (12,4.5) to (13,4.5);
\draw [thick,-] (13,3.5) to (13,4.5);
\draw [thick,-] (12,3.5) to (13,3.5);
\path[fill=yellow,opacity=.5] (12,3.5) to (13,3.5) to (13,4.5) to (12,4.5) to (12,3.5);
\node [right] at (13,4) {Figures \ref{fig:gen-structure}(d), \ref{fig:gen-structure}(g)};

\draw [thick,-] (12,2) to (12,3);
\draw [thick,-] (12,3) to (13,3);
\draw [thick,-] (13,2) to (13,3);
\draw [thick,-] (12,2) to (13,2);
\path[fill=violet,opacity=.5] (12,2) to (13,2) to (13,3) to (12,3) to (12,2);
\node [right] at (13,2.5) {Figure \ref{fig:gen-structure}(h)};

\draw [thick,-] (12,0.5) to (12,1.5);
\draw [thick,-] (12,1.5) to (13,1.5);
\draw [thick,-] (13,0.5) to (13,1.5);
\draw [thick,-] (12,0.5) to (13,0.5);
\path[fill=brown,opacity=.5] (12,0.5) to (13,0.5) to (13,1.5) to (12,1.5) to (12,0.5);
\node [right] at (13,1) {Figure \ref{fig:gen-structure}(d)};
\end{tikzpicture}
\caption{A phase diagram for the optimal mechanism}\label{fig:partition_uni}
\end{figure}

\begin{itemize}
 \item Consider the case when $(c_1,c_2)$ is low. Then, the seller knows that the buyer could have very low valuations, and thus sets a high threshold to sell any item $i$. Further, he finds it optimal to set the threshold for the bundle of two items to be less than the sum of the thresholds for the individual items (see Figure \ref{fig:a1}). Thus the case of low valuations turns out to be an interplay between selling the items individually and selling the items as a bundle.
 \item When $c_1$ is low but $c_2$ is high, the seller offers item $2$ for the least valuation $c_2$, but sets a threshold to sell item $1$ (see Figure \ref{fig:e1}). This is because the expected revenue gained by exclusion of certain valuations is always dominated by the revenue lost from them, and thus the seller finds no reason to withhold item $2$ for any valuation profile; see Section \ref{sec:sol-space} for a more precise explanation. A similar case occurs when $c_2$ is low but $c_1$ is high. In both of these cases, the seller optimizes revenue by selling the items individually.
 \item When $(c_1,c_2)$ is high, the seller sets a threshold close to $c_1+c_2$ to sell the bundle (see Figure \ref{fig:c1}). The threshold converges to the least possible valuation profile for the bundle, $c_1+c_2$, when $(c_1,c_2)\rightarrow\infty$. In this case, the seller optimizes revenue by selling the items as a bundle.
 \item Starting with $(c_1,c_2)=(0,0)$, consider moving the support set rectangle to infinity either horizontally or vertically or both. Then, the optimal mechanism starts as an interplay between individual sale and bundle sale, and ends up either as a bundle sale or as an individual sale, based on whether both $(c_1,c_2)$ are high, or only one of them is high. All the other structures (Figures \ref{fig:a1}, \ref{fig:b1}, \ref{fig:f1}, \ref{fig:d1}, and \ref{fig:g1}) are intermediates on the way to the bundle sale or the individual sale mechanisms.
\end{itemize}

The wide range of structures that we obtain for various values of $(c_1,c_2)$ indicates why support sets that are not bordered by the coordinate axes require significant analytical effort. In a companion paper \cite{TRN17b}, we prove that the structure of optimal mechanism in the two-item unit-demand setting falls within a class of five simple structures, when $z\sim\mbox{Unif}[c,c+b_1]\times[c,c+b_2]$.

\subsection{Our method}

Our method is as follows. From \cite{DDT15} and \cite{DDT17}, we know that the dual problem is an optimal transport problem that transfers mass from the support set $D$ to itself, subject to the constraint that the difference between the mass densities transferred out of the set and transferred into the set {\em convex dominates} a signed measure defined by the distribution of the valuations. When $(c_1,c_2)=(0,0)$, \citet{DDT13} provided a solution where the difference between the densities not just convex dominates the signed measure, but equals it. In another work, \citet{DDT17} provided an example ($z\sim \mbox{Unif}[4,16]\times[4,7]$) where the difference between the densities convex dominates but does not equal the signed measure. They construct a line measure that convex dominates $0$, add it to the signed measure, and then solve the example using the same method that they used to solve the $(c_1,c_2)=(0,0)$ case. They call this line measure as the ``shuffling measure''.

Their method can be used to find the optimal solution for more general distributions, provided we know the structure of the appropriate shuffling measure that yields the solution. Prior to our work, it was not clear what shuffling measure would work, even for the restricted setting of uniform distributions but over arbitrary rectangles. We prove that for the setting of uniform distributions, the optimal solution is always arrived at by using shuffling (line) measures of certain specific forms. We do the following.
\begin{itemize}
 \item[$\bullet$] We start with the shuffling (line) measure in \cite{DDT17}. We parametrize this measure using its depth, slope, and the length, and find the relation between these parameters so that the line measure convex dominates $0$.
 \item[$\bullet$] We then derive the conditions that these parameters must satisfy in order to solve the optimal transport problem. The conditions turn out to be polynomial equations of degree at most $4$.
 \item[$\bullet$] We identify conditions on the parameters $(c_1, c_2, b_1, b_2)$ so that the solutions to the polynomials yield a valid partition of $D$ (allocation probability is at most $1$ and the canonical partition is within $D$). We thus arrive at eight different structures. We then prove that the structure of the optimal mechanism is one of the eight, for any $(c_1,c_2,b_1,b_2)\geq 0$.
 \item[$\bullet$] We then generalize the shuffling measure, and use it to solve an example linear density function. This leads us to the conjecture that this class of shuffling measures will help identify the optimal mechanism for distributions with negative constant power rate.
\end{itemize}

Our work thus provides a method to construct an appropriate shuffling measure, and hence to arrive at the optimal solution, for various distributions whose support sets are not bordered by the coordinate axes. In our view, this is an important step towards understanding optimal mechanisms in multi-item settings. Furthermore, we feel that the special cases we have solved could act as guidelines to solve the problem of computing optimal mechanisms for general distributions or to suggest good menus for practical settings.

\subsection{Prior work}\label{sub:prior}

Several works have attempted to provide partial characterizations of optimal mechanisms in multi-dimensional settings. \citet{HN13} considered the case when valuations are correlated, and proved that if the size of the menu\footnote{The phrase ``size of the menu'' refers to the number of regions of constant allocation.} is bounded, then there exist joint distributions such that the revenue obtained from the mechanism is an arbitrarily small fraction of the optimal revenue. On the other hand, \citet{WT14}, \cite{TW17} proved that the optimal mechanisms have simple structures with at most four menu items, when the distribution of the buyer's valuation satisfies a certain power rate condition (which is satisfied by the uniform distribution). However, the exact structures and associated allocations were left open. \citet{HR15} provided examples of cases where the buyer's values for the items increase but the optimal revenue decreases. \citet{HH15} did a reverse mechanism design; they constructed a mechanism, identified conditions under which there exists a virtual valuation, and thereby established that the mechanism is optimal.


In the recent literature, more number of explicit solutions have become available. \citet{MV06,MV07} derived the optimal deterministic mechanism when the support set of the distribution of buyer's valuation $z$ is $[0,1]^2$. Further, they derived conditions under which the optimal deterministic mechanism is also overall optimal. \citet{GK14} provided a solution when the number of items $m\leq 6$ and the buyer's valuations $z=(z_1,\ldots,z_m)\sim\mbox{Unif}[0,1]^m$. In another work, \citet{Gian15} considered the case of arbitrary number of items, and proved that when $z_i\sim\mbox{exp}(\lambda)$ for all $i$, and are independent, then bundling maximizes revenue. \citet{MHJ15} proved that bundling is optimal when (i) the virtual valuations are nonnegative for all possible types, and (ii) the distribution satisfies a certain power rate condition.

In the two-item setting, exact solutions were provided by \citet{DDT13,DDT15,DDT17} and \citet{GK15} for a rich class of distributions. \citet{KM16} identified the dual when the valuation space is convex and the space of allocations are restricted. They also solved examples when the allocations are restricted to satisfy either the unit-demand constraint or the deterministic constraint.

Exact solutions for the setting we consider in this paper have largely been computed using the duality approach of \cite{DDT15}. Significant exceptions are the works by \citet{Pav06,Pav10,Pav11}. Pavlov's works computed the optimal solution when $z\sim\mbox{Unif}[c,c+1]^2, c>0$, using the {\em virtual valuation} approach.

There has also been significant interest in finding approximately optimal solutions when the distribution of the buyer's valuation satisfies certain conditions. See \cite{BGN17}, \cite{BILW14}, \cite{Bhat10, BCKW10, CDW12a, CDW12b, CDW13, CDW16, Cai13, CZ17}, \cite{CMS15}, \cite{CM16}, \cite{DW11}, \cite{DW12}, \cite{HN12}, \cite{Yao14} for relevant literature on approximate solutions. In this paper, however, our focus shall be on exact solutions.


The rest of the paper is organized as follows. In Section \ref{sec:prelim}, we formulate an optimization problem that describes the two-item single-buyer optimal mechanism. In Section \ref{sec:zero}, we revisit the case (solved in \cite{DDT13}) when the left-bottom corner of the support set is at $(0,0)$. This is done to set up the context and notation for the rest of the paper. In Section \ref{sec:sol-space}, we discuss the space of solutions and highlight a few interesting outcomes. In Section \ref{sec:gencase}, we define the shuffling measures required to find the optimal mechanism for the uniform distribution on arbitrary rectangles. We prove that the optimal solution has one of the eight simple structures. In Section \ref{sec:extension}, we discuss the conjecture and possible extension to other classes of distributions. In Section \ref{sec:conclude}, we provide some concluding remarks. In a companion paper \cite{TRN17b}, we address the unit-demand setting.

\section{Preliminaries}\label{sec:prelim}
Consider a two-item one-buyer setting. The buyer's valuation is $z=(z_1,z_2)$ for the two items, sampled according to the joint density $f(z)=f_1(z_1)f_2(z_2)$, where $f_1(z_1)$ and $f_2(z_2)$ are marginal densities. The support set of $f$ is defined as $D:=\{z:f(z)>0\}$. Throughout this paper, we restrict attention to an arbitrary rectangle $D=[c_1,c_1+b_1]\times[c_2,c_2+b_2]$, where $c_1, c_2, b_1, b_2$ are all nonnegative.

Our aim is to design an optimal mechanism. By the revelation principle \cite[Prop.~9.25]{NRTV07}, it suffices to focus on direct mechanisms. Further, we focus on mechanisms where the buyer has a quasilinear utility. Specifically, we assume an allocation function $q:D\rightarrow[0,1]^2$ and a payment function $t:D \rightarrow \mathbb{R}_+$ that represent, respectively, the probabilities of allocation of the items to the buyer and the amount that the buyer pays. In other words, for a reported valuation vector $\hat{z} = (\hat{z}_1, \hat{z}_2)$, item $i$ is allocated with probability $q_i(\hat{z})$, where $(q_1(\hat{z}), q_2(\hat{z})) = q(\hat{z})$, and the seller collects a revenue of $t(\hat{z})$ from the buyer. If the buyer's true valuation is $z$, and he reports $\hat{z}$, then his (quasilinear) utility is $\hat{u}(z, \hat{z}) := z \cdot q(\hat{z}) - t(\hat{z})$, which is the expected value of the lottery he receives minus the payment.

A mechanism $(q,t)$ is {\em incentive compatible} when truth telling is a weakly dominant strategy for the buyer, i.e., $\hat{u}(z,z) \geq \hat{u}(z,\hat{z})$ for every $z,\hat{z}\in D$. In this case the buyer's realized utility is
\begin{equation}
  \label{eqn:u(z)}
  u(z) := \hat{u}(z, z) = z\cdot q(z)-t(z).
\end{equation}
 The following result is well known:
\begin{theorem}\cite{Roc87}\/.\label{thm:Rochet}
A mechanism $(q,t)$, with $u(z) = z\cdot q(z)-t(z)$, is incentive compatible iff $u$ is continuous, convex and $\nabla u(z)=q(z)$ for a.e. $z\in D$.
\end{theorem}

An incentive compatible mechanism is {\em individually rational} if the buyer is not worse off by participating in the mechanism, i.e., $u(z) \geq 0$ for every $z \in D$, with zero being the buyer's utility if he chooses not to participate.

An optimal mechanism is one that maximizes the expected revenue to the seller subject to incentive compatibility and individual rationality.

Define a measure $\bar{\mu}$, supported on set $D$, as follows.
$$
  \bar{\mu}(A)=\int_D\mathbf{1}_A(z)\mu(z)\,dz+\int_{\partial D}\mathbf{1}_A(z)\mu_s(z)\,dz+\mu_p(A\cap(c_1,c_2))
$$
for all measurable sets $A$, where the functions $\mu$, $\mu_s$ and $\mu_p$ are defined as
\begin{align*}
\mu(z)&:=-z\cdot\nabla f(z)-3f(z),\\
\mu_s(z)&:=\begin{cases}(z\cdot n(z))f(z)&z\in\partial D\\0&\mbox{otherwise},\end{cases}\\
\mu_p(z)&:=\delta_{(c_1,c_2)}(z).
\end{align*}
The vector $n(z)$ denotes the normal to the surface $\partial D$ at $z$, and the notation $\delta$ in $\mu_p(z)$ denotes the Dirac-delta function. So $\mu_p(z)=1$ if $z=(c_1,c_2)$ and $0$ otherwise. We regard $\mu$ as the density of a signed measure that is absolutely continuous with the two-dimensional Lebesgue measure, and $\mu_s$ as the density of a signed measure that is absolutely continuous with the one-dimensional Lebesgue measure. We shall refer to both Lebesgue measures as $dz$. Observe that $\bar{\mu}$ is a signed Radon measure on $D$, and that $\mu$ and $\mu_s$ are the Radon-Nikodym derivatives of respective components of $\bar{\mu}$ w.r.t. the two-dimensional and one-dimensional Lebesgue measures respectively. We now state the following theorem.

\begin{theorem}\cite[Thm.~1]{DDT17}
The problem of determining the optimal incentive compatible and individually rational mechanism for a single additive buyer whose values for the goods are distributed according to $f:D\rightarrow\mathbb{R}$ is equivalent to solving the optimization problem
\begin{align}\label{eqn:primal}
  &\max_u \int_D u\,d\bar{\mu} \\
  \mbox{ subject to }\hspace*{.5in}
  &(a)\,u\mbox{ continuous, convex, increasing},\nonumber\\
  &(b)\,u(z)-u(z')\leq\|z-z'\|_1\,\forall z,z'\in D.\nonumber
\end{align}
\end{theorem}

We now recall the definition of the convex ordering relation. A function $f$ is increasing if $z\geq z'$ component-wise implies $f(z)\geq f(z')$.
\begin{definition} (See for e.g., \cite{DDT17}\/)
 Let $\alpha$ and $\beta$ be measures defined on a set $D$. We say $\alpha$ {\em first-order dominates} $\beta$ ($\alpha\succeq_{1}\beta$) if $\int_Df\,d\alpha\geq\int_Df\,d\beta$ for all continuous and increasing $f$. We say $\alpha$ {\em convex-dominates} $\beta$ ($\alpha\succeq_{cvx}\beta$) if $\int_Df\,d\alpha\geq\int_Df\,d\beta$ for all continuous, convex and increasing $f$.
\end{definition}

The dual of problem (\ref{eqn:primal}) is (see \cite[Thm.~2]{DDT17}):
\begin{align}\label{eqn:dual}
  &\inf_{\gamma} \int_{D\times D}\|z-z'\|_1\,d\gamma(z,z') \\
  \mbox{ subject to }\hspace*{.5in}
  &(a)\,\gamma\in Radon_+(D\times D)\nonumber\\
  &(b)\,\gamma(\cdot,D)=\gamma_1(\cdot),\,\gamma(D,\cdot)=\gamma_2(\cdot),\,\gamma_1-\gamma_2\succeq_{cvx}\bar{\mu}.\nonumber
\end{align}
By $\gamma\in Radon_+(D\times D)$, we mean that $\gamma$ is an unsigned Radon measure in $D\times D$. We derive the weak duality result in \ref{app:d} to provide an understanding of how the dual arises and why $\gamma$ may be interpreted as prices for violating the primal constraint.

The next lemma gives a sufficient condition for strong duality.
\begin{lemma}\cite[Cor.~1]{DDT17}\/.\label{lem:conditions}
Let $u^*$ and $\gamma^*$ be feasible for the aforementioned primal (\ref{eqn:primal}) and dual (\ref{eqn:dual}) problems, respectively\/. Then the objective functions of (\ref{eqn:primal}) and (\ref{eqn:dual}) with $u=u^*$ and $\gamma=\gamma^*$ are equal iff (i) $\int_Du^*\,d(\gamma_1^*-\gamma_2^*)=\int_Du^*\,d\bar{\mu}$, and (ii) $u^*(z)-u^*(z')=\|z-z'\|_1, \gamma^*-$a.e.
\end{lemma}
\begin{remark}
Condition (i) requires the integrals of the utility function with respect to (w.r.t.) $\gamma_1^*-\gamma_2^*$ and w.r.t. $\bar{\mu}$  be equal. Observe that 
since $u^*$ and $\gamma^*$ are feasible for (\ref{eqn:primal}) 
and (\ref{eqn:dual}), respectively, the integral w.r.t. $\gamma_1^*-\gamma_2^*$ is at least the integral w.r.t. $\bar{\mu}$. Condition (i) imposes that they are equal.
\end{remark}

\begin{remark}\label{rem:condition-ii}
The constraint of the primal (\ref{eqn:primal}) imposes that $u^*(z)-u^*(z') \leq \|z-z'\|_1$ for all $z, z' \in D$. Condition (ii) requires equality for $\gamma^*$-almost every $(z,z')$.
\end{remark}

Our problem now reduces to that of finding a $\gamma$ such that $\gamma_1-\gamma_2$ convex-dominates $\bar{\mu}$ and satisfies the conditions stated in Lemma \ref{lem:conditions}. The key, nontrivial technical contribution of our paper is to identify such a $\gamma$ when $z\sim\mbox{Unif}[c_1,c_1+b_1]\times[c_2,c_2+b_2]$, for all $(c_1,c_2,b_1,b_2)\geq 0$.

\section{Revisiting the Case When $(c_1,c_2)=(0,0)$}\label{sec:zero}
We first solve the case when $z\sim\mbox{Unif}[0,b_1]\times[0,b_2]$. The solution proceeds exactly according to the general characterization in \cite[Sect.~7]{DDT17} that uses the technique of optimal transport. We provide it here to set up the notation for the more general case that will follow in later sections. In particular, the dual variable $\gamma$ obtained here and the method to describe it will be used in later sections. Observe that

\begin{align}
  \mbox{(area density) } \, & \mu(z)=-3/(b_1b_2),\quad z\in D \nonumber \\
  \mbox{(line density) } \, & \mu_s(z)=(b_2\mathbf{1}\{z_2=b_2\}+b_1\mathbf{1}\{z_1=b_1\})/(b_1b_2)\quad z\in\partial D \nonumber\\
  \mbox{(point measure) } \, & \mu_p(\{0,0\})=1. \label{eqn:mu-zero}
\end{align}
In other words, $\mu$ is the density of a two-dimensional measure with density $-3/(b_1b_2)$ everywhere, $\mu_s$ is the density of a one-dimensional measure of line densities $b_2/(b_1b_2)$ and $b_1/(b_1b_2)$ on the top and right boundaries, respectively, and $\mu_p$ is a point measure of $1$ on $(0,0)$.

We now enumerate the steps to construct the so-called {\em canonical partition} of the support set $D$, with respect to $\bar{\mu}$-measure.
\begin{enumerate}
\item[(a)] Define the {\em outer boundary functions} $s_i:[0,b_i)\rightarrow[0,b_{-i})$, $i=1,2$, as
\begin{equation}\label{eqn:si-zi}
  s_i(z_i):=\max\left\{z_{-i}^*\in[0,b_{-i}):\int_{z_{-i}^*}^{b_{-i}}\mu(z)\,dz_{-i}+\mu_s(z_i,b_{-i})=0\right\}.
\end{equation}
Compute the functions $s_1(z_1)$ for all $z_1\in[0,b_1)$, and $s_2(z_2)$ for all $z_2\in[0,b_2)$\footnote{The functions $s_i(\cdot)$ are constructed so that the integrals of the components of $\bar{\mu}(z_1,z_2)$ over $z_2\in[s_1(z_1),b_2]$ equal zero for each $z_1$, and the integrals over $[s_2(z_2),b_1]$ equal zero for each $z_2$. This enables us to transfer mass from points of excess to those with deficit, and meet the excess exactly with the deficit. This will soon become clear when we construct the dual measure $\gamma$.}.
\item[(b)] Construct the {\em exclusion set} $Z$ to be the set formed by the intersection of $\{(z_1,z_2):z_2\leq s_1(z_1)\}$, $\{(z_1,z_2):z_1\leq s_2(z_2)\}$, and $\{(z_1,z_2):z_1+z_2\leq p\}$, where $p$ is chosen such that $\bar{\mu}(Z)=0$\footnote{Again, $\bar{\mu}(Z)=0$ enables us to transfer mass from points of excess to those with deficit.}. We call $p$ the {\em critical price}.
\item[(c)] Let the point of intersection between the curves $\{z_2= s_1(z_1)\}$ and $\{z_1+z_2=p\}$ be denoted by $P$, and the point of intersection between $\{z_1=s_2(z_2)\}$ and $\{z_1+z_2=p\}$ be denoted by $Q$. Let $P=(P_1,P_2)$ and $Q=(Q_1,Q_2)$ denote the respective co-ordinates of $P$ and $Q$. We call these points $P$ and $Q$ as the {\em critical points}. We now partition $D\backslash Z$ as follows.
\begin{itemize}
 \item[$\bullet$] $A:=([c_1,P_1]\times[P_2,c_2+b_2])\backslash Z$
 \item[$\bullet$] $B:=([Q_1,c_1+b_1]\times[c_2,Q_2])\backslash Z$
 \item[$\bullet$] $W:=D\backslash(Z\cup A\cup B)$.
\end{itemize}
\item[(d)] The partition of $D$ into $A$, $B$, $W$ and $Z$, as explained in items (a)-(c), is termed a canonical partition\footnote{The terms {\em outer boundary function}, {\em critical price}, {\em critical points}, and {\em canonical partition} are defined based on \cite[Def.~12]{DDT17} and \cite[Def.~13]{DDT17}.}.
\end{enumerate}

We now define the allocation and the payment functions $(q(z),t(z))$, considering each region of the canonical partition constructed above as a constant allocation region.
\begin{equation}\label{eqn:q-full}
  (q_1(z),q_2(z),t(z))=\begin{cases}(0,0,0)&\mbox{if }z\in Z\\(-s_1'(z_1),1,s_1(z_1)-z_1s_1'(z_1))&\mbox{if }z\in A\\(1,-s_2'(z_2),s_2(z_2)-z_2s_2'(z_2))&\mbox{if }z\in B\\(1,1,p)&\mbox{if }z\in W.\end{cases}
\end{equation}

\begin{figure}[h!]
\centering\begin{tabular}{ccc}
\subfloat[]{\begin{tikzpicture}[scale=0.4,font=\small,axis/.style={very thick, ->, >=stealth'}]
\draw [axis,thick,-] (1,1)--(9,1);
\draw [axis,thick,-] (9,1)--(9,9);
\draw [axis,thick,-] (9,9)--(1,9);
\draw [axis,thick,-] (1,9)--(1,1);
\draw [thick,dotted] (1,19/3)--(9,19/3);
\draw [thick,dotted] (19/3,1)--(19/3,9);
\node [above] at (4,19/3) {$s_1(z_1)$};
\node [above,rotate=90] at (19/3,4) {$s_2(z_2)$};
\end{tikzpicture}}&
\subfloat[]{\begin{tikzpicture}[scale=0.4,font=\small,axis/.style={very thick, ->, >=stealth'}]
\draw [axis,thick,-] (1,1)--(9,1);
\draw [axis,thick,-] (9,1)--(9,9);
\draw [axis,thick,-] (9,9)--(1,9);
\draw [axis,thick,-] (1,9)--(1,1);
\draw [thick,dotted] (1,19/3)--(9,19/3);
\draw [thick,dotted] (19/3,1)--(19/3,9);
\draw [axis,thick,-] (2.5,19/3)--(19/3,2.5);
\node [above,rotate=-45] at (53/12,53/12) {\scriptsize$z_1+z_2=p$};
\path[fill=gray!50,opacity=.9] (1,19/3) to (2.5,19/3) to (19/3,2.5) to (19/3,1) to (1,1) to (1,19/3);
\node at (3,3) {\Large$Z$};
\end{tikzpicture}}&
\subfloat[]{\begin{tikzpicture}[scale=0.4,font=\small,axis/.style={very thick, ->, >=stealth'}]
\draw [axis,thick,-] (1,1)--(9,1);
\draw [axis,thick,-] (9,1)--(9,9);
\draw [axis,thick,-] (9,9)--(1,9);
\draw [axis,thick,-] (1,9)--(1,1);
\draw [thick,dotted] (1,19/3)--(9,19/3);
\draw [thick,dotted] (19/3,1)--(19/3,9);
\draw [axis,thick,-] (2.5,19/3)--(19/3,2.5);
\path[fill=gray!50,opacity=.9] (1,19/3) to (2.5,19/3) to (19/3,2.5) to (19/3,1) to (1,1) to (1,19/3);
\node at (3,3) {\Large$Z$};
\draw [axis,thick,-] (2.5,19/3)--(2.5,9);
\draw [axis,thick,-] (19/3,2.5)--(9,2.5);
\path[fill=gray!50,opacity=.6] (1,19/3) to (2.5,19/3) to (2.5,9) to (1,9) to (1,19/3);
\node at (1.75,23/3) {\large$A$};
\path[fill=gray!50,opacity=.6] (19/3,1) to (19/3,2.5) to (9,2.5) to (9,1) to (19/3,1);
\node at (23/3,1.75) {\large$B$};
\path[fill=gray!50,opacity=.3] (2.5,19/3) to (19/3,2.5) to (9,2.5) to (9,9) to (2.5,9) to (2.5,19/3);
\node at (19/3,19/3) {\huge$W$};
\draw[black,fill=black] (2.5,19/3) circle (1ex);
\node [below] at (2.5,19/3) {$P$};
\draw[black,fill=black] (19/3,2.5) circle (1ex);
\node [left] at (19/3,2.5) {$Q$};
\end{tikzpicture}}
\end{tabular}
\caption{When $1\leq b_1/b_2\leq 2$, we illustrate (a) the construction of $s_i(z_i)$, (b) the construction of exclusion set $Z$ such that $\bar{\mu}(Z)=0$, and (c) the partition of $D\backslash Z$ into $A$, $B$, and $W$.}\label{fig:illust-sec3}
\end{figure}
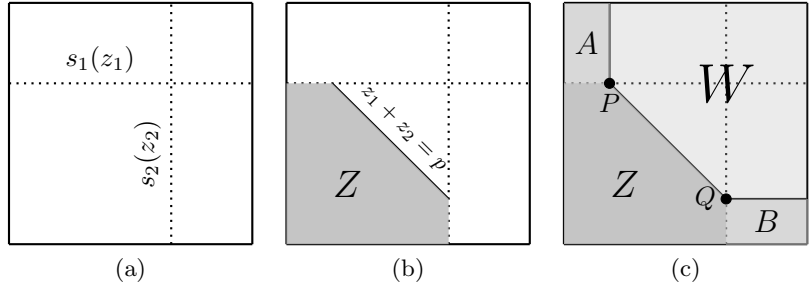

We now construct the canonical partition of $D$ using the steps enumerated above in (a)-(d), when the valuation $z\sim\mbox{Unif}[0,b_1]\times[0,b_2]$, with $1\leq b_1/b_2\leq 2$. The construction is illustrated in Figure \ref{fig:illust-sec3}.
\begin{enumerate}
\item[(a)] Using the values of $\mu$, $\mu_s$, and $\mu_p$ in (\ref{eqn:mu-zero}) and the formula for outer boundary functions $s_i(z_i)$ in (\ref{eqn:si-zi}), it is easy to see that $s_1(z_1)=2b_2/3$ for every $z_1\in[0,b_1)$, and $s_2(z_2)=2b_1/3$ for every $z_2\in[0,b_2)$.
\item[(b)] We now choose the critical price $p$ so that $\bar{\mu}(Z)=0$. It is easy to see that choosing
$$
  p=\frac{2(b_1+b_2)-\sqrt{2b_1b_2}}{3}
$$
makes $\bar{\mu}(Z)=0$.
\item[(c)] We now compute the critical points $P$ and $Q$. They are given by
$$
  P=((2b_1-\sqrt{2b_1b_2})/3,2b_2/3),\,Q=(2b_1/3,(2b_2-\sqrt{2b_1b_2})/3).
$$
But the points $P$ and $Q$ are valid only when they fall within $[0,b_1]\times[0,b_2]$. Note that $P_1\geq 0$ holds since $b_1\geq b_2$, and that $Q_2\geq 0$ holds since $(2b_2-\sqrt{2b_1b_2})/3\geq 0$ whenever $b_1/b_2\leq 2$. We thus have the following canonical partition of $D\backslash Z$, when $1 \leq b_1/b_2 \leq 2$:
$$
  A=[0,P_1]\times[2b_2/3,b_2],\,B=[2b_1/3,b_1]\times[0,Q_2],\,W=D\backslash(A\cup B\cup Z).
$$
\end{enumerate}

\begin{figure}[h!]
\centering
\begin{tabular}{ccc}
\subfloat[]{\begin{tikzpicture}[scale=0.35,font=\small,axis/.style={very thick, ->, >=stealth'}]
\draw [axis,thick,-] (1,1)--(9,1);
\draw [axis,thick,-] (9,1)--(9,7);
\draw [axis,thick,-] (9,7)--(1,7);
\draw [axis,thick,-] (1,7)--(1,1);
\draw [thick,dotted] (1,5)--(9,5);
\draw [thick,dotted] (19/3,1)--(19/3,7);
\node [above] at (4,5) {$s_1(z_1)$};
\node [above,rotate=90] at (19/3,3) {$s_2(z_2)$};
\end{tikzpicture}}&
\subfloat[]{\begin{tikzpicture}[scale=0.35,font=\small,axis/.style={very thick, ->, >=stealth'}]
\draw [axis,thick,-] (1,1)--(9,1);
\draw [axis,thick,-] (9,1)--(9,7);
\draw [axis,thick,-] (9,7)--(1,7);
\draw [axis,thick,-] (1,7)--(1,1);
\draw [thick,dotted] (1,5)--(9,5);
\draw [thick,dotted] (19/3,1)--(19/3,7);
\draw [axis,thick,-] (2,5)--(6,1);
\node [above,rotate=-45] at (4,3) {\scriptsize$z_1+z_2=p$};
\path[fill=gray!50,opacity=.9] (1,5) to (2,5) to (6,1) to (1,1) to (1,5);
\node at (2.5,2.5) {\Large$Z$};
\end{tikzpicture}}&
\subfloat[]{\begin{tikzpicture}[scale=0.35,font=\small,axis/.style={very thick, ->, >=stealth'}]
\draw [axis,thick,-] (1,1)--(9,1);
\draw [axis,thick,-] (9,1)--(9,7);
\draw [axis,thick,-] (9,7)--(1,7);
\draw [axis,thick,-] (1,7)--(1,1);
\draw [thick,dotted] (1,5)--(9,5);
\draw [thick,dotted] (19/3,1)--(19/3,7);
\draw [axis,thick,-] (2,5)--(6,1);
\path[fill=gray!50,opacity=.9] (1,5) to (2,5) to (6,1) to (1,1) to (1,5);
\node at (2.5,2.5) {\Large$Z$};
\draw [axis,thick,-] (2,5)--(2,7);
\path[fill=gray!50,opacity=.6] (1,5) to (2,5) to (2,7) to (1,7) to (1,5);
\node at (1.5,6) {\large$A$};
\path[fill=gray!50,opacity=.3] (2,5) to (6,1) to (9,1) to (9,7) to (2,7) to (2,5);
\node at (19/3,5) {\huge$W$};
\draw[black,fill=black] (2,5) circle (1ex);
\node [below] at (2,5) {$P$};
\end{tikzpicture}}
\end{tabular}
\caption{When $b_1/b_2>2$, we illustrate (a) the construction of $s_i(z_i)$, (b) the construction of exclusion set $Z$ such that $\bar{\mu}(Z)=0$, and (c) the partition of $D\backslash Z$ into $A$ and $W$.}\label{fig:illust-sec3-2}
\end{figure}
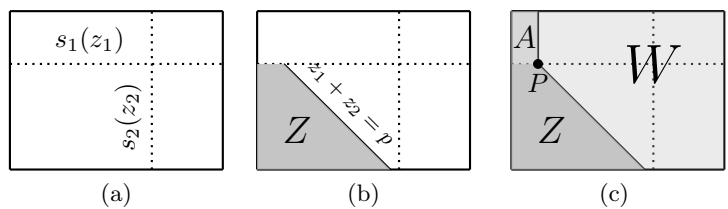

When $b_1/b_2=2$, we have $Q_2=0$, reducing region $B$ to a null.

When $b_1/b_2>2$, we construct a canonical partition of $D$ by taking $B$ to be null. The construction is illustrated in Figure \ref{fig:illust-sec3-2}.
\begin{enumerate}
\item[(a)] The outer boundary functions $s_i(z_i)$ remain the same as in the earlier case.
\item[(b)] To construct the exclusion set $Z$, we need to construct $p$ such that the $\bar{\mu}$-measure of $Z=\{(z_1,z_2):z_2\leq 2b_2/3, z_1+z_2\leq p\}$ is zero. It is easy to derive the critical price $p=b_1/2+b_2/3$.
\item[(c)] The critical point $P$ remains the same as in the earlier case. We now have the following canonical partition of $D\backslash Z$, when $b_1/b_2 > 2$:
$$
  A=[0,P_1]\times[2b_2/3,b_2],\,W=D\backslash(A\cup Z).
$$
\end{enumerate}

Observe that the allocation function turns out to be as depicted in Figure \ref{fig:figa} when $1\leq b_1/b_2\leq 2$, and as in Figure \ref{fig:figb} when $b_1/b_2>2$. 

{\bf The outer boundary functions $s_i(\cdot)$ and the critical price $p$:} Recall from the characterization results of \citet{Pav11} and \citet{WT14} that the optimal mechanism for the uniform distribution has at most four constant allocation regions, and that the allocations must be of the form $(0,0)$, $(a_1,1)$, and $(1,a_2)$. From the expressions of $(q,t)$ in (\ref{eqn:q-full}), the outer boundary function $s_1$ determines the price of the allocation $(a_1,1)$, and its slope determines the value of $a_1$. Similarly, the outer boundary function $s_2$ determines the price of allocation $(1,a_2)$, and its slope determines the value of $a_2$. The critical price $p$ used to construct the exclusion set $Z$ determines the price of the allocation $(1,1)$. Thus all the five parameters in the optimal mechanism -- two allocation parameters and three price parameters -- are determined by the outer boundary functions $s_1$, $s_2$, and the critical price $p$.

We now proceed to prove that our constructions are indeed the optimal mechanisms.

\begin{proposition}\label{prop:known}
 Consider $z\sim\mbox{Unif }[0,b_1]\times[0,b_2]$ and suppose that $b_1\geq b_2>0$. Then the optimal mechanism is as depicted in Figure \ref{fig:figa} when $1 \leq b_1/b_2 \leq 2$, and is as depicted in Figure \ref{fig:figb} when $b_1/b_2 \geq 2$.
\end{proposition}

\begin{figure}
\centering
\begin{tabular}{cc}
\subfloat[]{\label{fig:figa}\includegraphics[scale=.15]{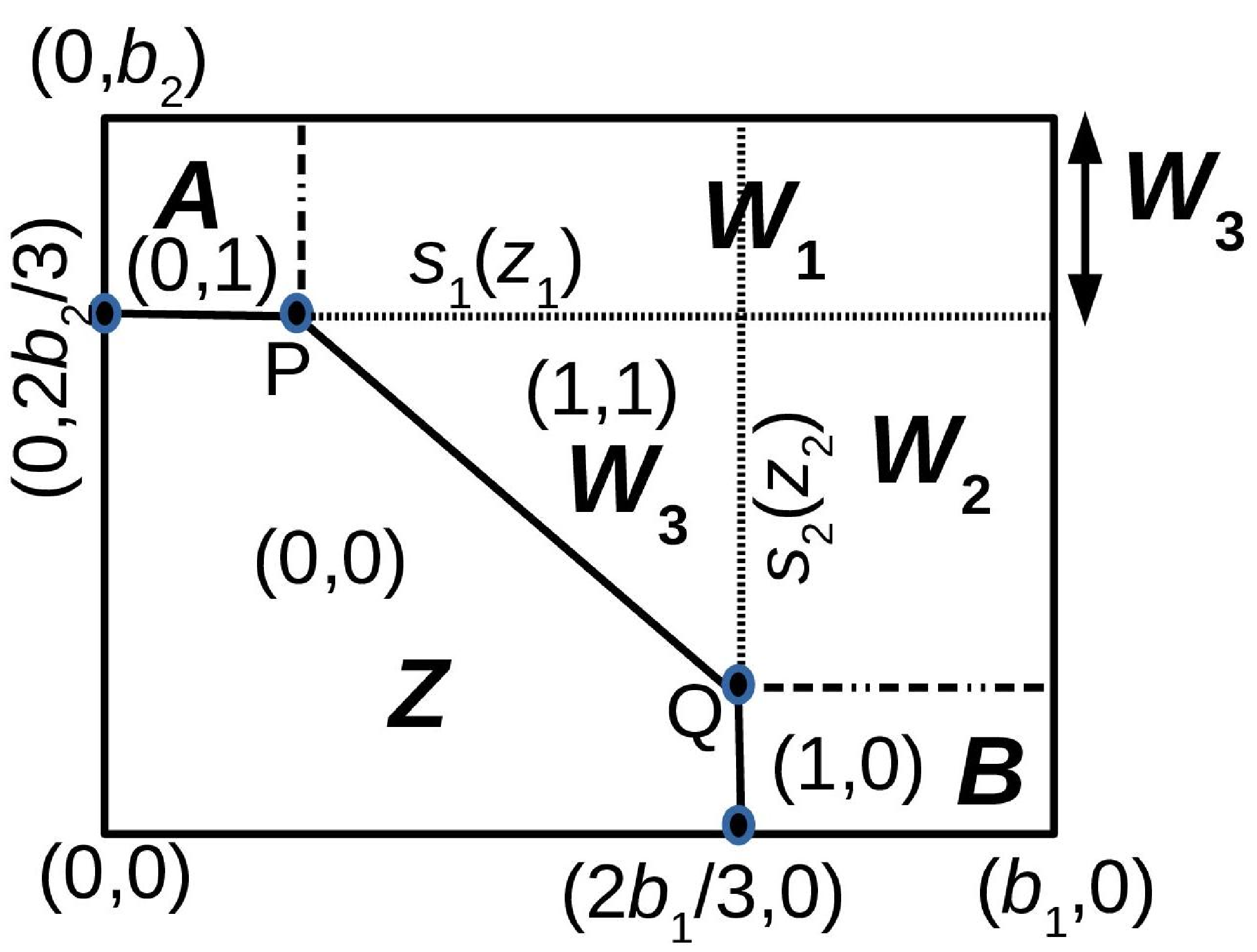}}&
\subfloat[]{\label{fig:figb}\includegraphics[scale=.15]{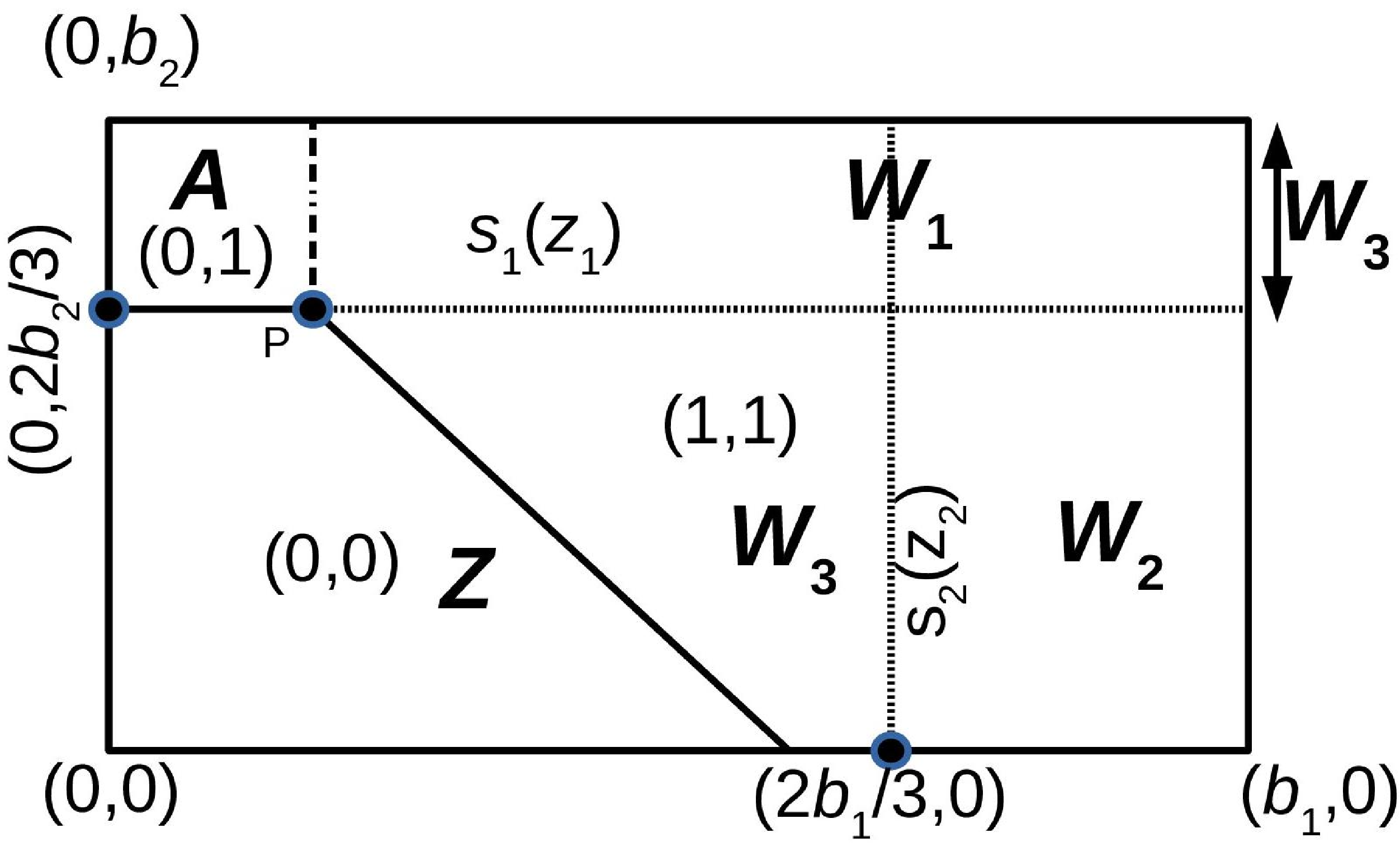}}
\end{tabular}
\caption{The optimal mechanism for the case $z\sim\mbox{Unif}[0,b_1]\times[0,b_2]$, when (a) $1\leq b_1/b_2\leq 2$ and (b) $b_1/b_2\geq 2$.}\label{fig:zero-solution}
\end{figure}

\begin{proof}
To prove this theorem, we must find a feasible $u$ and a feasible $\gamma$, and show that they satisfy the constraints of Lemma \ref{lem:conditions}. We define $u$ through the allocation function $q$ defined in (\ref{eqn:q-full}) and the relation $\nabla u=q$.

We now proceed to find the dual variable $\gamma$. We first partition $W=D\backslash(Z\cup A\cup B)$ into three regions:
\begin{multline*}
 W_1:[P_1,b_1)\times[s_1(z_1),b_2], \quad W_2:[s_2(z_2),b_1]\times[Q_2,s_1(z_1)], \\ W_3:\mbox{ the remaining region}
\end{multline*}
when the mechanism is as in Figure \ref{fig:figa}, and the same three regions but replacing $Q_2$ by $0$, when the mechanism is as in Figure \ref{fig:figb}. Observe that $W_3$ contains a portion of the right boundary adjacent to $W_1$. We define the set $\partial D^+$ to be the set of points in top and right boundaries of $D$, i.e., $\partial D^+:=\{z:(z_1,z_2)\in D:z_1=b_1\mbox{ or }z_2=b_2\}$.

We now set up the dual variable $\gamma$ measure as follows. First, let $\gamma_1:=\gamma_1^Z+\gamma_1^{D\backslash Z}$, with $\gamma_1^Z=\bar{\mu}^Z$, the $\bar{\mu}$ measure restricted to $Z$, and $\gamma_1^{D\backslash Z}=\bar{\mu}_+^{D\backslash Z}$. So $\gamma_1$ is supported on $(\partial D^+\cup Z)$. Next we specify the transition kernel $\gamma(\cdot~|~z)$ for $z\in(\partial D^+\cup Z)$.
\begin{enumerate}
 \item[(a)] For $z\in Z$, we define $\gamma(\cdot~|~z)=\delta_z(\cdot)$. We interpret this as the mass being retained at each $z\in Z$.
 \item[(b)] For $z\in(A\cup W_1)\cap\partial D^+$, $\gamma(\cdot~|~z)$ is defined by the uniform probability density on the line $\{z_1\}\times[s_1(z_1),b_2)$, and zero elsewhere. We interpret this as a transfer of $\mu_s(z)$ from the boundary to the above line.
 \item[(c)] For $z\in(B\cup W_2)\cap\partial D^+$, $\gamma(\cdot~|~z)$ is defined by the uniform probability density on the line $[s_2(z_2),b_1)\times\{z_2\}$, and zero elsewhere. Again, we interpret this as a transfer of $\mu_s(z)$ from the boundary to the above line.
 \item[(d)] For $z\in W_3\cap\partial D^+$, $\gamma(\cdot~|~z)$ is defined as follows. The total mass $\mu_s(W_3\cap\partial D^+)$ is spread uniformly on $W_3\backslash\partial D^+$ with equal contribution from each $z$ in $W_3\cap\partial D^+$. (Note that $W_3\cap\partial D^+$ is non-empty; see the right boundary in Figures \ref{fig:figa} and \ref{fig:figb}.)
\end{enumerate}

We then define $\gamma(F)=\int_{(z,z')\in F}\gamma_1(dz)\gamma(dz'~|~z)$ for any measurable $F\in D\times D$. It is then easy to check, by virtue of the choices of $s_1(z_1)$, $s_2(z_2)$, $Z$, and the matchings in (a)--(d), that $\gamma_2^Z=\gamma(Z,\cdot)=\bar{\mu}^Z$, and $\gamma_2^{D\backslash Z}=\gamma(D\backslash Z,\cdot)=\bar{\mu}_-^{D\backslash Z}$. Thus $(\gamma_1-\gamma_2)^Z=0$ and $(\gamma_1-\gamma_2)^{D\backslash Z}=\bar{\mu}^{D\backslash Z}$.

We now verify if $\gamma$ is feasible. To prove $\gamma_1-\gamma_2\succeq_{cvx}\bar{\mu}$, it suffices to show that (i) $(\gamma_1-\gamma_2)^Z\succeq_{cvx}\bar{\mu}^Z$, and (ii) $(\gamma_1-\gamma_2)^{D\backslash Z}\succeq_{cvx}\bar{\mu}^{D\backslash Z}$. Part (ii) holds trivially. We now show that $0\succeq_{cvx}\bar{\mu}^Z$. Observe that the components of $\bar{\mu}^Z$ are positive only at $(c_1,c_2)$, the left-bottom corner, and are negative elsewhere. Further, we have $\bar{\mu}_+(Z)=\bar{\mu}_-(Z)=1$, since the exclusion set $Z$ was constructed to have $\bar{\mu}(Z)=0$. So for any increasing function $f$, we have $\int_Zf(z)\,d\bar{\mu}^Z\leq 0$. This proves that $0\succeq_{cvx}\bar{\mu}^Z$, and also that $\gamma_1-\gamma_2\succeq_{cvx}\bar{\mu}$.

We now verify condition (i) in Lemma \ref{lem:conditions}.
\begin{multline*}
  \int_D u\,d(\gamma_1-\gamma_2)=\int_Z u\,d(\gamma_1-\gamma_2)^Z+\int_{D\backslash Z} u\,d(\gamma_1-\gamma_2)^{D\backslash Z}\\=\int_{D\backslash Z} u\,d\bar{\mu}=\int_D u\,d\bar{\mu}
\end{multline*}
where the second and the last equalities follow because $u(z)=0$ when $z\in Z$.

We now verify condition (ii) in Lemma \ref{lem:conditions}. To see why $u(z)-u(z')=\|z-z'\|_1$ holds $\gamma$-a.e., it suffices to check this for those $z$ on $(\partial D^+\cup Z)$ and the corresponding $z'$ for which $\gamma(\cdot~|~z)$ is nonzero, as in the four cases (a)--(d) above. For $z,z'$ in (a), $u(z)-u(z')=0$. In (b), $u(z)-u(z')=z_2-z_2'=\|z-z'\|_1$. In (c), $u(z)-u(z')=z_1-z_1'=\|z-z'\|_1$. In (d), $u(z)-u(z')=(z_1-z_1')+(z_2-z_2')=\|z-z'\|_1$. Thus $u(z)-u(z')=\|z-z'\|_1$ holds $\gamma$-a.e.
\end{proof}

In this example, the dual variable $\gamma$ was constructed so that $(\gamma_1 - \gamma_2)^{D\backslash Z} = \bar{\mu}^{D\backslash Z}$. In the more general case to be considered in the rest of the paper, we must shuffle the mass $\bar{\mu}+\bar{\alpha}$ (in place of $\bar{\mu}$), for some $\bar{\alpha}$ that convex-dominates the zero measure. So the dual variable $\gamma$ will be such that $(\gamma_1 - \gamma_2)^{D\backslash Z} = \bar{\mu}^{D\backslash Z} + \bar{\alpha}$.

\section{A Discussion of Optimal Solutions for the General Case}\label{sec:sol-space}
We now discuss the space of solutions for any nonnegative $c_1,c_2,b_1,b_2$. The main result of this paper is that the optimal mechanism is given as follows. For a summarizing phase diagram, see Figure \ref{fig:partition_uni}. To see a portrayal of all eight structures, see Figure \ref{fig:gen-structure}.

\subsection{Optimal solutions}
We describe the optimal solutions via decision trees. In Section \ref{sub:discussion}, we provide a discussion on the solutions. The formal statements are in Section \ref{sec:gencase}.

\subsubsection{$\frac{c_1}{b_1}$ and $\frac{c_2}{b_2}$ small}
When the values of $\frac{c_1}{b_1}$ and $\frac{c_2}{b_2}$ are small, in the sense that
\begin{multline}\label{eqn:c1-c2-small}
  \left\{c_1\leq b_1, c_2\leq 2b_2(b_1+c_1)/(b_1+3c_1)\right\}\\\cup\left\{c_2\leq b_2, c_1\leq 2b_1(b_2+c_2)/(b_2+3c_2)\right\},
\end{multline}
the optimal mechanism has one of the four structures in Figure \ref{fig:gen-structure-1}. Defining
\begin{align*}
p_{a_i}^*&:=\frac{2b_{-i}-c_{-i}}{3}-\frac{4c_i}{9}+\frac{2}{9}\sqrt{2c_i(2c_i+3(b_{-i}+c_{-i}))},\\
r_i&:=\frac{2b_{-i}(2b_i+5c_i)-c_{-i}(2b_i-3c_i)}{3(2b_i+3c_i)},
\end{align*}
the exact structure is given as follows.
\begin{enumerate}
 \item[(a)] The optimal mechanism is as in Figure \ref{fig:a}, if there exists $p_{a_i}\in[r_i,p_{a_i}^*]$, $i=1,2$, solving the following two equations simultaneously: 
\begin{multline}\label{eqn:fig5a-initial-1}
  (c_1-c_2-2(b_1-b_2)-(p_{a_1}-p_{a_2}))(c_1-2b_1+3p_{a_2})(c_2-2b_2+3p_{a_1})  \\
  +8c_1(b_2-p_{a_1})(c_1-2b_1+3p_{a_2})-8c_2(b_1-p_{a_2})(c_2-2b_2+3p_{a_1})     
  =  0,
\end{multline}
\begin{multline}\label{eqn:fig5a-initial-2}
 3\,\prod_{i=1}^2(b_i(c_{-i}-2b_{-i}+3p_{a_i})-4c_i(b_{-i}-p_{a_i}))\\-\sum_{i=1}^2((b_i(c_{-i}-2b_{-i}+3p_{a_i})-4c_i(b_{-i}-p_{a_i}))(c_i-2b_i+3p_{a_{-i}})(c_{-i}+b_{-i}))\\-\frac{3}{8}\,\prod_{i=1}^2((2b_{-i}-c_{-i}-p_{a_i})(c_i-2b_i+3p_{a_{-i}})-4c_{-i}(b_i-p_{a_{-i}}))=0.
\end{multline}
 \item[(b)] The optimal mechanism is as in Figure \ref{fig:b}, if (i) the condition in (a) does not hold, and (ii) there exists $p_{a_1}\in[r_1,p_{a_1}^*]$ that solves the following equation:
  \begin{multline}\label{eqn:fig5b-initial}
   -8c_1b_2^2+(c_2b_1-b_2b_1-b_2c_1)(c_2-2b_2)\\+\left(c_2/2-b_1\right)(c_2-2b_2)^2-(3/8)(c_2-2b_2)^3\\+\left(c_1(4c_2-3b_2)+3b_1b_2+2(c_2-2c_1)(c_2-2b_2)-15/8(c_2-2b_2)^2\right)p_{a_1}\\+\left(3c_2/2-21/8(c_2-2b_2)\right)p_{a_1}^2-(9/8)p_{a_1}^3=0.
 \end{multline}
 \item[(c)] The optimal mechanism is as in Figure \ref{fig:f}, if (i) the conditions in (a)-(b) do not hold, and (ii) there exists $p_{a_2}\in[r_2,p_{a_2}^*]$ that solves the following equation:
  \begin{multline}\label{eqn:fig5c-initial}
   -8c_2b_1^2+(c_1b_2-b_2b_1-b_1c_2)(c_1-2b_1)\\+\left(c_1/2-b_2\right)(c_1-2b_1)^2-(3/8)(c_1-2b_1)^3\\+\left(c_2(4c_1-3b_1)+3b_1b_2+2(c_1-2c_2)(c_1-2b_1)-15/8(c_1-2b_1)^2\right)p_{a_2}\\+\left(3c_1/2-21/8(c_1-2b_1)\right)p_{a_2}^2-(9/8)p_{a_2}^3=0.
 \end{multline}
 \item[(d)] The optimal mechanism is as in Figure \ref{fig:c}, if the conditions in $(a)$--$(c)$ do not hold.
\end{enumerate}
The decision tree in Figure \ref{fig:dec-tree-1} summarizes the procedure to find the exact structure. We discuss this in detail in Section \ref{SUB:GC1}. The results are formally stated in Theorem \ref{thm:str-1} of that section.

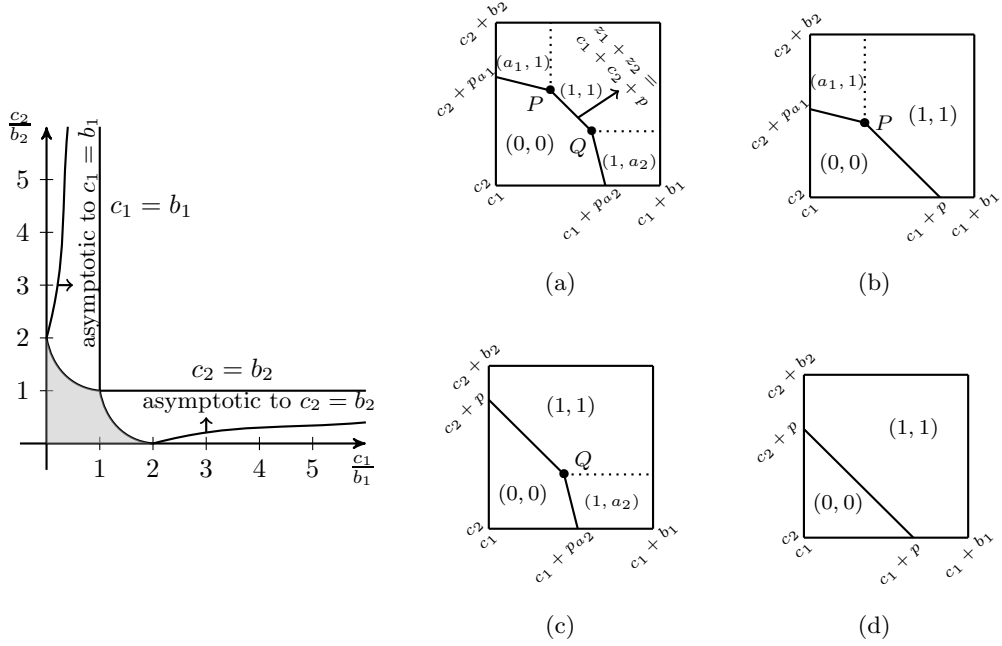
\begin{figure}[t!]
\centering
\begin{minipage}{.42\textwidth}
\centering
\begin{tikzpicture}[scale=0.35,font=\small,axis/.style={very thick, ->, >=stealth'}]
\draw [axis,thick,->] (0,-1)--(0,12);
\node [right] at (11,-1) {$\frac{c_1}{b_1}$};
\draw [axis,thick,->] (-1,0)--(12,0);
\node [above] at (-1,11) {$\frac{c_2}{b_2}$};
\draw [thin,-] (-0.25,2) -- (0.25,2);
\node [left] at (-0.25,2) {$1$};
\draw [thin,-] (-0.25,4) -- (0.25,4);
\node [left] at (-0.25,4) {$2$};
\draw [thin,-] (-0.25,6) -- (0.25,6);
\node [left] at (-0.25,6) {$3$};
\draw [thin,-] (-0.25,8) -- (0.25,8);
\node [left] at (-0.25,8) {$4$};
\draw [thin,-] (-0.25,10) -- (0.25,10);
\node [left] at (-0.25,10) {$5$};
\draw [thin,-] (2,-0.25) -- (2,0.25);
\node [below] at (2,-0.25) {$1$};
\draw [thin,-] (4,-0.25) -- (4,0.25);
\node [below] at (4,-0.25) {$2$};
\draw [thin,-] (6,-0.25) -- (6,0.25);
\node [below] at (6,-0.25) {$3$};
\draw [thin,-] (8,-0.25) -- (8,0.25);
\node [below] at (8,-0.25) {$4$};
\draw [thin,-] (10,-0.25) -- (10,0.25);
\node [below] at (10,-0.25) {$5$};
\draw [thick,-] (0,4) to[out=-80,in=175] (2,2);
\draw [thick,-] (4,0) to[out=175,in=-80] (2,2);
\draw [thick,-] (2,2)--(2,12);
\node [right] at (2,9) {$c_1=b_1$};
\draw [thick,-] (2,2)--(12,2);
\node [above] at (7,2) {$c_2=b_2$};
\draw [thick,-] (4,0) to[out=15,in=-175] (12,0.8);
\draw [thick,->] (0.4,6) to (1,6);
\node [right] at (0.8,8) {\rotatebox{90}{\footnotesize asymptotic to $c_1=b_1$}};
\draw [thick,-] (0,4) to[out=75,in=-95] (0.8,12);
\draw [thick,->] (6,0.4) to (6,1);
\node [above] at (8,0.8) {\footnotesize asymptotic to $c_2=b_2$};

\path[fill=gray!50,opacity=.5] (0,4) to[out=-80,in=175] (2,2) to (2,2) to[out=-80,in=175] (4,0) to (0,0) -- (0,4) ;
\end{tikzpicture}
\end{minipage}
\begin{minipage}{.54\textwidth}
\centering
\begin{tabular}{cc}
\subfloat[]{\label{fig:a}\begin{tikzpicture}[scale=0.18,font=\scriptsize,axis/.style={very thick, -}]
\node [rotate=45] at (0,-1) {\tiny$c_1$};
\node [rotate=45] at (-1,0) {\tiny$c_2$};
\draw [axis,thick,-] (0,0)--(12,0);
\node [rotate=45] at (12,-1.2) {\tiny$c_1+b_1$};
\draw [axis,thick,-] (0,0)--(0,12);
\node [rotate=45] at (-1,12) {\tiny$c_2+b_2$};
\draw [axis,thick,-] (0,12)--(12,12);
\draw [axis,thick,-] (12,0)--(12,12);
\draw [axis,thick,-] (0,8)--(4,7);
\node [rotate=45] at (-2.2,7) {\tiny$c_2+p_{a_1}$};
\draw [axis,thick,-] (7,4)--(4,7);
\draw [axis,thick,-] (7,4)--(8,0);
\node [rotate=45] at (7,-2) {\tiny$c_1+p_{a_2}$};
\foreach \Point in {(4,7),(7,4)}{
 \node at \Point {\textbullet};}
\draw [axis,thick,dotted] (4,7)--(4,12);
\draw [axis,thick,dotted] (7,4)--(12,4);
\node [left] at (4.25,6) {$P$};
\node [below] at (6,4.25) {$Q$};
\draw [axis,thick,->] (6,5)--(9,7);
\node [rotate=-45] at (9.5,9.5) {\tiny$z_1+z_2=$};
\node [rotate=-45] at (8.5,8.5) {\tiny$c_1+c_2+p$};
\node at (2.5,3) {$(0,0)$};
\node at (9.75,1.75) {\tiny$(1,a_2)$};
\node at (2,9) {\tiny$(a_1,1)$};
\node at (6.25,7) {\tiny$(1,1)$};
\end{tikzpicture}}&
\subfloat[]{\label{fig:b}\begin{tikzpicture}[scale=0.18,font=\scriptsize,axis/.style={very thick, -}]
\node [rotate=45] at (0,-1) {\tiny$c_1$};
\node [rotate=45] at (-1,0) {\tiny$c_2$};
\draw [axis,thick,-] (0,0)--(12,0);
\node [rotate=45] at (12,-1.2) {\tiny$c_1+b_1$};
\draw [axis,thick,-] (0,0)--(0,12);
\node [rotate=45] at (-1,12) {\tiny$c_2+b_2$};
\draw [axis,thick,-] (0,12)--(12,12);
\draw [axis,thick,-] (12,0)--(12,12);
\draw [axis,thick,-] (0,6.5)--(4,5.5);
\node [rotate=45] at (-2.2,5.5) {\tiny$c_2+p_{a_1}$};
\draw [axis,thick,-] (9.5,0)--(4,5.5);
\node [rotate=45] at (8.5,-1.8) {\tiny$c_1+p$};
\draw [axis,thick,dotted] (4,5.5)--(4,12);
\foreach \Point in {(4,5.5),(4,5.5)}{
 \node at \Point {\textbullet};}
\node [right] at (4,5.5) {$P$};
\node at (2.5,2.5) {$(0,0)$};
\node at (2,8.5) {\tiny$(a_1,1)$};
\node at (9,6) {$(1,1)$};
\end{tikzpicture}}\\
\subfloat[]{\label{fig:f}\begin{tikzpicture}[scale=0.18,font=\scriptsize,axis/.style={very thick, -}]
\node [rotate=45] at (0,-1) {\tiny$c_1$};
\node [rotate=45] at (-1,0) {\tiny$c_2$};
\draw [axis,thick,-] (0,0)--(12,0);
\node [rotate=45] at (12,-1.2) {\tiny$c_1+b_1$};
\draw [axis,thick,-] (0,0)--(0,12);
\node [rotate=45] at (-1,12) {\tiny$c_2+b_2$};
\draw [axis,thick,-] (0,12)--(12,12);
\draw [axis,thick,-] (12,0)--(12,12);
\draw [axis,thick,-] (0,9.5)--(5.5,4);
\node [rotate=45] at (-1.8,8.5) {\tiny$c_2+p$};
\draw [axis,thick,-] (5.5,4)--(6.5,0);
\node [rotate=45] at (5.5,-2) {\tiny$c_1+p_{a_2}$};
\draw [axis,thick,dotted] (5.5,4)--(12,4);
\foreach \Point in {(5.5,4),(5.5,4)}{
 \node at \Point {\textbullet};}
\node [right] at (5.5,5) {$Q$};
\node at (2.5,2.5) {$(0,0)$};
\node at (9,2) {\tiny$(1,a_2)$};
\node at (6,9) {$(1,1)$};
\end{tikzpicture}}&
\subfloat[]{\label{fig:c}\begin{tikzpicture}[scale=0.18,font=\scriptsize,axis/.style={very thick, -}]
\node [rotate=45] at (0,-1) {\tiny$c_1$};
\node [rotate=45] at (-1,0) {\tiny$c_2$};
\draw [axis,thick,-] (0,0)--(12,0);
\node [rotate=45] at (12,-1.2) {\tiny$c_1+b_1$};
\draw [axis,thick,-] (0,0)--(0,12);
\node [rotate=45] at (-1,12) {\tiny$c_2+b_2$};
\draw [axis,thick,-] (0,12)--(12,12);
\node [rotate=45] at (-2,7) {\tiny$c_2+p$};
\draw [axis,thick,-] (12,0)--(12,12);
\node [rotate=45] at (7,-2) {\tiny$c_1+p$};
\draw [axis,thick,-] (0,8)--(8,0);
\node at (2.5,2.5) {$(0,0)$};
\node at (8,8) {$(1,1)$};
\end{tikzpicture}}
\end{tabular}
\end{minipage}
\caption{When $c_1,c_2,b_1, b_2$ fall in the shaded region on the left side, the optimal mechanism has one of the four structures described on the right side. $q_1, q_2$ in each structure indicate the corresponding allocation probabilities.}\label{fig:gen-structure-1}
\end{figure}
\begin{figure}[t!]
\centering
\begin{tikzpicture}[scale=0.42,font=\normalsize,axis/.style={very thick, ->, >=stealth'}]
\draw [axis,thick,-] (1.25,0)--(17.25,0);
\draw [axis,thick,-] (1.25,0)--(1.25,-2);
\draw [axis,thick,-] (1.25,-2)--(17.25,-2);
\draw [axis,thick,-] (17.25,-2)--(17.25,0);
\node at (9.25,-1) {The case when $(c_1,c_2,b_1,b_2)$ satisfy (\ref{eqn:c1-c2-small})};
\draw [axis,thick,->] (5.75,-2)--(5.75,-3.5);
\draw [axis,thick,-] (1.5,-3.5)--(10,-3.5);
\draw [axis,thick,-] (1.5,-3.5)--(1.5,-6.5);
\draw [axis,thick,-] (1.5,-6.5)--(10,-6.5);
\draw [axis,thick,-] (10,-6.5)--(10,-3.5);
\node at (5.75,-4.5) {Does $\exists p_{a_i}\in[r_i,p_{a_i}^*]$};
\node at (5.75,-5.5) {solving (\ref{eqn:fig5a-initial-1}) and (\ref{eqn:fig5a-initial-2})?};
\draw [axis,thick,->] (10,-5)--(12.5,-5);
\node at (11.25,-4) {\bf Yes};
\draw [axis,thick,-] (12.5,-3.5)--(17,-3.5);
\draw [axis,thick,-] (12.5,-3.5)--(12.5,-6.5);
\draw [axis,thick,-] (12.5,-6.5)--(17,-6.5);
\draw [axis,thick,-] (17,-6.5)--(17,-3.5);
\node at (14.75,-4.5) {Fig. \ref{fig:a} is};
\node at (14.75,-5.5) {optimal};
\draw [axis,thick,->] (5.75,-6.5)--(5.75,-8);
\node at (7,-7.25) {\bf No};
\draw [axis,thick,-] (1.5,-8)--(10,-8);
\draw [axis,thick,-] (1.5,-8)--(1.5,-11);
\draw [axis,thick,-] (1.5,-11)--(10,-11);
\draw [axis,thick,-] (10,-11)--(10,-8);
\node at (5.75,-9) {Does $\exists p_{a_1}\in[r_1,p_{a_1}^*]$};
\node at (5.75,-10) {solving (\ref{eqn:fig5b-initial})?};
\draw [axis,thick,->] (10,-9.5)--(12.5,-9.5);
\node at (11.25,-8.5) {\bf Yes};
\draw [axis,thick,-] (12.5,-8)--(17,-8);
\draw [axis,thick,-] (12.5,-8)--(12.5,-11);
\draw [axis,thick,-] (12.5,-11)--(17,-11);
\draw [axis,thick,-] (17,-11)--(17,-8);
\node at (14.75,-9) {Fig. \ref{fig:b} is};
\node at (14.75,-10) {optimal};
\draw [axis,thick,->] (5.75,-11)--(5.75,-12.5);
\node at (7,-11.75) {\bf No};
\draw [axis,thick,-] (1.5,-12.5)--(10,-12.5);
\draw [axis,thick,-] (1.5,-12.5)--(1.5,-15.5);
\draw [axis,thick,-] (1.5,-15.5)--(10,-15.5);
\draw [axis,thick,-] (10,-15.5)--(10,-12.5);
\node at (5.75,-13.5) {Does $\exists p_{a_2}\in[r_2,p_{a_2}^*]$};
\node at (5.75,-14.5) {solving (\ref{eqn:fig5c-initial})?};
\draw [axis,thick,->] (10,-14)--(12.5,-14);
\node at (11.25,-13) {\bf Yes};
\draw [axis,thick,-] (12.5,-12.5)--(17,-12.5);
\draw [axis,thick,-] (12.5,-12.5)--(12.5,-15.5);
\draw [axis,thick,-] (12.5,-15.5)--(17,-15.5);
\draw [axis,thick,-] (17,-15.5)--(17,-12.5);
\node at (14.75,-13.5) {Fig. \ref{fig:f} is};
\node at (14.75,-14.5) {optimal};
\draw [axis,thick,->] (5.75,-15.5)--(5.75,-17);
\node at (7,-16.25) {\bf No};
\draw [axis,thick,-] (4,-17)--(8,-17);
\draw [axis,thick,-] (4,-17)--(4,-20);
\draw [axis,thick,-] (4,-20)--(8,-20);
\draw [axis,thick,-] (8,-20)--(8,-17);
\node at (6,-18) {Fig. \ref{fig:c} is};
\node at (6,-19) {optimal};
\end{tikzpicture}
\caption{Decision tree illustrating the structure of optimal mechanism when $c_1,c_2,b_1,b_2$ fall in the shaded region in Figure \ref{fig:gen-structure-1}, i.e., when $\frac{c_1}{b_1}$ and $\frac{c_2}{b_2}$ are small.}\label{fig:dec-tree-1}
\end{figure}
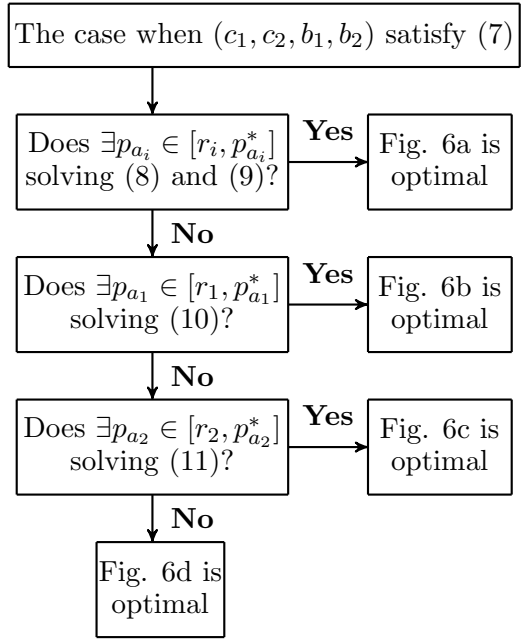

\subsubsection{$\frac{c_1}{b_1}$ is small but $\frac{c_2}{b_2}$ is large}
When the value of $\frac{c_2}{b_2}$ is large but $\frac{c_1}{b_1}$ is small, in the sense that
\begin{equation}\label{eqn:c1-small-c2-large}
  \left\{c_1\leq b_1, c_2\in[2b_2(b_1+c_1)/(b_1+3c_1),2b_2(b_1/(b_1-c_1))^2]\right\},
\end{equation}
the optimal mechanism has one of the two structures depicted in Figure \ref{fig:gen-structure-2}. The exact structures are given as follows.
\begin{enumerate}
\item[(a)] The optimal mechanism is as depicted in Figure \ref{fig:d}, if there exists $(p_{a_1}\in[(2b_2-c_2)_+,b_2],\,a_1\leq1)$ that solves the following two equations simultaneously:
\begin{equation}\label{eqn:fig6-initial}
  a_1=p_{a_1}\left(\frac{\frac{3}{2}p_{a_1}+c_2}{b_1b_2-c_1p_{a_1}}\right);\quad a_1=\frac{p_{a_1}}{b_1-c_1}\sqrt{\frac{2(p_{a_1}+c_2)}{b_2}}.
\end{equation}
\item[(b)] The optimal mechanism is as depicted in Figure \ref{fig:c'}, if conditions in (a) fail.
\end{enumerate}
The decision tree in Figure \ref{fig:dec-tree-2} summarizes the procedure to find the exact structure. We discuss this in detail in Section \ref{SUB:GC2}. The results are formally stated in Theorem \ref{thm:gc4} of that section.

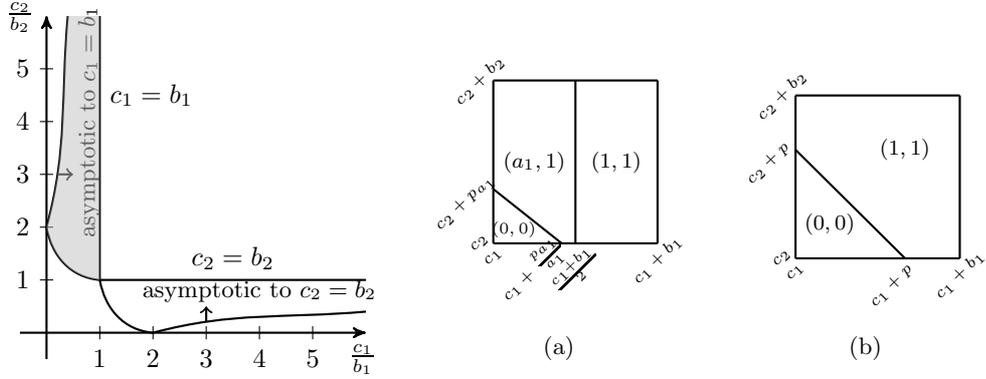
\begin{figure}[h!]
\centering
\begin{minipage}{.42\textwidth}
\centering
\begin{tikzpicture}[scale=0.35,font=\small,axis/.style={very thick, ->, >=stealth'}]
\draw [axis,thick,->] (0,-1)--(0,12);
\node [right] at (11,-1) {$\frac{c_1}{b_1}$};
\draw [axis,thick,->] (-1,0)--(12,0);
\node [above] at (-1,11) {$\frac{c_2}{b_2}$};
\draw [thin,-] (-0.25,2) -- (0.25,2);
\node [left] at (-0.25,2) {$1$};
\draw [thin,-] (-0.25,4) -- (0.25,4);
\node [left] at (-0.25,4) {$2$};
\draw [thin,-] (-0.25,6) -- (0.25,6);
\node [left] at (-0.25,6) {$3$};
\draw [thin,-] (-0.25,8) -- (0.25,8);
\node [left] at (-0.25,8) {$4$};
\draw [thin,-] (-0.25,10) -- (0.25,10);
\node [left] at (-0.25,10) {$5$};
\draw [thin,-] (2,-0.25) -- (2,0.25);
\node [below] at (2,-0.25) {$1$};
\draw [thin,-] (4,-0.25) -- (4,0.25);
\node [below] at (4,-0.25) {$2$};
\draw [thin,-] (6,-0.25) -- (6,0.25);
\node [below] at (6,-0.25) {$3$};
\draw [thin,-] (8,-0.25) -- (8,0.25);
\node [below] at (8,-0.25) {$4$};
\draw [thin,-] (10,-0.25) -- (10,0.25);
\node [below] at (10,-0.25) {$5$};
\draw [thick,-] (0,4) to[out=-80,in=175] (2,2);
\draw [thick,-] (4,0) to[out=175,in=-80] (2,2);
\draw [thick,-] (2,2)--(2,12);
\node [right] at (2,9) {$c_1=b_1$};
\draw [thick,-] (2,2)--(12,2);
\node [above] at (7,2) {$c_2=b_2$};
\draw [thick,-] (4,0) to[out=15,in=-175] (12,0.8);
\draw [thick,->] (0.4,6) to (1,6);
\node [right] at (0.8,8) {\rotatebox{90}{\footnotesize asymptotic to $c_1=b_1$}};
\draw [thick,-] (0,4) to[out=75,in=-95] (0.8,12);
\draw [thick,->] (6,0.4) to (6,1);
\node [above] at (8,0.8) {\footnotesize asymptotic to $c_2=b_2$};

\path[fill=gray!50,opacity=.5] (0,4) to[out=-80,in=175] (2,2) to (2,2) to (2,12) to (0.8,12) to[out=-95,in=75] (0,4);
\end{tikzpicture}
\end{minipage}
\begin{minipage}{.54\textwidth}
\centering
\begin{tabular}{cc}
\subfloat[]{\label{fig:d}\begin{tikzpicture}[scale=0.18,font=\scriptsize,axis/.style={very thick, -}]
\node [rotate=45] at (0,-1) {\tiny$c_1$};
\node [rotate=45] at (-1,0) {\tiny$c_2$};
\draw [axis,thick,-] (0,0)--(12,0);
\node [rotate=45] at (12,-1.2) {\tiny$c_1+b_1$};
\draw [axis,thick,-] (0,0)--(0,12);
\node [rotate=45] at (-1,12) {\tiny$c_2+b_2$};
\draw [axis,thick,-] (0,12)--(12,12);
\node [rotate=45] at (-2,3) {\tiny$c_2+p_{a_1}$};
\draw [axis,thick,-] (12,0)--(12,12);
\node [rotate=45] at (3,-2) {\tiny$c_1+\frac{p_{a_1}}{a_1}$};
\draw [axis,thick,-] (0,4)--(5,0);
\draw [axis,thick,-] (6,0)--(6,12);
\node [rotate=45] at (6,-2) {\tiny$\frac{c_1+b_1}{2}$};
\node at (1.5,1) {\tiny$(0,0)$};
\node at (3,6) {$(a_1,1)$};
\node at (9,6) {$(1,1)$};
\end{tikzpicture}}&
\subfloat[]{\label{fig:c'}\begin{tikzpicture}[scale=0.18,font=\scriptsize,axis/.style={very thick, -}]
\node [rotate=45] at (0,-1) {\tiny$c_1$};
\node [rotate=45] at (-1,0) {\tiny$c_2$};
\draw [axis,thick,-] (0,0)--(12,0);
\node [rotate=45] at (12,-1.2) {\tiny$c_1+b_1$};
\draw [axis,thick,-] (0,0)--(0,12);
\node [rotate=45] at (-1,12) {\tiny$c_2+b_2$};
\draw [axis,thick,-] (0,12)--(12,12);
\node [rotate=45] at (-2,7) {\tiny$c_2+p$};
\draw [axis,thick,-] (12,0)--(12,12);
\node [rotate=45] at (7,-2) {\tiny$c_1+p$};
\draw [axis,thick,-] (0,8)--(8,0);
\node at (2.5,2.5) {$(0,0)$};
\node at (8,8) {$(1,1)$};
\end{tikzpicture}}
\end{tabular}
\end{minipage}
\caption{When $c_1,c_2,b_1,b_2$ fall in the shaded region on the left side, the optimal mechanism has one of the two structures described on the right side.}\label{fig:gen-structure-2}
\end{figure}

\begin{figure}[h!]
\centering
\begin{tikzpicture}[scale=0.4,font=\normalsize,axis/.style={very thick, ->, >=stealth'}]
\draw [axis,thick,-] (-0.5,0)--(16.5,0);
\draw [axis,thick,-] (-0.5,0)--(-0.5,-2);
\draw [axis,thick,-] (-0.5,-2)--(16.5,-2);
\draw [axis,thick,-] (16.5,-2)--(16.5,0);
\node at (8,-1) {The case when $(c_1,c_2,b_1,b_2)$ satisfy (\ref{eqn:c1-small-c2-large})};
\draw [axis,thick,->] (8,-2)--(8,-3.5);
\draw [axis,thick,-] (2,-3.5)--(14,-3.5);
\draw [axis,thick,-] (2,-3.5)--(2,-6.5);
\draw [axis,thick,-] (2,-6.5)--(14,-6.5);
\draw [axis,thick,-] (14,-6.5)--(14,-3.5);
\node at (8,-4.5) {Does $\exists p_{a_1}\in[(2b_2-c_2)_+,b_2]$};
\node at (8,-5.5) {and $a_1 \leq 1$ solving (\ref{eqn:fig6-initial})?};
\draw [axis,thick,->] (7,-6.5)--(5,-8);
\node at (4.5,-7.25) {\bf Yes};
\draw [axis,thick,-] (3,-8)--(7,-8);
\draw [axis,thick,-] (3,-8)--(3,-11);
\draw [axis,thick,-] (3,-11)--(7,-11);
\draw [axis,thick,-] (7,-11)--(7,-8);
\node at (5,-9) {Fig. \ref{fig:d} is};
\node at (5,-10) {optimal};
\draw [axis,thick,->] (9,-6.5)--(11,-8);
\node at (11.5,-7.25) {\bf No};
\draw [axis,thick,-] (9,-8)--(13,-8);
\draw [axis,thick,-] (9,-8)--(9,-11);
\draw [axis,thick,-] (9,-11)--(13,-11);
\draw [axis,thick,-] (13,-11)--(13,-8);
\node at (11,-9) {Fig. \ref{fig:c'} is};
\node at (11,-10) {optimal};
\end{tikzpicture}
\caption{Decision tree illustrating the structure of optimal mechanism when $c_1,c_2,b_1, b_2$ fall in the shaded region in Figure
\ref{fig:gen-structure-2}, i.e., when $\frac{c_1}{b_1}$ is small and $\frac{c_2}{b_2}$ is large.}\label{fig:dec-tree-2}
\end{figure}
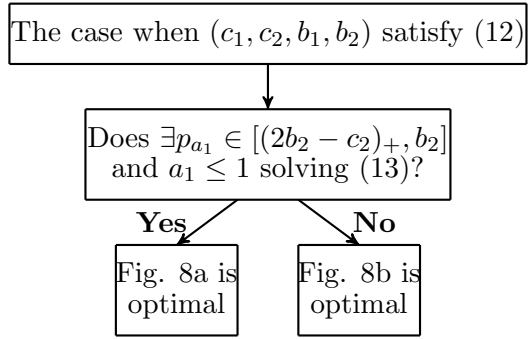

\subsubsection{$\frac{c_1}{b_1}$ is small but $\frac{c_2}{b_2}$ is very large}
When the value of $\frac{c_2}{b_2}$ is very large but $\frac{c_1}{b_1}$ is small, in the sense that
$$
  \left\{c_1\leq b_1, c_2\geq 2b_2(b_1/(b_1-c_1))^2\right\},
$$
the optimal mechanism is as in Figure \ref{fig:e}. We discuss this in detail in Section \ref{SUB:GC2}. The result is formally stated in Theorem \ref{thm:gc5} of that section.

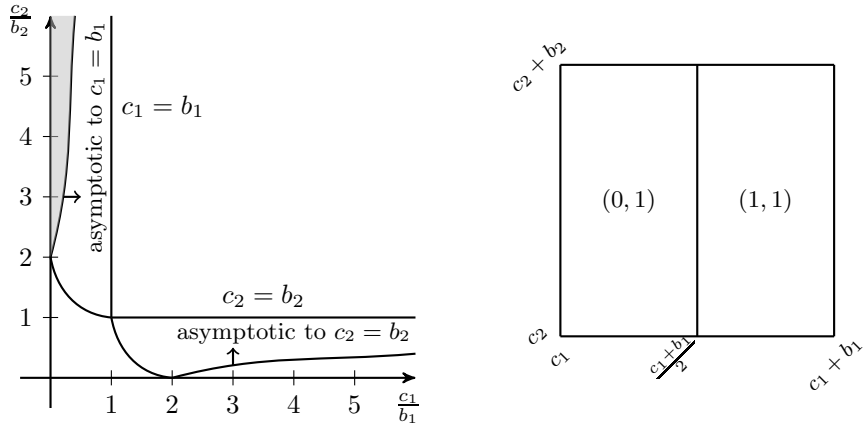
\begin{figure}[h!]
\centering
\begin{minipage}{.48\textwidth}
\centering
\begin{tikzpicture}[scale=0.4,font=\small,axis/.style={very thick, ->, >=stealth'}]
\draw [axis,thick,->] (0,-1)--(0,12);
\node [right] at (11,-1) {$\frac{c_1}{b_1}$};
\draw [axis,thick,->] (-1,0)--(12,0);
\node [above] at (-1,11) {$\frac{c_2}{b_2}$};
\draw [thin,-] (-0.25,2) -- (0.25,2);
\node [left] at (-0.25,2) {$1$};
\draw [thin,-] (-0.25,4) -- (0.25,4);
\node [left] at (-0.25,4) {$2$};
\draw [thin,-] (-0.25,6) -- (0.25,6);
\node [left] at (-0.25,6) {$3$};
\draw [thin,-] (-0.25,8) -- (0.25,8);
\node [left] at (-0.25,8) {$4$};
\draw [thin,-] (-0.25,10) -- (0.25,10);
\node [left] at (-0.25,10) {$5$};
\draw [thin,-] (2,-0.25) -- (2,0.25);
\node [below] at (2,-0.25) {$1$};
\draw [thin,-] (4,-0.25) -- (4,0.25);
\node [below] at (4,-0.25) {$2$};
\draw [thin,-] (6,-0.25) -- (6,0.25);
\node [below] at (6,-0.25) {$3$};
\draw [thin,-] (8,-0.25) -- (8,0.25);
\node [below] at (8,-0.25) {$4$};
\draw [thin,-] (10,-0.25) -- (10,0.25);
\node [below] at (10,-0.25) {$5$};
\draw [thick,-] (0,4) to[out=-80,in=175] (2,2);
\draw [thick,-] (4,0) to[out=175,in=-80] (2,2);
\draw [thick,-] (2,2)--(2,12);
\node [right] at (2,9) {$c_1=b_1$};
\draw [thick,-] (2,2)--(12,2);
\node [above] at (7,2) {$c_2=b_2$};
\draw [thick,-] (4,0) to[out=15,in=-175] (12,0.8);
\draw [thick,->] (0.4,6) to (1,6);
\node [right] at (0.8,8) {\rotatebox{90}{\footnotesize asymptotic to $c_1=b_1$}};
\draw [thick,-] (0,4) to[out=75,in=-95] (0.8,12);
\draw [thick,->] (6,0.4) to (6,1);
\node [above] at (8,0.8) {\footnotesize asymptotic to $c_2=b_2$};
\path[fill=gray!50,opacity=.5] (0,4) to[out=75,in=-95] (0.8,12)--(0,12)--(0,4);
\end{tikzpicture}
\end{minipage}
\begin{minipage}{.48\textwidth}
\centering
\begin{tikzpicture}[scale=0.3,font=\footnotesize,axis/.style={very thick, -}]
\node [rotate=45] at (0,-1) {\scriptsize$c_1$};
\node [rotate=45] at (-1,0) {\scriptsize$c_2$};
\draw [axis,thick,-] (0,0)--(12,0);
\node [rotate=45] at (12,-1.2) {\scriptsize$c_1+b_1$};
\draw [axis,thick,-] (0,0)--(0,12);
\node [rotate=45] at (-1,12) {\scriptsize$c_2+b_2$};
\draw [axis,thick,-] (0,12)--(12,12);
\draw [axis,thick,-] (12,0)--(12,12);
\node [rotate=45] at (5,-1) {\scriptsize$\frac{c_1+b_1}{2}$};
\draw [axis,thick,-] (6,0)--(6,12);
\node at (3,6) {$(0,1)$};
\node at (9,6) {$(1,1)$};
\end{tikzpicture}
\end{minipage}
\caption{When $c_1,c_2,b_1,b_2$ fall in the shaded region on the left side, the optimal mechanism has the structure described on the right side.}\label{fig:e}
\end{figure}

\subsubsection{$\frac{c_2}{b_2}$ is small but $\frac{c_1}{b_1}$ is large}
When the value of $\frac{c_1}{b_1}$ is large but $\frac{c_2}{b_2}$ is small, in the sense that
$$
  \left\{c_2\leq b_2, c_1\in[2b_1(b_2+c_2)/(b_2+3c_2),2b_1(b_2/(b_2-c_2))^2]\right\},
$$
the optimal mechanism has the structure as depicted either in Figure \ref{fig:g}, or in Figure \ref{fig:c''}. Notice that this is the symmetric counterpart of the case where $\frac{c_2}{b_2}$ is large but $\frac{c_1}{b_1}$ is small. The exact structures are symmetric counterparts of Figures \ref{fig:d} and \ref{fig:c'}. See Theorem \ref{thm:gc4-sym} of Section \ref{SUB:GC3}.

\subsubsection{$\frac{c_2}{b_2}$ is small but $\frac{c_1}{b_1}$ is very large}
When the value of $\frac{c_1}{b_1}$ is very large but $\frac{c_2}{b_2}$ is small, in the sense that
$$
  \left\{c_2\leq b_2, c_1\geq 2b_1(b_2/(b_2-c_2))^2\right\},
$$
the optimal mechanism is as in Figure \ref{fig:h}.  See Theorem \ref{thm:gc5-sym} of Section \ref{SUB:GC3}.

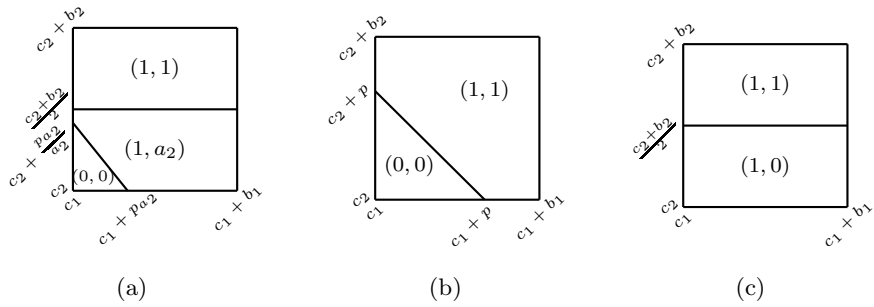
\begin{figure}[h!]
\centering
\begin{tabular}{ccc}
\subfloat[]{\label{fig:g}\begin{tikzpicture}[scale=0.18,font=\scriptsize,axis/.style={very thick, -}]
\node [rotate=45] at (0,-1) {\tiny$c_1$};
\node [rotate=45] at (-1,0) {\tiny$c_2$};
\draw [axis,thick,-] (0,0)--(12,0);
\node [rotate=45] at (12,-1.2) {\tiny$c_1+b_1$};
\draw [axis,thick,-] (0,0)--(0,12);
\node [rotate=45] at (-1,12) {\tiny$c_2+b_2$};
\draw [axis,thick,-] (0,12)--(12,12);
\node [rotate=45] at (-2.5,2.5) {\tiny$c_2+\frac{p_{a_2}}{a_2}$};
\node [rotate=45] at (-1.8,6) {\tiny$\frac{c_2+b_2}{2}$};
\draw [axis,thick,-] (12,0)--(12,12);
\node [rotate=45] at (4,-2) {\tiny$c_1+p_{a_2}$};
\draw [axis,thick,-] (0,5)--(4,0);
\draw [axis,thick,-] (0,6)--(12,6);
\node at (1.5,1) {\tiny$(0,0)$};
\node at (6,3) {$(1,a_2)$};
\node at (6,9) {$(1,1)$};
\end{tikzpicture}}&
\subfloat[]{\label{fig:c''}\begin{tikzpicture}[scale=0.18,font=\scriptsize,axis/.style={very thick, -}]
\node [rotate=45] at (0,-1) {\tiny$c_1$};
\node [rotate=45] at (-1,0) {\tiny$c_2$};
\draw [axis,thick,-] (0,0)--(12,0);
\node [rotate=45] at (12,-1.2) {\tiny$c_1+b_1$};
\draw [axis,thick,-] (0,0)--(0,12);
\node [rotate=45] at (-1,12) {\tiny$c_2+b_2$};
\draw [axis,thick,-] (0,12)--(12,12);
\node [rotate=45] at (-2,7) {\tiny$c_2+p$};
\draw [axis,thick,-] (12,0)--(12,12);
\node [rotate=45] at (7,-2) {\tiny$c_1+p$};
\draw [axis,thick,-] (0,8)--(8,0);
\node at (2.5,2.5) {$(0,0)$};
\node at (8,8) {$(1,1)$};
\end{tikzpicture}}&
\subfloat[]{\label{fig:h}\begin{tikzpicture}[scale=0.18,font=\scriptsize,axis/.style={very thick, -}]
\node [rotate=45] at (0,-1) {\tiny$c_1$};
\node [rotate=45] at (-1,0) {\tiny$c_2$};
\draw [axis,thick,-] (0,0)--(12,0);
\node [rotate=45] at (12,-1.2) {\tiny$c_1+b_1$};
\draw [axis,thick,-] (0,0)--(0,12);
\node [rotate=45] at (-1,12) {\tiny$c_2+b_2$};
\draw [axis,thick,-] (0,12)--(12,12);
\node [rotate=45] at (-2,5) {\tiny$\frac{c_2+b_2}{2}$};
\draw [axis,thick,-] (12,0)--(12,12);
\draw [axis,thick,-] (0,6)--(12,6);
\node at (6,3) {$(1,0)$};
\node at (6,9) {$(1,1)$};
\end{tikzpicture}}
\end{tabular}
\caption{The structures of the optimal mechanism when $\frac{c_2}{b_2}$ is small and  $\frac{c_1}{b_1}$ is large.}\label{fig:gen-structure-3}
\end{figure}

\subsubsection{$\frac{c_1}{b_1}$ and $\frac{c_2}{b_2}$ large}
In the remaining cases when $\frac{c_1}{b_1}$ and $\frac{c_2}{b_2}$ are large, in the sense that
$$
  \left\{c_1\geq b_1, c_2\geq b_2\right\},
$$
the optimal mechanism is given by pure bundling as in Figure \ref{fig:gen-structure-4}.

\begin{figure}[h!]
\centering
\begin{minipage}{.48\textwidth}
\centering
\begin{tikzpicture}[scale=0.4,font=\small,axis/.style={very thick, ->, >=stealth'}]
\draw [axis,thick,->] (0,-1)--(0,12);
\node [right] at (11,-1) {$\frac{c_1}{b_1}$};
\draw [axis,thick,->] (-1,0)--(12,0);
\node [above] at (-1,11) {$\frac{c_2}{b_2}$};
\draw [thin,-] (-0.25,2) -- (0.25,2);
\node [left] at (-0.25,2) {$1$};
\draw [thin,-] (-0.25,4) -- (0.25,4);
\node [left] at (-0.25,4) {$2$};
\draw [thin,-] (-0.25,6) -- (0.25,6);
\node [left] at (-0.25,6) {$3$};
\draw [thin,-] (-0.25,8) -- (0.25,8);
\node [left] at (-0.25,8) {$4$};
\draw [thin,-] (-0.25,10) -- (0.25,10);
\node [left] at (-0.25,10) {$5$};
\draw [thin,-] (2,-0.25) -- (2,0.25);
\node [below] at (2,-0.25) {$1$};
\draw [thin,-] (4,-0.25) -- (4,0.25);
\node [below] at (4,-0.25) {$2$};
\draw [thin,-] (6,-0.25) -- (6,0.25);
\node [below] at (6,-0.25) {$3$};
\draw [thin,-] (8,-0.25) -- (8,0.25);
\node [below] at (8,-0.25) {$4$};
\draw [thin,-] (10,-0.25) -- (10,0.25);
\node [below] at (10,-0.25) {$5$};
\draw [thick,-] (0,4) to[out=-80,in=175] (2,2);
\draw [thick,-] (4,0) to[out=175,in=-80] (2,2);
\draw [thick,-] (2,2)--(2,12);
\node [right] at (2,9) {$c_1=b_1$};
\draw [thick,-] (2,2)--(12,2);
\node [above] at (7,2) {$c_2=b_2$};
\draw [thick,-] (4,0) to[out=15,in=-175] (12,0.8);
\draw [thick,->] (0.4,6) to (1,6);
\node [right] at (0.8,8) {\rotatebox{90}{\footnotesize asymptotic to $c_1=b_1$}};
\draw [thick,-] (0,4) to[out=75,in=-95] (0.8,12);
\draw [thick,->] (6,0.4) to (6,1);
\node [above] at (8,0.8) {\footnotesize asymptotic to $c_2=b_2$};
\path[fill=gray!50,opacity=.5] (2,12)--(2,2)--(12,2)--(12,12)--(2,12);
\end{tikzpicture}
\end{minipage}
\begin{minipage}{.48\textwidth}
\centering
\begin{tikzpicture}[scale=0.3,font=\footnotesize,axis/.style={very thick, -}]
\node [rotate=45] at (0,-1) {\scriptsize$c_1$};
\node [rotate=45] at (-1,0) {\scriptsize$c_2$};
\draw [axis,thick,-] (0,0)--(12,0);
\node [rotate=45] at (12,-1.2) {\scriptsize$c_1+b_1$};
\draw [axis,thick,-] (0,0)--(0,12);
\node [rotate=45] at (-1,12) {\scriptsize$c_2+b_2$};
\draw [axis,thick,-] (0,12)--(12,12);
\node [rotate=45] at (-1.25,7) {\scriptsize$c_2+p$};
\draw [axis,thick,-] (12,0)--(12,12);
\node [rotate=45] at (7,-1.25) {\scriptsize$c_1+p$};
\draw [axis,thick,-] (0,8)--(8,0);
\node at (2.5,2.5) {$(0,0)$};
\node at (8,8) {$(1,1)$};
\end{tikzpicture}
\end{minipage}
\caption{When $c_1,c_2,b_1,b_2$ fall in the shaded region on the left side, the optimal mechanism has the structure described on the right side.}\label{fig:gen-structure-4}
\end{figure}
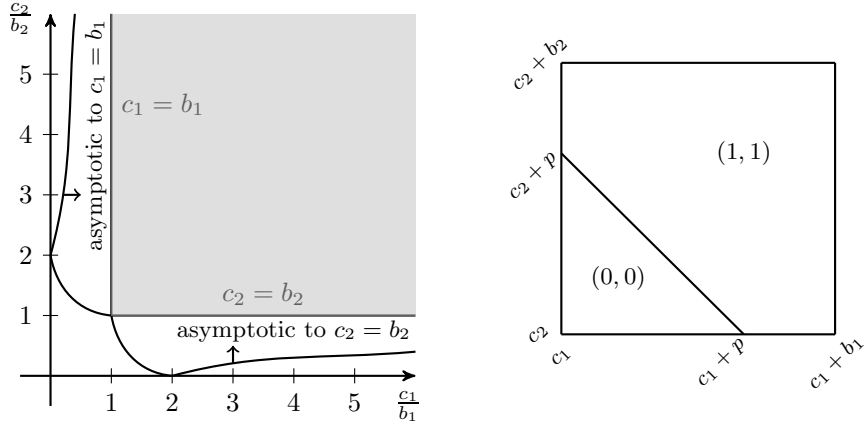

Figure \ref{fig:partition_uni} provides a self-explanatory `phase diagram' of optimal mechanisms as a function of $\frac{c_1}{b_1}$ and $\frac{c_2}{b_2}$. Interesting cases occur when either $c_1\leq b_1$ or $c_2\leq b_2$.  See Theorem \ref{thm:figc} of Section \ref{SUB:GC4}.

\subsection{Discussion}
\label{sub:discussion}

{\em Further refinements}. Observe that some of the regions in the phase diagram are not completely partitioned to map each figure to its corresponding region of optimality. For example, consider the region when $(c_1,c_2,b_1,b_2)$ satisfies (\ref{eqn:c1-c2-small}). The phase diagram does not indicate the sub-region where \ref{fig:a} is optimal. This is because the explicit map is not only dependent on the ratios $c_1/b_1$ and $c_2/b_2$, but also on the individual values of $(c_1,c_2,b_1,b_2)$.

We illustrate this dependence in Figure \ref{fig:full-split}. We consider three examples of $(b_1,b_2)$, numerically compute the optimal mechanisms, and provide a complete partition of the $\frac{c_1}{b_1},\frac{c_2}{b_2}$ space. Observe that in each example, the region of optimality for the structures is different. However, the mechanisms are scale-invariant, and thus we can fix $b_2=1$, without loss of generality. These examples also prove the `existence result', i.e., for each of these eight structures, there exists some $(c_1,c_2,b_1,b_2)$ for which that structure is optimal.

\begin{figure}[h!]
\centering
\begin{tabular}{cc}
\subfloat[]{\label{fig:b1-3-by-5}\includegraphics[height=6cm, width=6cm]{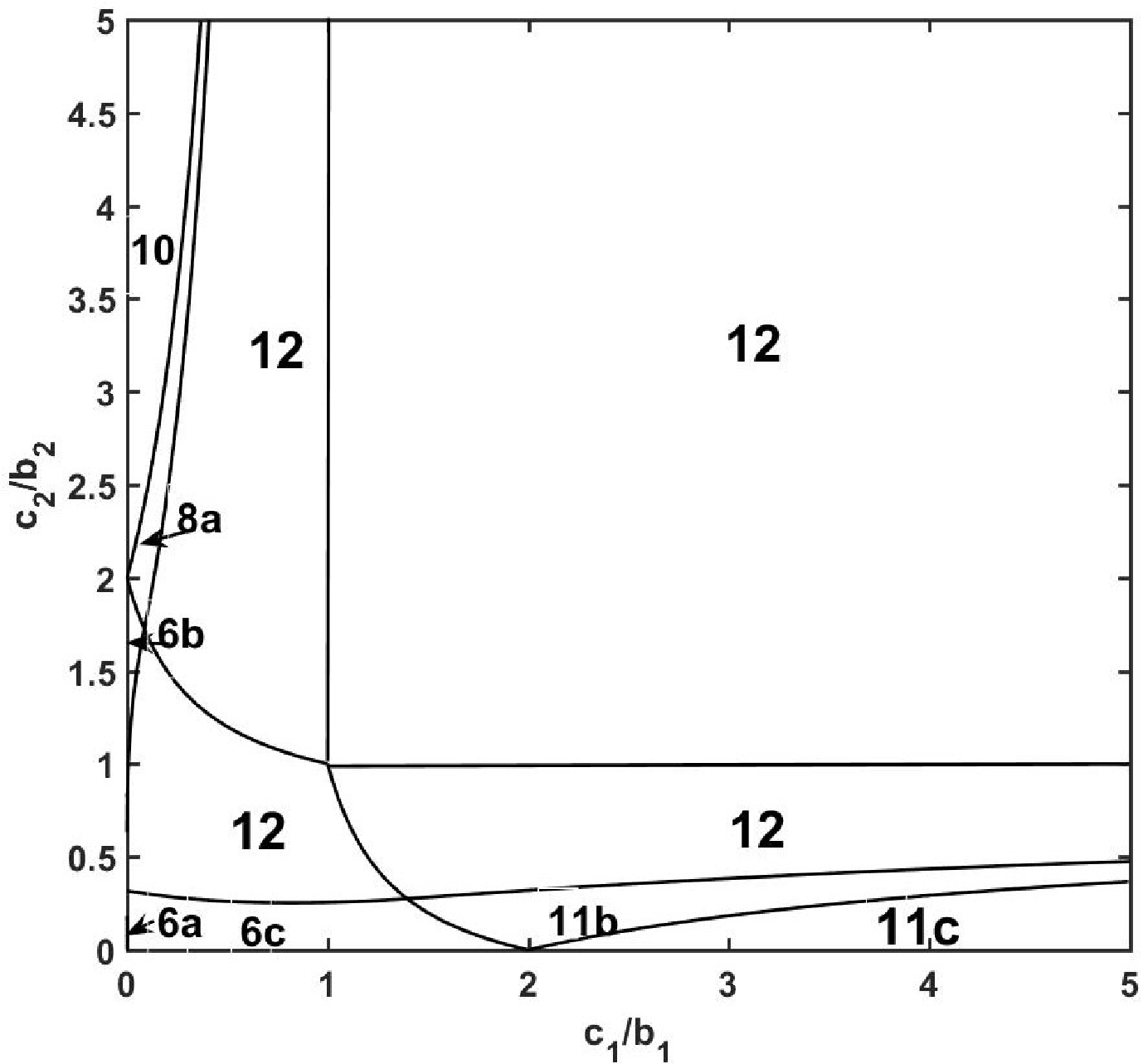}}&
\subfloat[]{\label{fig:b1-3-by-2}\includegraphics[height=6cm, width=6cm]{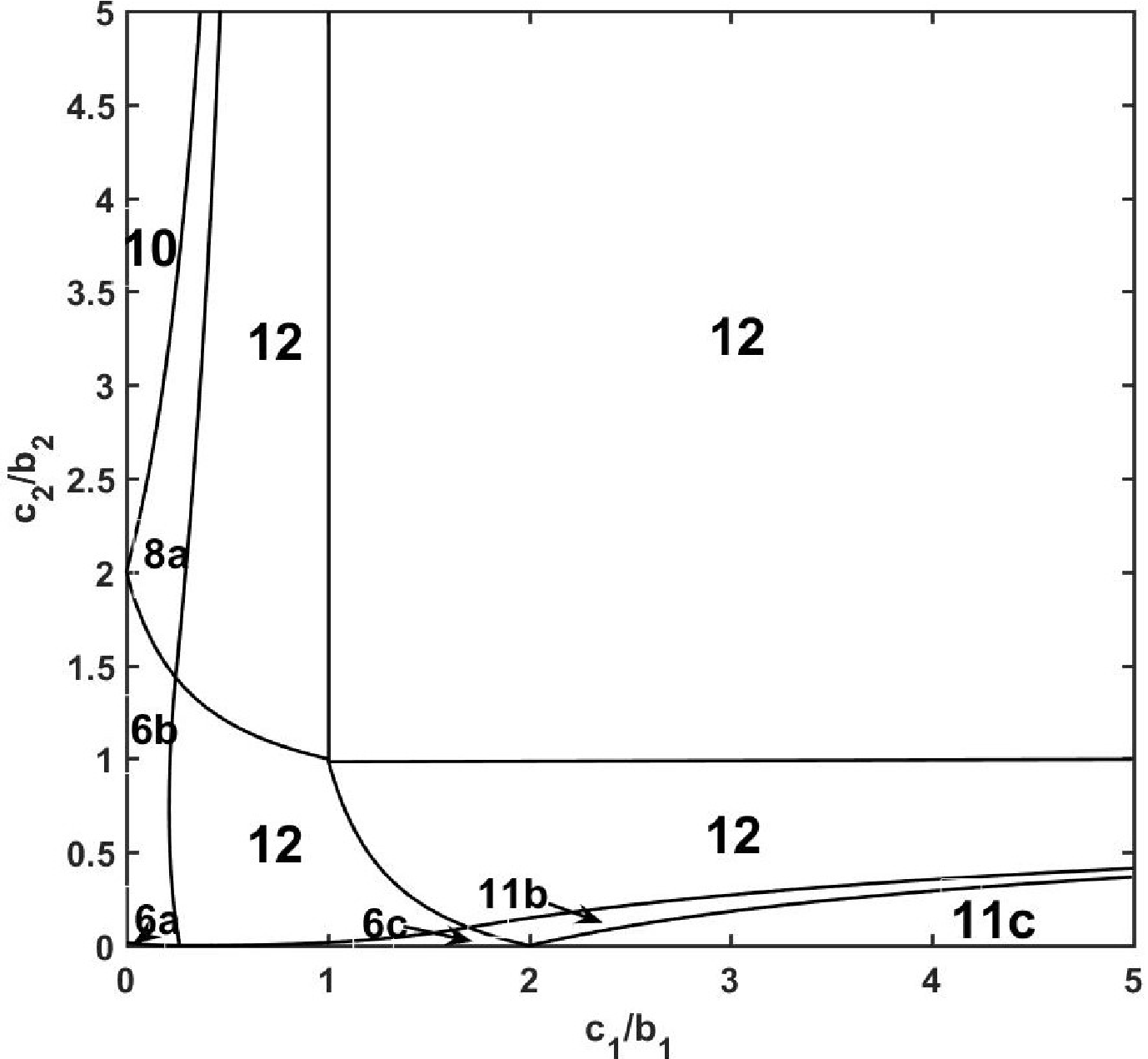}}\\
\multicolumn{2}{c}{\subfloat[]{\label{fig:b1-1}\includegraphics[height=6cm, width=6cm]{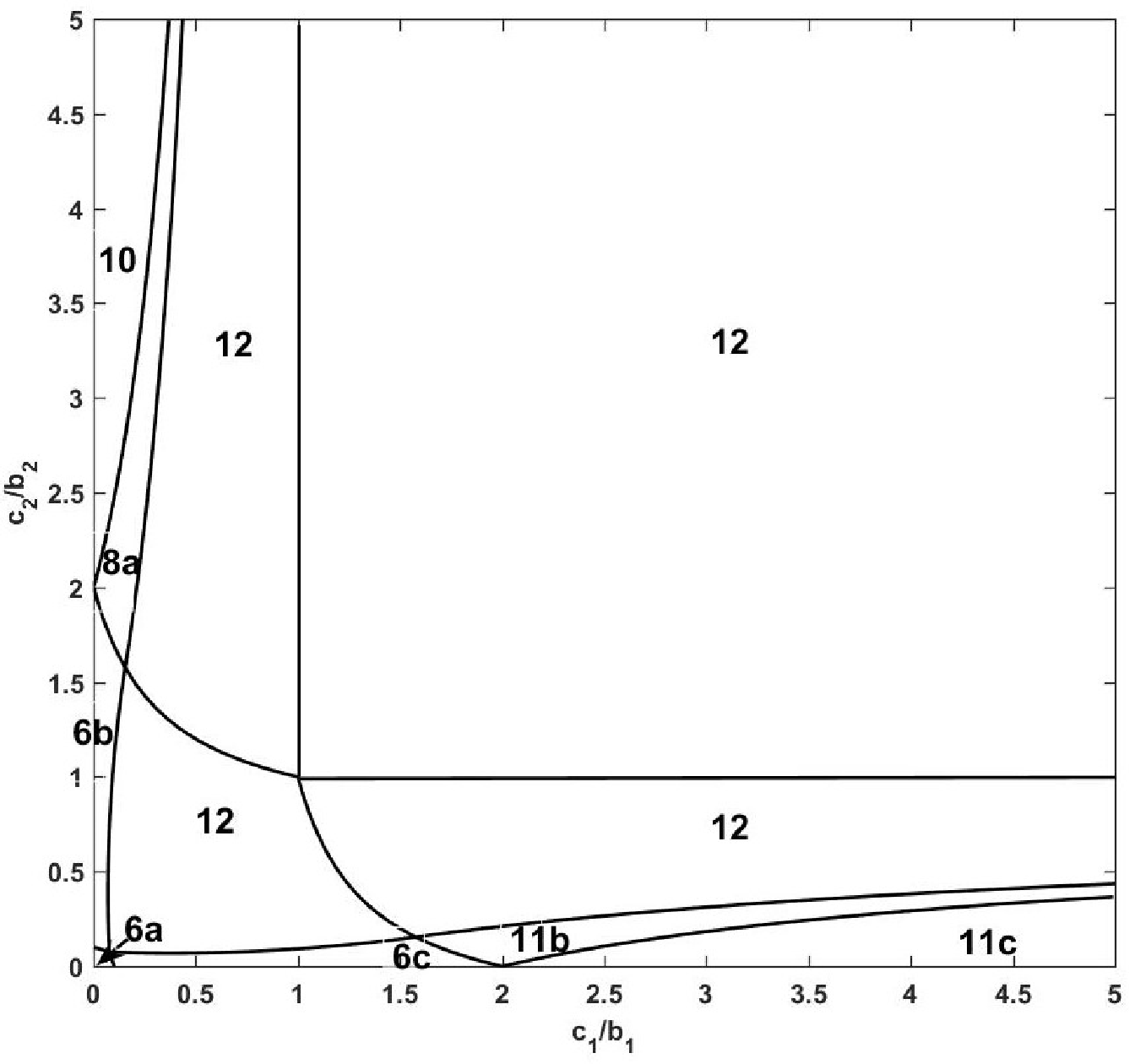}}}
\end{tabular}
\caption{The optimal mechanisms when (a) $b_1=0.6; b_2=1$ (b) $b_1=1.5; b_2=1$ (c) $b_1=1; 
b_2=1$.}\label{fig:full-split}
\end{figure}

{\em Four menu items}. Recall from footnote 2 that the power rate $\Delta:D\rightarrow\mathbb{R}$ is 
$$\Delta(z_1,z_2):=-3-\frac{z_1f_1'(z_1)}{f_1(z_1)}-\frac{z_2f_2'(z_2)}{f_2(z_2)}.$$ 
\citet{WT14} proved that the optimal mechanism for a distribution with $\Delta$ equaling a negative constant, has at most four menu items. The uniform distribution has power rate $-3$ for all $z\in D$, and it is easy to verify that in each of the structures in Figure \ref{fig:gen-structure}, the number of menu items is at most four, in agreement with the result in \cite{WT14}.

{\em Revenue monotonicity}. Not all structures with four menu items are possible. We now elaborate on this. \citet{HR15} showed that, in general, the optimal mechanism may have $z\geq z'\nRightarrow t(z)\geq t(z')$, i.e., the optimal revenue could decrease despite an increase in the buyer's valuations for both the items. Consider the hypothetical structures depicted in Figure \ref{fig:anomaly}. Each has four menu items in accordance with \cite{WT14}. But there exist $z$ and $z'$ with $z \geq z'$ and yet $t(z)<t(z')$, if the prices for the allocations $(a_1,1)$ and $(1,a_2)$ are different. From our results in this paper, in the case of uniformly distributed valuations, the optimal mechanisms are as in Figure \ref{fig:gen-structure} and satisfy $z\geq z'\Rightarrow q(z)\geq q(z')$ and thus $z\geq z'\Rightarrow t(z)\geq t(z')$ which is revenue monotonicity. In particular, the hypothetical structures in Figure \ref{fig:anomaly} cannot occur.

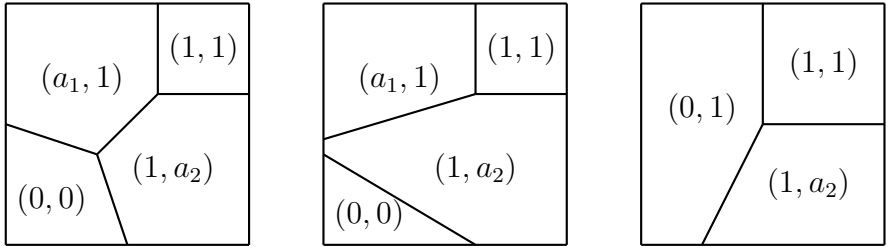
\begin{figure}[h!]
\centering
\begin{minipage}{.32\textwidth}
\centering
\begin{tikzpicture}[scale=0.4,font=\small,axis/.style={very thick, ->, >=stealth'}]
\draw [axis,thick,-] (1,1)--(9,1);
\draw [axis,thick,-] (9,1)--(9,9);
\draw [axis,thick,-] (9,9)--(1,9);
\draw [axis,thick,-] (1,9)--(1,1);
\draw [axis,thick,-] (1,5)--(4,4);
\draw [axis,thick,-] (4,4)--(5,1);
\draw [axis,thick,-] (4,4)--(6,6);
\draw [axis,thick,-] (6,6)--(6,9);
\draw [axis,thick,-] (6,6)--(9,6);
\node at (2.5,2.5) {\large$(0,0)$};
\node at (6.5,3.5) {\large$(1,a_2)$};
\node at (3.5,6.5) {\large$(a_1,1)$};
\node at (7.5,7.5) {\large$(1,1)$};
\end{tikzpicture}
\end{minipage}
\begin{minipage}{.32\textwidth}
\centering
\begin{tikzpicture}[scale=0.4,font=\small,axis/.style={very thick, ->, >=stealth'}]
\draw [axis,thick,-] (1,1)--(9,1);
\draw [axis,thick,-] (9,1)--(9,9);
\draw [axis,thick,-] (9,9)--(1,9);
\draw [axis,thick,-] (1,9)--(1,1);
\draw [axis,thick,-] (1,4)--(6,1);
\draw [axis,thick,-] (1,4.5)--(6,6);
\draw [axis,thick,-] (6,6)--(6,9);
\draw [axis,thick,-] (6,6)--(9,6);
\node at (2.5,2) {\large$(0,0)$};
\node at (6,3.5) {\large$(1,a_2)$};
\node at (3.5,6.5) {\large$(a_1,1)$};
\node at (7.5,7.5) {\large$(1,1)$};
\end{tikzpicture}
\end{minipage}
\begin{minipage}{.32\textwidth}
\centering
\begin{tikzpicture}[scale=0.4,font=\small,axis/.style={very thick, ->, >=stealth'}]
\draw [axis,thick,-] (1,1)--(9,1);
\draw [axis,thick,-] (9,1)--(9,9);
\draw [axis,thick,-] (9,9)--(1,9);
\draw [axis,thick,-] (1,9)--(1,1);
\draw [axis,thick,-] (3,1)--(5,5);
\draw [axis,thick,-] (5,5)--(5,9);
\draw [axis,thick,-] (5,5)--(9,5);
\node at (6.5,3) {\large$(1,a_2)$};
\node at (3,5.5) {\large$(0,1)$};
\node at (7,7) {\large$(1,1)$};
\end{tikzpicture}
\end{minipage}
\caption{An illustration of the structures having at most four menu items and not satisfying the condition $z\geq z'\Rightarrow t(z)\geq t(z')$.}\label{fig:anomaly}
\end{figure}

{\em No exclusion region}. When $\frac{c_1}{b_1}$ is small and $\frac{c_2}{b_2}$ is very large, in the sense that $\{c_1\leq b_1, c_2\geq 2b_2(b_1/(b_1-c_1))^2\}$, we show (in Theorem \ref{thm:gc5}) that the optimal mechanism is to sell according to the menu depicted in Figure \ref{fig:e}, i.e., to sell the second item with probability $1$ for the least valuation $c_2$, and sell item $1$ at a reserve price as indicated by Myerson's revenue maximizing mechanism. Similar is the case when $\frac{c_2}{b_2}$ is small and $\frac{c_1}{b_1}$ is very large.

This is interesting because it shows the existence of an optimal multi-dimensional mechanism without an exclusion region. An intuitive explanation of why we do not have an exclusion region in Figure \ref{fig:e} is as follows. Consider the case where the seller offers each allocation with a small increase in price, say $\epsilon$. The seller then loses a revenue of $c_2$ from the valuations $\{z:u(z)\leq\epsilon\}$, and gains an extra revenue of $\epsilon$ from the valuations $\{z:u(z)\geq\epsilon\}$. The mechanism will have no exclusion region when the loss dominates the gain. For small values of $\epsilon$, observe that the expected loss in revenue is 
$$c_2\cdot Pr(\{z:u(z)\leq\epsilon\})=c_2\frac{\epsilon(b_1-c_1)/2}{b_1b_2}+((c_1+b_1)/2+c_2)\frac{\epsilon^2}{2}\approx c_2\frac{\epsilon(b_1-c_1)/2}{b_1b_2},$$ 
and that the expected gain in revenue is 
$$\epsilon\cdot Pr(\{z:u(z)\geq\epsilon\})=\epsilon\cdot(1-Pr(\{z:u(z)\leq\epsilon\}))\approx\epsilon.$$ 
The loss dominates the gain when $c_2\geq 2b_2\frac{b_1}{b_1-c_1}$. (The actual threshold will depend on more precise calculations than our order estimates.) Both the loss and the gain are of the order of $\epsilon$, which explains the possibility of loss dominating gain at very high values of $c_2$. Figure \ref{fig:h} has no exclusion region due to a symmetric reasoning.

{\em Deterministic sale}. When both $\frac{c_1}{b_1}$ and $\frac{c_2}{b_2}$ are large, $\{c_1\geq b_1, c_2\geq b_2\}$, the optimal mechanism is to bundle the two items and sell the bundle at the reserve price. The reserve price $c_1+c_2+p$ is such that $p=\left(\sqrt{(c_1+c_2)^2+6b_1b_2}-c_1-c_2\right)/3$. So as $c_1+c_2\rightarrow\infty$, the reserve price is $c_1+c_2+O(b_1b_2/(c_1+c_2))$. When $c_1<b_1$, the mechanism is deterministic for any $c_2\geq 2b_2(b_1/(b_1-c_1))^2$, and when $c_1\geq b_1$, it is deterministic for any $c_2\geq b_2$.

\section{The Solution for the Uniform Density on a Rectangle}\label{sec:gencase}
In this section, we formally identify the optimal mechanism when $z\sim\mbox{Unif}[c_1,c_1+b_1]\times[c_2,c_2+b_2]$. We compute the components of $\bar{\mu}$ (i.e., $\mu,\mu_s,\mu_p$), with $f(z)=\frac{1}{b_1b_2}$ for $z\in D=[c_1,c_1+b_1]\times[c_2,c_2+b_2]$, are as follows:
\begin{align*}
 \mbox{(area density) }\, & \mu(z)=-3/(b_1b_2),\quad z \in D\\
 \mbox{(line density) }\, & \mu_s(z)=\sum_{i=1}^2(-c_i\mathbf{1}(z_i=c_i)+(c_i+b_i)\mathbf{1}(z_i=c_i+b_i))/(b_1b_2),\\ &\hspace*{3.2in} z\in\partial D\\
 \mbox{(point measure) }\, & \mu_p(\{c_1,c_2\})=1.
\end{align*}

To construct the canonical partition of $D$, we first need to compute the outer boundary functions $s_i(\cdot)$. When $c_i>0$, we have $\mu_s(c_i,z_{-i})=-c_i/(b_1b_2)<0$ for all $z_{-i}$. A natural extension of the definition of $s_i(z_i)$ in (\ref{eqn:si-zi}) to accommodate the line density $\mu_s$ at $z_i=c_i$ is to write
\begin{multline*}
  s_i(c_i)=\max\left\{z_{-i}\in[c_{-i},c_{-i}+b_{-i}):\right.\\
  \left.\int_{z_{-i}^*}^{c_{-i}+b_{-i}}\mu_s(c_i,z_{-i})\,dz_{-i}+\mu_p(c_i,c_{-i}+b_{-i})=0\right\}.
\end{multline*}
However, this is inadequate since $\mu_s(c_i,z_{-i})<0$ for all $z_{-i}$ and $\mu_p(c_i,c_{-i}+b_{-i})=0$, and the integral that goes to define $s_i(c_i)$ never equals 0. So the method of construction of canonical partition explained in Section \ref{sec:zero}, especially that of computing the outer boundary functions $s_i$ at $z_i=c_i$, cannot be used. The construction method needs an extension. Specifically, when $(c_1,c_2)\ne (0,0)$, we need to shuffle $\bar{\mu}$ so that the function $s_i$ is defined everywhere in $[c_i,c_i+b_i)$. We add to $\bar{\mu}$ a ``shuffling measure'' $\bar{\alpha}$ that convex dominates $0$. We then find $s_i(z_i)$ using $\bar{\mu}+\bar{\alpha}$, instead of $\bar{\mu}$.

We now investigate the possible structure of such a shuffling measure. We know (from \cite{Pav11}) that the optimal mechanism is as in Figure \ref{fig:a} when $z\sim\mbox{Unif}[c,c+1]^2$, $c\leq 0.077$, and also that $a_1=a_2=0$ when $z\sim\mbox{Unif}[0,1]^2$. In other words, the outer boundary functions $s_i(z_i)$ are constant when $c=0$, and linear when $c\in(0,0.077]$. Thus we anticipate that adding a shuffling measure whose density is linear on the top-left and bottom-right boundaries, with point measures at the top-left and bottom-right corners to neutralize the negative line densities at $z_i=c_i$, would yield the solution.

\citet{DDT17} identify such a shuffling measure when they solve for $D = [4,16]\times[4,7]$. In this paper, we identify shuffling measures for all rectangles $D$ on the positive quadrant. Specifically, we do the following.
\begin{enumerate}
\item[(i)] We suitably parametrize the shuffling measure $\bar{\alpha}$, and find the conditions on the parameters for $\bar{\alpha}\succeq_{cvx}0$ to hold.
\begin{itemize}
\item The shuffling measure is defined using six parameters: $(p_{a_1},a_1,m_1)$ and $(p_{a_2},a_2,m_2)$. The parameters $(p_{a_1},a_1,m_1)$ are respectively the functions of the start point, the slope, and the length of the line segment $P'P$ when projected on to the x-axis; see Figure \ref{fig:illust-parameters}. The description of $(p_{a_2},a_2,m_2)$ is similar, and $m_2$ is the length of the line segment $QQ'$ when projected on to the y-axis. See Figure \ref{fig:illust-parameters}.
\item Using the conditions for $\bar{\alpha}\succeq_{cvx}0$, we write $(a_i,m_i)$ in terms of $p_{a_i}$. We then have only two parameters, $p_{a_1}$ and $p_{a_2}$, to evaluate.
\end{itemize}
\item[(ii)] We then identify the canonical partition with respect to $\bar{\mu}+\bar{\alpha}$. The outer boundary functions $s_i(\cdot)$, the critical price $p$, and the critical points $P$ and $Q$, are expressed in terms of $p_{a_1}$ and $p_{a_2}$. We use a slightly modified procedure to identify the canonical partition (as explained below), instead of the procedure described in Section \ref{sec:zero}.
\begin{itemize}
\item We first compute the outer boundary functions $s_i(\cdot)$ using the components of $\bar{\mu}+\bar{\alpha}$, exactly as enumerated in bullet (a), in Section \ref{sec:zero}.
\item We then compute the critical points $P$ and $Q$ before computing the parameter $p$. This is because the critical points can be fixed using the parameters $m_1$ and $m_2$ computed in the previous bullet.
\item We then compute the critical price $p$ which must be such that, in addition to satisfying $\bar{\mu}(Z)=0$, the critical points $P$ and $Q$ in Figure \ref{fig:illust-parameters} must be connected by a $135^\circ$ line.
\item We finally compute $p_{a_1}$ and $p_{a_2}$ that simultaneously solve the polynomials that arise out of the two constraints.
\end{itemize}
\item[(iii)] Then, for each rectangle $D$, we identify the parameters of $\bar{\alpha}$ so that the canonical partition becomes a valid partition. The dual variable is then computed as in the proof of Proposition \ref{prop:known}.
\begin{itemize}
\item The parameters thus computed need not give rise to a valid partition for all $(c_1,c_2,b_1,b_2)\geq 0$. Specifically, the probability of allocation, $a_i$, could be more than $1$ or the critical points $P$ and $Q$ could be outside $D$. Recall from Section \ref{sec:zero} that the point $Q$ fell outside $D$ when $b_1/b_2>2$, during our initial analysis.
\item In Sections 5.1.1--5.1.4, we consider the case when $(c_i,b_i)$ satisfies (\ref{eqn:c1-c2-small}). We show that the optimal mechanism is one of the four structures in Figure \ref{fig:a} -- \ref{fig:c}, depending on $a_i\leq 1$ or $a_i>1$, $i=1,2$.
\item In Section 5.2, we consider the case when $\frac{c_1}{b_1}$ is large and $\frac{c_2}{b_2}$ is small. We show that the critical point $P$ moves below the bottom boundary of $D$ if we continue to use $\bar{\alpha}$ as the shuffling measure. We therefore switch to $\bar{\beta}$, a variant of $\bar{\alpha}$. Using similar arguments as in Section 5.1, we show that the optimal mechanism is one of the three structures in Figures \ref{fig:gen-structure-2} and \ref{fig:e}. A symmetric argument is employed in Section 5.3.
\item In Section 5.4, we consider the case when both $\frac{c_1}{b_1}$ and $\frac{c_2}{b_2}$ are large. We prove that pure bundling is optimal, using some sufficient conditions derived in the earlier sections.
\end{itemize}
\end{enumerate}
We now fill in the details. We begin by describing the shuffling measure $\bar{\alpha}$.
\begin{figure}
\centering
\includegraphics[height=5.5cm,width=5.5cm]{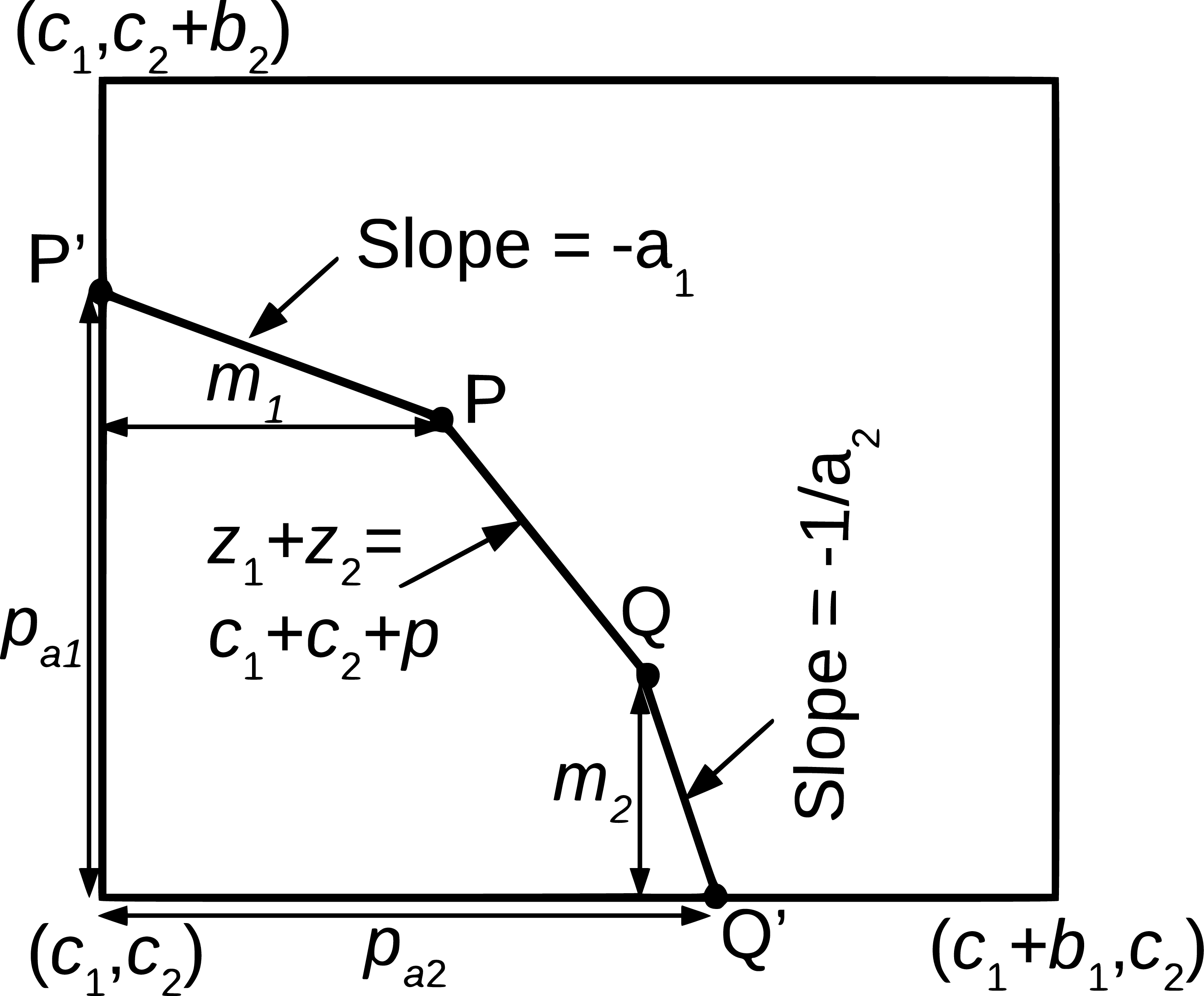}
\caption{An illustration of the parameters in the shuffling measure $\bar{\alpha}$.}\label{fig:illust-parameters}
\end{figure}

\subsection{Optimal mechanisms when $\frac{c_1}{b_1}$ and $\frac{c_2}{b_2}$ are small}\label{SUB:GC1}
For $m_1>0$, define $\tilde{D}^{(1)}:[c_1,c_1+b_1]\times\{c_2+b_2\}$, an interval on the top boundary of $D$ starting from the top-left corner. For $a_1, p_{a_1} > 0$, define a linear function $\alpha_s^{(1)}:\tilde{D}^{(1)}\rightarrow\mathbb{R}$, given by
\begin{equation}\label{eqn:alpha_s}
\alpha_s^{(1)}(x,c_2+b_2):=(2b_2-c_2-3p_{a_1}+3a_1(x-c_1))/(b_1b_2), \quad x \in [c_1,c_1+m_1].
\end{equation}
Define $\alpha_p^{(1)}:=c_1(b_2-p_{a_1})/(b_1b_2)\delta_{(c_1,c_2+b_2)}$, a point measure of mass $c_1(b_2-p_{a_1})/(b_1b_2)$ at location $(c_1,c_2+b_2)$. Finally, define the measure
$$
\bar{\alpha}^{(1)}(A):=\int_{c_1}^{c_1+m_1}\mathbf{1}_A(x,c_2+b_2)\alpha_s^{(1)}(x,c_2+b_2)\,dx+\alpha_p^{(1)}(A\cap(c_1,c_2+b_2))
$$
for all measurable sets $A\subseteq\tilde{D}^{(1)}$. See Figure \ref{fig:measure}.

Observe that the measure $\bar{\alpha}^{(1)}$ is characterized by three parameters: $p_{a_1}$, $a_1$, $m_1$. As the discussion proceeds, we will observe that these parameters determine respectively the price of the allocation $(a_1,1)$, the probability of allocation of the first good $a_1$, and the point of transition between the allocation regions $(a_1,1)$ and $(1,1)$.

We now discuss a property of measures that convex dominates zero. All proofs of assertions in this subsection are relegated to \ref{app:a}.
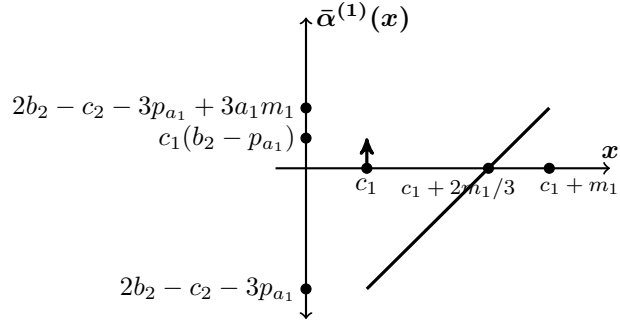
\begin{figure}
\centering
\begin{tikzpicture}[scale=0.4,font=\small,axis/.style={very thick, ->, >=stealth'}]
\draw[thick,<->] (0,-5)--(0,5);
\draw[thick,->] (-1,0)--(10,0);
\node[right] at (0,5) {$\bm{\bar{\alpha}^{(1)}(x)}$};
\node[above] at (10,0) {$\bm{x}$};
\draw[axis,->] (2,0)--(2,1);
\node[left] at (0,1) {$c_1(b_2-p_{a_1})$};
\draw[black,fill=black] (0,1) circle (1ex);
\draw[axis,-] (2,-4)--(8,2);
\node[left] at (0,-4) {$2b_2-c_2-3p_{a_1}$};
\draw[black,fill=black] (0,-4) circle (1ex);
\node[left] at (0,2) {$2b_2-c_2-3p_{a_1}+3a_1m_1$};
\draw[black,fill=black] (0,2) circle (1ex);
\node[below] at (2,0) {$c_1$};
\draw[black,fill=black] (2,0) circle (1ex);
\node[below] at (5,0) {\scriptsize$c_1+2m_1/3$};
\draw[black,fill=black] (6,0) circle (1ex);
\node[below] at (9,0) {\scriptsize$c_1+m_1$};
\draw[black,fill=black] (8,0) circle (1ex);
\end{tikzpicture}
\caption{The measure $\bar{\alpha}^{(1)}$.}\label{fig:measure}
\end{figure}
\begin{proposition}\label{prop:cvx}
Consider a measure $\alpha$ defined in the interval $[c_1,c_1+m_1]$, whose density is nonnegative in the intervals $[c_1,c_1+l_1]$ and $[c_1+l_2,c_1+m_1]$, and is nonpositive in the interval $[c_1+l_1,c_1+l_2]$, for some $l_1 < l_2$. If $\alpha([c_1,c_1+m_1])=\int_{[c_1,c_1+m_1]}x\,\alpha(dx)=0$, then $\alpha\succeq_{cvx}0$.
\end{proposition}
\begin{corollary}\label{cor:cvx}
Suppose $m_1=\frac{4c_1(b_2-p_{a_1})}{c_2-2b_2+3p_{a_1}}$ and $a_1=\frac{(c_2-2b_2+3p_{a_1})^2}{8c_1(b_2-p_{a_1})}$. Then, $\bar{\alpha}^{(1)}(\tilde{D}^{(1)})=\int_{[c_1,c_1+m_1]}x\,\bar{\alpha}^{(1)}(dx,c_2+b_2)=0$, and hence $\int_{\tilde{D}^{(1)}}f\,d\bar{\alpha}^{(1)}=0$ for any affine function $f$ on $\tilde{D}^{(1)}$. Furthermore, $\bar{\alpha}^{(1)}\succeq_{cvx}0$.
\end{corollary}

We now define a similar interval $\tilde{D}^{(2)}:=\{c_1+b_1\}\times[c_2+m_2]$ for some $m_2>0$ on the right boundary of $D$, starting from the bottom-right corner. For $a_2, p_{a_2} > 0$, we define a similar line measure $\alpha_s^{(2)}$ at the interval $\tilde{D}^{(2)}$, a point measure $\alpha_p^{(2)}$ at $(c_1+b_1,c_2)$, and
$$
\bar{\alpha}^{(2)}(A):=\int_{c_2}^{c_2+m_2}\mathbf{1}_A(c_1+b_1,x)\alpha_s^{(2)}(c_1+b_1,x)\,dx+\alpha_p^{(2)}(A\cap(c_1+b_1,c_2)).
$$
As we will later see, the parameters $p_{a_2}$, $a_2$, and $m_2$ that characterize $\bar{\alpha}^{(2)}$ will respectively determine the price of allocation $(1,a_2)$, the probability $a_2$, and the point of transition between $(1,a_2)$ and $(1,1)$. Define $\bar{\alpha}:=\bar{\alpha}^{(1)}+\bar{\alpha}^{(2)}$ on $D$.

We now identify the canonical partition of $D$ with respect to $\bar{\mu}+\bar{\alpha}$, using the steps enumerated in bullet (ii), Section \ref{sec:gencase}. The construction is illustrated in Figure \ref{fig:illust-sec5}.
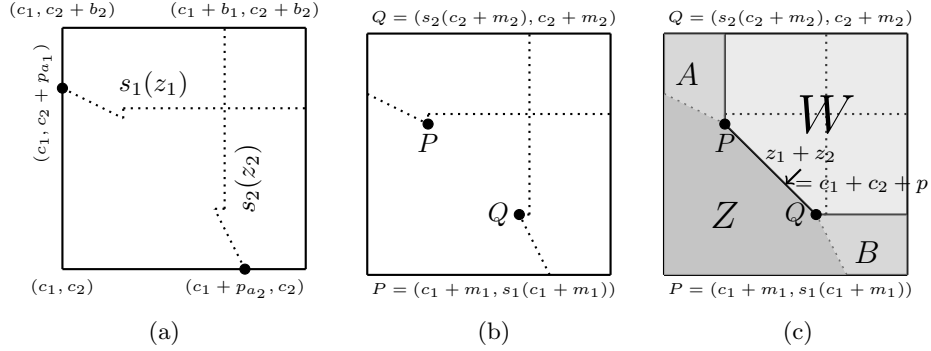
\begin{figure}[h!]
\centering
\begin{tabular}{ccc}
\subfloat[]{\begin{tikzpicture}[scale=0.4,font=\small,axis/.style={very thick, ->, >=stealth'}]
\draw [axis,thick,-] (1,1)--(9,1);
\draw [axis,thick,-] (9,1)--(9,9);
\draw [axis,thick,-] (9,9)--(1,9);
\draw [axis,thick,-] (1,9)--(1,1);
\draw [thick,dotted] (1,7)--(3,6);
\draw [thick,dotted] (3,6)--(3,19/3);
\draw [thick,dotted] (3,19/3)--(9,19/3);
\draw [thick,dotted] (7,1)--(6,3);
\draw [thick,dotted] (19/3,3)--(6,3);
\draw[thick,dotted] (19/3,3)--(19/3,9);
\node [above] at (4,19/3) {$s_1(z_1)$};
\node [below,rotate=90] at (19/3,4) {$s_2(z_2)$};
\node [below] at (1,1) {\tiny$(c_1,c_2)$};
\node [above] at (7,9) {\tiny$(c_1+b_1,c_2+b_2)$};
\node [above] at (1,9) {\tiny$(c_1,c_2+b_2)$};
\draw[black,fill=black] (7,1) circle (1ex);
\node [below] at (7,1) {\tiny$(c_1+p_{a_2},c_2)$};
\draw[black,fill=black] (1,7) circle (1ex);
\node [above,rotate=90] at (1,6.5) {\tiny$(c_1,c_2+p_{a_1})$};
\end{tikzpicture}}&
\subfloat[]{\begin{tikzpicture}[scale=0.4,font=\small,axis/.style={very thick, ->, >=stealth'}]
\draw [axis,thick,-] (1,1)--(9,1);
\draw [axis,thick,-] (9,1)--(9,9);
\draw [axis,thick,-] (9,9)--(1,9);
\draw [axis,thick,-] (1,9)--(1,1);
\draw [thick,dotted] (1,7)--(3,6);
\draw [thick,dotted] (3,6)--(3,19/3);
\draw [thick,dotted] (3,19/3)--(9,19/3);
\draw [thick,dotted] (7,1)--(6,3);
\draw [thick,dotted] (19/3,3)--(6,3);
\draw[thick,dotted] (19/3,3)--(19/3,9);
\draw[black,fill=black] (3,6) circle (1ex);
\node [below] at (3,6) {$P$};
\draw[black,fill=black] (6,3) circle (1ex);
\node [left] at (6,3) {$Q$};
\node [right] at (0.8,0.5) {\tiny$P=(c_1+m_1,s_1(c_1+m_1))$};
\node [right] at (0.8,9.5) {\tiny$Q=(s_2(c_2+m_2),c_2+m_2)$};
\end{tikzpicture}}&
\subfloat[]{\begin{tikzpicture}[scale=0.4,font=\small,axis/.style={very thick, ->, >=stealth'}]
\draw [axis,thick,-] (1,1)--(9,1);
\draw [axis,thick,-] (9,1)--(9,9);
\draw [axis,thick,-] (9,9)--(1,9);
\draw [axis,thick,-] (1,9)--(1,1);
\draw [thick,dotted] (1,7)--(3,6);
\draw [thick,dotted] (3,6)--(3,19/3);
\draw [thick,dotted] (3,19/3)--(9,19/3);
\draw [thick,dotted] (7,1)--(6,3);
\draw [thick,dotted] (19/3,3)--(6,3);
\draw[thick,dotted] (19/3,3)--(19/3,9);
\path[fill=gray!50,opacity=.9] (1,7) to (3,6) to (6,3) to (7,1) to (1,1) to (1,7);
\node at (3,3) {\Large$Z$};
\draw [axis,thick,-] (3,6)--(3,9);
\draw [axis,thick,-] (6,3)--(9,3);
\draw [axis,thick,-] (3,6)--(6,3);
\path[fill=gray!50,opacity=.6] (1,7) to (3,6) to (3,9) to (1,9) to (1,7);
\node at (1.75,23/3) {\large$A$};
\path[fill=gray!50,opacity=.6] (7,1) to (6,3) to (9,3) to (9,1) to (7,1);
\node at (23/3,1.75) {\large$B$};
\path[fill=gray!50,opacity=.3] (3,6) to (6,3) to (9,3) to (9,9) to (3,9) to (3,6);
\node at (19/3,19/3) {\huge$W$};
\node at (5.5,5) {\scriptsize$z_1+z_2$};
\node at (7.5,4) {\scriptsize$=c_1+c_2+p$};
\draw [thick,->] (5.5,4.5)--(5,4);
\draw[black,fill=black] (3,6) circle (1ex);
\node [below] at (3,6) {$P$};
\draw[black,fill=black] (6,3) circle (1ex);
\node [left] at (6,3) {$Q$};
\node [right] at (0.8,0.5) {\tiny$P=(c_1+m_1,s_1(c_1+m_1))$};
\node [right] at (0.8,9.5) {\tiny$Q=(s_2(c_2+m_2),c_2+m_2)$};
\end{tikzpicture}}
\end{tabular}
\caption{An illustration of (a) the construction of $s_i(z_i)$ (b) the computation of $P$ and $Q$ and (c) the construction of exclusion set and the canonical partition.}\label{fig:illust-sec5}
\end{figure}

\begin{enumerate}
\item[(a)] The outer boundary functions $s_i(\cdot)$ with respect to $\bar{\mu}+\bar{\alpha}$ are computed using the following expressions that are similar to (\ref{eqn:si-zi}).
\begin{align}
  s_i(z_i)&=\max\left\{z_{-i}\in[c_{-i},c_{-i}+b_{-i}):\int_{z_{-i}^*}^{c_{-i}+b_{-i}}\mu(c_i,z_{-i})\,dz_{-i}\right.\nonumber\\
  &\hspace{0.75in}\left.+(\mu_s+\alpha_s^{(i)})(c_i,c_{-i}+b_{-i})=0\right\},\,\forall z_i\in(c_i,c_i+b_i),\nonumber\\
  s_i(c_i)&=\max\left\{z_{-i}\in[c_{-i},c_{-i}+b_{-i}):\right.\nonumber\\
  &\hspace*{0.75in}\left.\int_{z_{-i}^*}^{c_{-i}+b_{-i}}\mu_s(c_i,z_{-i})\,dz_{-i}+\alpha_p^{(i)}(c_i,c_{-i}+b_{-i})=0\right\}.\label{eqn:si-zi-small}
\end{align}
Observe that the components of $\bar{\mu}$ in (\ref{eqn:si-zi}) have been replaced by the components of $(\bar{\mu}+\bar{\alpha})$ in the above equations. The values of $s_i(\cdot)$ are given by
$$
  s_i(z_i)=\begin{cases}c_{-i}+p_{a_i}-a_i(z_i-c_i),&c_i\leq z_i\leq c_i+m_i\\\frac{2}{3}(c_{-i}+b_{-i}),&c_i+m_i<z_i\leq c_i+b_i.\end{cases}
$$
Observe that $s_i$ is discontinuous at $z_i=c_i+m_i$. The values of $a_i$ and $m_i$ are as given in the statement of Corollary \ref{cor:cvx}. They are also given below.
\begin{align}
  a_1&=\frac{(c_2-2b_2+3p_{a_1})^2}{8c_1(b_2-p_{a_1})},\quad a_2=\frac{(c_1-2b_1+3p_{a_2})^2}{8c_2(b_1-p_{a_2})},\label{eqn:for-a1-a2}\\
  m_1&=\frac{4c_1(b_2-p_{a_1})}{c_2-2b_2+3p_{a_1}},\quad m_2=\frac{4c_2(b_1-p_{a_2})}{c_1-2b_1+3p_{a_2}}.\nonumber
\end{align}
\item[(b)] The critical points $P$ and $Q$ are given by $P=(c_1+m_1,s_1(c_1+m_1))$, and $Q=(s_2(c_2+m_2),c_2+m_2)$. Substituting the values of $s_i(c_i+m_i)$, $a_i$, and $m_i$ from the previous bullet, we have
\begin{align}
  P&=\left(c_1+4c_1(b_2-p_{a_1})/(c_2-2b_2+3p_{a_1}),c_2+(2b_2-c_2-p_{a_1})/2\right), \nonumber\\
  Q&=\left(c_1+(2b_1-c_1-p_{a_2})/2,c_2+4c_2(b_1-p_{a_2})/(c_1-2b_1+3p_{a_2})\right). \label{eqn:P-Q-new}
\end{align}
\item[(c)] We now choose the critical price $p$ so that (i) $\bar{\mu}(Z)=0$, and (ii) the critical points $P$ and $Q$ lie on the line $z_1+z_2=c_1+c_2+p$. The latter constraint holds when $P_1+P_2=Q_1+Q_2$ holds. Substituting the expressions for $P$ and $Q$ (from (\ref{eqn:P-Q-new})) and simplifying, we have
\begin{multline}\label{eqn:pa1-pa2}
  (c_1-c_2-2(b_1-b_2)-(p_{a_1}-p_{a_2}))(c_1-2b_1+3p_{a_2})(c_2-2b_2+3p_{a_1})\\+8c_1(b_2-p_{a_1})(c_1-2b_1+3p_{a_2})-8c_2(b_1-p_{a_2})(c_2-2b_2+3p_{a_1})=0.
\end{multline}
The constraint $\bar{\mu}(Z)=0$ is the same as $\bar{\mu}(W)=0$, since (i) the regions $A$ and $B$ have been constructed to satisfy $\bar{\mu}(A)=\bar{\mu}(B)=0$, and (ii) it is easy to verify that $\bar{\mu}(D)=0$. From an examination of Figure \ref{fig:a} and simple geometry considerations, it is immediate that $-\bar{\mu}(W)=0$ iff
\begin{multline*}
  3(b_1-(P_1-c_1))(b_2-(Q_2-c_2))-\frac{3}{2}(P_2-Q_2)(Q_1-P_1)\\-(c_2+b_2)(b_1-(P_1-c_1))-(c_1+b_1)(b_2-(Q_2-c_2))=0.
\end{multline*}
This upon substitution of the expressions for $P$ and $Q$ yields
\begin{multline}\label{eqn:big-W}
 3\,\prod_{i=1}^2(b_i(c_{-i}-2b_{-i}+3p_{a_i})-4c_i(b_{-i}-p_{a_i}))\\-\sum_{i=1}^2((b_i(c_{-i}-2b_{-i}+3p_{a_i})-4c_i(b_{-i}-p_{a_i}))(c_i-2b_i+3p_{a_{-i}})(c_{-i}+b_{-i}))\\-(3/8)\,\prod_{i=1}^2((2b_{-i}-c_{-i}-p_{a_i})(c_i-2b_i+3p_{a_{-i}})-4c_{-i}(b_i-p_{a_{-i}}))=0.
\end{multline}
The values of $(p_{a_1},p_{a_2})$ are computed by solving (\ref{eqn:pa1-pa2}) and (\ref{eqn:big-W}) simultaneously. Note that these equations are the same as (\ref{eqn:fig5a-initial-1}) and (\ref{eqn:fig5a-initial-2}) respectively. The critical price $p$ can then be computed from $P_1+P_2=Q_1+Q_2=c_1+c_2+p$, from which we derive $p$ as
\begin{equation}\label{eqn:for-p1}
  p=\frac{4c_1(b_2-p_{a_1})}{c_2-2b_2+3p_{a_1}}+\frac{2b_2-c_2-p_{a_1}}{2}=\frac{2b_1-c_1-p_{a_2}}{2}+\frac{4c_2(b_1-p_{a_2})}{c_1-2b_1+3p_{a_2}}.
\end{equation}
\item[(d)] We thus have the following canonical partition of $D$.
\begin{align*}
Z&=\{(z_1,z_2):z_2\leq s_1(z_1)\}\cap\{(z_1,z_2):z_1\leq s_2(z_2)\}\\
 &\hspace*{1in}\cap\{(z_1,z_2):z_1+z_2\leq c_1+c_2+p\},\\
A&=([c_1,P_1]\times[s_1(z_1),c_2+b_2])\backslash Z,\\
B&=([s_2(z_2),c_1+b_1]\times[c_2,Q_2])\backslash Z,\\
W&=D\backslash(A\cup B\cup Z).
\end{align*}
\end{enumerate}

The allocation and the payment functions $(q(z),t(z))$ are given by (\ref{eqn:q-full}).

We now verify if the values of $(p_{a_1},p_{a_2})$ that solve (\ref{eqn:pa1-pa2}) and (\ref{eqn:big-W}) simultaneously, give rise to a valid partition. The partition is valid, if $(p_{a_1},p_{a_2})$ are such that (i) $p_{a_i}\in[0,b_{-i}]$, (ii) $a_i\in[0,1]$, (iii) the critical points $P$ and $Q$ lie in $D$, and (iv) $P$ lies to the north-west of $Q$.

To prove the validity of the partition, we do the following.
\begin{itemize}
\item We first analyze how the values of $a_i$ and the points $P$ and $Q$ vary when $p_{a_i}$ varies.
\item We then compute $r_i$ such that the critical points coincide when we fix $p_{a_i}=r_i$. This constitutes a start point for our exploration. We also compute $p_{a_i}^*$ such that $a_i=1$ when we fix $p_{a_i}=p_{a_i}^*$. This constitutes the end point for our exploration.
\item Start at $p_{a_i}=r_i$. Move the values of $p_{a_1}$ and $p_{a_2}$ in such a way that $P_1+P_2=Q_1+Q_2$ holds, and stop either when $p_{a_i}=p_{a_i}^*$ for some $i$, or when $\bar{\mu}(W)=0$. We consider four cases based on whether $r_i\leq p_{a_i}^*$ or not.
\item A case-by-case analysis will then establish that the optimal mechanism falls within one of the four structures depicted in Figure \ref{fig:gen-structure-1}, when $c_i, b_i$ satisfy (\ref{eqn:c1-c2-small}).
\end{itemize}

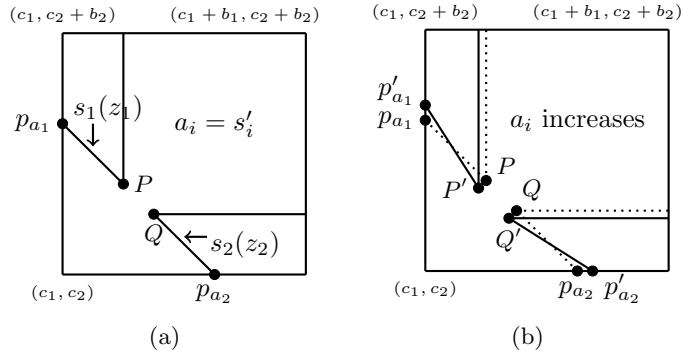
\begin{figure}[h!]
\centering
\begin{tabular}{cc}
\subfloat[]{\label{fig:obs-1}\begin{tikzpicture}[scale=0.4,font=\small,axis/.style={very thick, ->, >=stealth'}]
\draw [axis,thick,-] (1,1)--(9,1);
\draw [axis,thick,-] (9,1)--(9,9);
\draw [axis,thick,-] (9,9)--(1,9);
\draw [axis,thick,-] (1,9)--(1,1);
\draw [thick] (1,6)--(3,4);
\draw [thick] (6,1)--(4,3);
\draw [thick] (3,4)--(3,9);
\draw [thick] (4,3)--(9,3);
\node [below] at (1,1) {\tiny$(c_1,c_2)$};
\node [above] at (7,9) {\tiny$(c_1+b_1,c_2+b_2)$};
\node [above] at (1,9) {\tiny$(c_1,c_2+b_2)$};
\draw[black,fill=black] (6,1) circle (1ex);
\node [below] at (6,1) {$p_{a_2}$};
\draw[black,fill=black] (1,6) circle (1ex);
\node [left] at (1,6) {$p_{a_1}$};
\draw[black,fill=black] (3,4) circle (1ex);
\node [right] at (3,4) {\footnotesize$P$};
\draw[black,fill=black] (4,3) circle (1ex);
\node [below] at (4,3) {\footnotesize$Q$};
\node at (2.5,6.5) {$s_1(z_1)$};
\draw [thick,->] (2,6)--(2,5.25);
\node at (7,2) {$s_2(z_2)$};
\draw [thick,->] (5.75,2.25)--(5,2.25);
\node at (6,6) {$a_i=s_i'$};
\end{tikzpicture}}&
\subfloat[]{\label{fig:obs-2}\begin{tikzpicture}[scale=0.4,font=\small,axis/.style={very thick, ->, >=stealth'}]
\draw [axis,thick,-] (1,1)--(9,1);
\draw [axis,thick,-] (9,1)--(9,9);
\draw [axis,thick,-] (9,9)--(1,9);
\draw [axis,thick,-] (1,9)--(1,1);
\draw [thick,dotted] (1,6)--(3,4);
\draw [thick,dotted] (6,1)--(4,3);
\draw [thick,dotted] (3,4)--(3,9);
\draw [thick,dotted] (4,3)--(9,3);
\draw [thick] (1,6.5)--(2.75,3.75);
\draw [thick] (6.5,1)--(3.75,2.75);
\draw [thick] (2.75,3.75)--(2.75,9);
\draw [thick] (3.75,2.75)--(9,2.75);
\node [below] at (1,1) {\tiny$(c_1,c_2)$};
\node [above] at (7,9) {\tiny$(c_1+b_1,c_2+b_2)$};
\node [above] at (1,9) {\tiny$(c_1,c_2+b_2)$};
\draw[black,fill=black] (6,1) circle (1ex);
\node [below] at (6,1) {$p_{a_2}$};
\draw[black,fill=black] (1,6) circle (1ex);
\node [left] at (1,6) {$p_{a_1}$};
\draw[black,fill=black] (6.5,1) circle (1ex);
\node [below] at (7.5,1.25) {$p_{a_2}'$};
\draw[black,fill=black] (1,6.5) circle (1ex);
\node [left] at (1,7) {$p_{a_1}'$};
\draw[black,fill=black] (3,4) circle (1ex);
\node [right] at (3,4.5) {\footnotesize$P$};
\draw[black,fill=black] (4,3) circle (1ex);
\node [above] at (4.5,3) {\footnotesize$Q$};
\draw[black,fill=black] (2.75,3.75) circle (1ex);
\node [left] at (2.75,3.75) {\footnotesize$P'$};
\draw[black,fill=black] (3.75,2.75) circle (1ex);
\node [below] at (3.75,2.75) {\footnotesize$Q'$};
\node at (6,6) {$a_i$ increases};
\end{tikzpicture}}
\end{tabular}
\caption{(a) An illustration of observation \ref{obs:pa1-function} -- the values of $p_{a_i}$ determine the partition of the type space. (b) An illustration of observation \ref{obs:pa1-increase} -- an increase in $p_{a_i}$ to $p_{a_i}'$ increases $a_i$, and also moves the critical points towards south-west.}\label{fig:obs-1-2}
\end{figure}

We begin with the following observations.
\begin{observation}\label{obs:pa1-function}
The value of $p_{a_1}$ determines the value of $a_1$, the critical point $P$, and the outer boundary function $s_1(z_1)$. Similarly, the value of $p_{a_2}$ determines the value of $a_2$, the critical point $Q$, and the outer boundary function $s_2(z_2)$. This can be observed from (\ref{eqn:for-a1-a2}) and (\ref{eqn:P-Q-new}). See the illustration in Figure \ref{fig:obs-1}.
\end{observation}

So if $p_{a_i}$ are fixed such that $P_1+P_2=Q_1+Q_2$, then the partition of the type space is also fixed.
\begin{observation}\label{obs:pa1-increase}
Assume $p_{a_i}\in[(2b_{-i}-c_{-i})/3,b_{-i}]$. Then, an increase in $p_{a_1}$ increases $a_1$, and moves $P$ towards south-west (i.e., decreases both $P_1$ and $P_2$). Similarly, an increase in $p_{a_2}$ increases $a_2$, and moves $Q$ towards south-west. See the illustration in Figure \ref{fig:obs-2}.
\end{observation}

We now define $p_{a_i}^*:=\frac{2b_{-i}-c_{-i}}{3}-\frac{4c_i}{9}+\frac{2}{9}\sqrt{2c_i(2c_i+3(b_{-i}+c_{-i}))}$ and $r_i:=\frac{2b_{-i}(2b_i+5c_i)-c_{-i}(2b_i-3c_i)}{3(2b_i+3c_i)}$. Observe that $a_i=1$ when $p_{a_i}=p_{a_i}^*$, and that $p_{a_i}^*\in[(2b_{-i}-c_{-i})/3,b_{-i}]$. We consider four cases based on the values of $r_i$ and $p_{a_i}^*$.

\subsubsection{{\bf Case 1: $r_i\leq p_{a_i}^*$, $i=1,2$}}

In this case, we have $a_i\leq 1$ at $p_{a_i}=r_i$. We now make a series of claims, proofs of which are in \ref{app:a}.

\begin{enumerate}
 \item[1.] When $c_i,b_i$ satisfy (\ref{eqn:c1-c2-small}), then $r_i\in[(2b_{-i}-c_{-i})/3,b_{-i}]$, $i=1,2$. When $p_{a_i}=r_i$, we have $P=Q$, i.e., the critical points $P$ and $Q$ coincide. Moreover, $\bar{\mu}(W)\geq 0$.
 \item[2.] Fixing $p_{a_i}=r_i$, we now increase $p_{a_1}$, and adjust $p_{a_2}$ so that $P_1+P_2=Q_1+Q_2$ holds (i.e., (\ref{eqn:pa1-pa2}) is satisfied). Then (i) $p_{a_2}$ increases and (ii) $P$ lies to the north-west of $Q$ if $p_{a_i}\in[r_i,p_{a_i}^*]$.
 \item[3.] If $P$ is to the north-west of $Q$, then $\bar{\mu}(W)$ decreases with increase in $p_{a_i}$.
\end{enumerate}

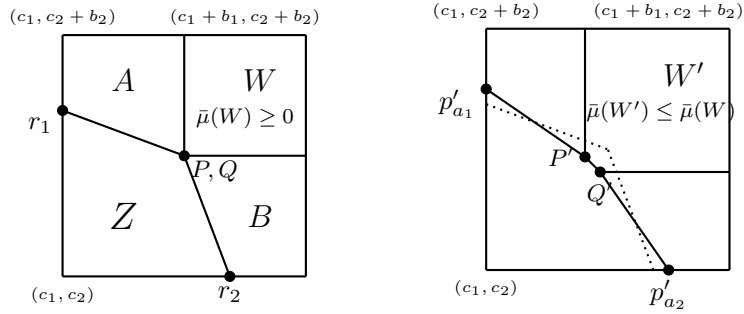
\begin{figure}[h!]
\centering
\begin{minipage}{.32\textwidth}
\centering
\begin{tikzpicture}[scale=0.4,font=\small,axis/.style={very thick, ->, >=stealth'}]
\draw [axis,thick,-] (1,1)--(9,1);
\draw [axis,thick,-] (9,1)--(9,9);
\draw [axis,thick,-] (9,9)--(1,9);
\draw [axis,thick,-] (1,9)--(1,1);
\draw [thick] (1,6.5)--(5,5);
\draw [thick] (6.5,1)--(5,5);
\draw [thick] (5,5)--(5,9);
\draw [thick] (5,5)--(9,5);
\node [below] at (1,1) {\tiny$(c_1,c_2)$};
\node [above] at (7,9) {\tiny$(c_1+b_1,c_2+b_2)$};
\node [above] at (1,9) {\tiny$(c_1,c_2+b_2)$};
\draw[black,fill=black] (6.5,1) circle (1ex);
\node [below] at (6.5,1) {$r_2$};
\draw[black,fill=black] (1,6.5) circle (1ex);
\node [left] at (1,6) {$r_1$};
\node at (3,3) {\Large$Z$};
\node at (3,7.5) {\large$A$};
\node at (7.5,3) {\large$B$};
\node at (7.5,7.5) {\large$W$};
\node at (7,6.25) {\scriptsize$\bar{\mu}(W)\geq 0$};
\draw[black,fill=black] (5,5) circle (1ex);
\node at (6,4.5) {\footnotesize$P,Q$};
\end{tikzpicture}
\end{minipage}
\begin{minipage}{.1\textwidth}
\hspace*{.1\textwidth}
\end{minipage}
\begin{minipage}{.32\textwidth}
\centering
\begin{tikzpicture}[scale=0.4,font=\small,axis/.style={very thick, ->, >=stealth'}]
\draw [axis,thick,-] (1,1)--(9,1);
\draw [axis,thick,-] (9,1)--(9,9);
\draw [axis,thick,-] (9,9)--(1,9);
\draw [axis,thick,-] (1,9)--(1,1);
\draw [thick,dotted] (1,6.5)--(5,5);
\draw [thick,dotted] (6.5,1)--(5,5);
\draw [thick] (1,7)--(4.25,4.75);
\draw [thick] (7,1)--(4.75,4.25);
\draw [thick] (4.25,4.75)--(4.75,4.25);
\draw [thick] (4.25,4.75)--(4.25,9);
\draw [thick] (4.75,4.25)--(9,4.25);
\node [below] at (1,1) {\tiny$(c_1,c_2)$};
\node [above] at (7,9) {\tiny$(c_1+b_1,c_2+b_2)$};
\node [above] at (1,9) {\tiny$(c_1,c_2+b_2)$};
\draw[black,fill=black] (7,1) circle (1ex);
\node [below] at (7,1) {$p_{a_2}'$};
\draw[black,fill=black] (1,7) circle (1ex);
\node [left] at (1,6.5) {$p_{a_1}'$};
\node at (7.5,7.5) {\large$W'$};
\node at (6.75,6.25) {\scriptsize$\bar{\mu}(W')\leq\bar{\mu}(W)$};
\draw[black,fill=black] (4.25,4.75) circle (1ex);
\node [left] at (4.25,4.75) {\footnotesize$P'$};
\draw[black,fill=black] (4.75,4.25) circle (1ex);
\node [below] at (4.75,4.25) {\footnotesize$Q'$};
\end{tikzpicture}
\end{minipage}
\caption{An illustration of claims 1--3. (L) If $p_{a_i}=r_i$, then by claim 1, we have $P=Q$, and $\bar{\mu}(W)\geq 0$; (R) If $p_{a_1}$ increases from $r_1$ to $p_{a_1}'$, and if $p_{a_2}$ is adjusted to satisfy (\ref{eqn:pa1-pa2}), then by claim 2, we have $p_{a_2}$ increasing from $r_2$ to $p_{a_2}'$ (say), and $P$ lying on the northwest of $Q$. By claim 3, we have $\bar{\mu}(W')\leq\bar{\mu}(W)$.}\label{fig:illust-sec5-1}
\end{figure}

Figure \ref{fig:illust-sec5-1} provides an illustration of claims 1--3. We continue to increase $p_{a_1}$ and adjust $p_{a_2}$ so that (\ref{eqn:pa1-pa2}) is satisfied. We stop when either (i) $\bar{\mu}(W)=0$ or (ii) $p_{a_2}=p_{a_2}^*$ or (iii) $p_{a_1}=p_{a_1}^*$. We have three cases based on which of them occurs first. We make the following claims regarding the structures of optimal mechanism in each of these cases. A summary of these claims is given in the form of a decision tree in Figure \ref{fig:dec-tree}.

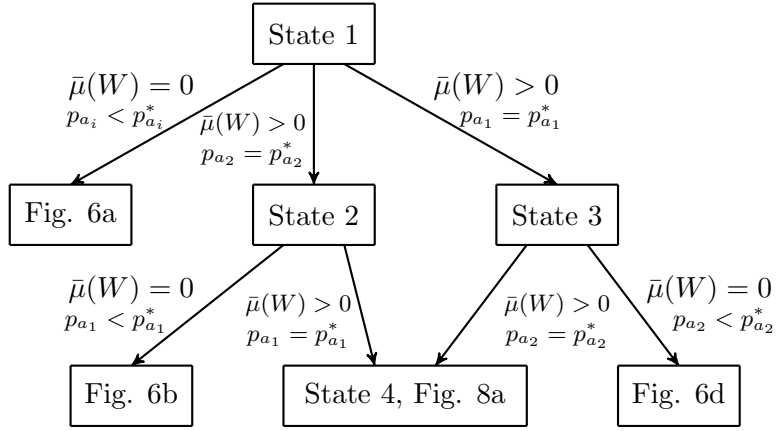
\begin{figure}[h!]
\centering
\begin{tikzpicture}[scale=0.4,font=\normalsize,axis/.style={very thick, ->, >=stealth'}]
\draw [axis,thick,-] (8,0)--(12,0);
\draw [axis,thick,-] (8,0)--(8,-2);
\draw [axis,thick,-] (8,-2)--(12,-2);
\draw [axis,thick,-] (12,-2)--(12,0);
\node at (10,-1) {State 1};
\draw [axis,thick,->] (9,-2)--(2,-6);
\node at (4,-2.75) {$\bar{\mu}(W)=0$};
\node at (3.5,-3.75) {\footnotesize$p_{a_i}<p_{a_i}^*$};
\draw [axis,thick,-] (0,-6)--(4,-6);
\draw [axis,thick,-] (0,-6)--(0,-8);
\draw [axis,thick,-] (0,-8)--(4,-8);
\draw [axis,thick,-] (4,-8)--(4,-6);
\node at (2,-7) {Fig. \ref{fig:a}};
\draw [axis,thick,->] (10,-2)--(10,-6);
\node at (8,-4) {\footnotesize$\bar{\mu}(W)>0$};
\node at (8,-5) {\footnotesize$p_{a_2}=p_{a_2}^*$};
\draw [axis,thick,-] (8,-6)--(12,-6);
\draw [axis,thick,-] (8,-6)--(8,-8);
\draw [axis,thick,-] (8,-8)--(12,-8);
\draw [axis,thick,-] (12,-8)--(12,-6);
\node at (10,-7) {State 2};
\draw [axis,thick,->] (11,-2)--(18,-6);
\node at (16,-2.75) {$\bar{\mu}(W)>0$};
\node at (16.5,-3.75) {\footnotesize$p_{a_1}=p_{a_1}^*$};
\draw [axis,thick,-] (16,-6)--(20,-6);
\draw [axis,thick,-] (16,-6)--(16,-8);
\draw [axis,thick,-] (16,-8)--(20,-8);
\draw [axis,thick,-] (20,-8)--(20,-6);
\node at (18,-7) {State 3};
\draw [axis,thick,->] (9,-8)--(4,-12);
\node at (4,-9.5) {$\bar{\mu}(W)=0$};
\node at (3.5,-10.5) {\footnotesize$p_{a_1}<p_{a_1}^*$};
\draw [axis,thick,-] (2,-12)--(6,-12);
\draw [axis,thick,-] (2,-12)--(2,-14);
\draw [axis,thick,-] (2,-14)--(6,-14);
\draw [axis,thick,-] (6,-14)--(6,-12);
\node at (4,-13) {Fig. \ref{fig:b}};
\draw [axis,thick,->] (11,-8)--(12,-12);
\node at (9.5,-10) {\footnotesize$\bar{\mu}(W)>0$};
\node at (9.5,-11) {\footnotesize$p_{a_1}=p_{a_1}^*$};
\draw [axis,thick,-] (9,-12)--(17,-12);
\draw [axis,thick,-] (9,-12)--(9,-14);
\draw [axis,thick,-] (9,-14)--(17,-14);
\draw [axis,thick,-] (17,-14)--(17,-12);
\node at (13,-13) {State 4, Fig. \ref{fig:d}};
\draw [axis,thick,->] (17,-8)--(14,-12);
\node at (18,-10) {\footnotesize$\bar{\mu}(W)>0$};
\node at (18,-11) {\footnotesize$p_{a_2}=p_{a_2}^*$};
\draw [axis,thick,->] (19,-8)--(22,-12);
\node at (23,-9.5) {$\bar{\mu}(W)=0$};
\node at (23.5,-10.5) {\footnotesize$p_{a_2}<p_{a_2}^*$};
\draw [axis,thick,-] (20,-12)--(24,-12);
\draw [axis,thick,-] (20,-12)--(20,-14);
\draw [axis,thick,-] (20,-14)--(24,-14);
\draw [axis,thick,-] (24,-14)--(24,-12);
\node at (22,-13) {Fig. \ref{fig:c}};
\end{tikzpicture}
\caption{The decision tree illustrating how to choose the structure of the optimal mechanism in Figure \ref{fig:gen-structure-1}. Case $i$ in Section 5.1.i starts from State $i$, $i=1,2,3,4$.}\label{fig:dec-tree}
\end{figure}
\begin{enumerate}
 \item[4.] Suppose that $\bar{\mu}(W)=0$ occurs when $p_{a_i}<p_{a_i}^*$. Then we claim that the optimal mechanism is as in Figure \ref{fig:a}. The proof is similar to the proof of Proposition \ref{prop:known}.
 \item[5.] Suppose instead that $p_{a_2}=p_{a_2}^*$ occurs first, i.e., we have $\bar{\mu}(W)>0$ and $p_{a_1}<p_{a_1}^*$ when $p_{a_2}$ equals $p_{a_2}^*$. Then, $a_2$ equals $1$, and we stop further increase in $p_{a_2}$. The structure in Figure \ref{fig:a} coincides with that in Figure \ref{fig:b}, and the regions $B$ and $W$ of the canonical partition merge. The picture is as in Figure \ref{fig:4a-4b-(a)}. As we increase $p_{a_1}$ to $p_{a_1}'$ further, we fix $p_{a_2}$ corresponding to the $p_{a_1}$ that led to $a_2=1$; by Observation \ref{obs:pa1-function}, the critical point $Q$ and the outer boundary function $s_2(z_2)$ also stand fixed. The partition thus moves from Figure \ref{fig:4a-4b-(a)} to \ref{fig:4a-4b-(b)}, when we increase $p_{a_1}$. Observe that the condition $P_1+P_2=Q_1+Q_2$ is no longer true, and thus we fix $p=P_1+P_2$ to complete the canonical partition. We now claim that increasing $p_{a_1}$ moves $P$ towards south-west, and also decreases $\bar{\mu}(W)$.
\end{enumerate}

Continue to increase $p_{a_1}$ until either (i) $\bar{\mu}(W)=0$, or (ii) $p_{a_1}=p_{a_1}^*$. We have two cases based on which of them occurs first.
\begin{enumerate}
 \item[6.] Suppose that $\bar{\mu}(W)=0$ when $p_{a_1}<p_{a_1}^*$, then the optimal mechanism is as in Figure \ref{fig:b}. The proof again, is similar to the proof of Proposition \ref{prop:known}. Equating $-\bar{\mu}(W)=0$ for Figure \ref{fig:b}, and substituting for $a_1$ and $p$, we have
 \begin{multline}\label{eqn:pa1-(b)}
   -8c_1b_2^2+(c_2b_1-b_2b_1-b_2c_1)(c_2-2b_2)\\+\left(c_2/2-b_1\right)(c_2-2b_2)^2-(3/8)(c_2-2b_2)^3\\+\left(c_1(4c_2-3b_2)+3b_1b_2+2(c_2-2c_1)(c_2-2b_2)-15/8(c_2-2b_2)^2\right)p_{a_1}\\+\left(3c_2/2-21/8(c_2-2b_2)\right)p_{a_1}^2-(9/8)p_{a_1}^3=0.
 \end{multline}
 \item[7.] Suppose instead that $p_{a_1}$ equals $p_{a_1}^*$, but still $\bar{\mu}(W)>0$. Then $a_1$ equals $1$, and we stop further increase in $p_{a_1}$. The regions $A$ and $W$ in the canonical partition merge, and the structure in Figure \ref{fig:b} coincides with that in Figure \ref{fig:c}. The value of $p$ equals $P_1+P_2=p_{a_1}^*$. We fix $p_{a_1}=p_{a_1}^*$, and by Observation \ref{obs:pa1-function}, the critical point $P$ and the outer boundary function $s_1(z_1)$ stand fixed. The partition thus moves from Figure \ref{fig:4b-4c-(a)} to \ref{fig:4b-4c-(b)}. We claim that a decrease in $p$ (from $p_{a_1}^*$) decreases $\bar{\mu}(W)$. Furthermore, it is easy to see that $\bar{\mu}(W)<0$ when $p=0^+$. Hence, by the continuity of $\bar{\mu}(W)$ in the $p$ parameter, there must be a $p$ when $\bar{\mu}(W)=0$. This happens when $p=p^*:=\left(\sqrt{(c_1+c_2)^2+6b_1b_2}-c_1-c_2\right)/3$. The resulting partition is optimal, and the mechanism is as in Figure \ref{fig:c}.
 \item[8.] If in contrast to what happens in point 5, suppose that $p_{a_1}=p_{a_1}^*$ occurs first, i.e., we have $\bar{\mu}(W)>0$ and $p_{a_2}<p_{a_2}^*$ when $p_{a_1}$ equals $p_{a_1}^*$. Then analogous arguments hold by symmetry, and the mechanism is as in Figure \ref{fig:f} instead of \ref{fig:b}.
\end{enumerate}
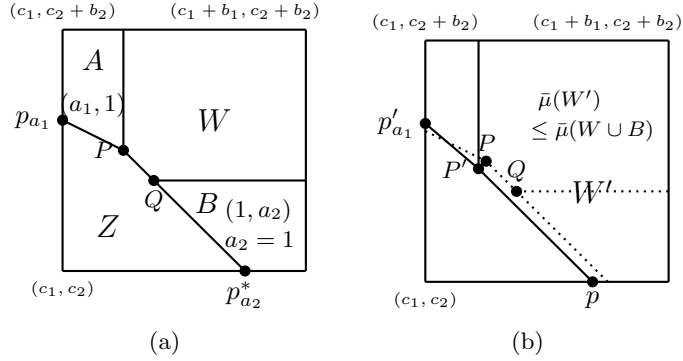
\begin{figure}[h!]
\centering
\begin{tabular}{cc}
\subfloat[]{\label{fig:4a-4b-(a)}\begin{tikzpicture}[scale=0.4,font=\small,axis/.style={very thick, ->, >=stealth'}]
\draw [axis,thick,-] (1,1)--(9,1);
\draw [axis,thick,-] (9,1)--(9,9);
\draw [axis,thick,-] (9,9)--(1,9);
\draw [axis,thick,-] (1,9)--(1,1);
\draw [thick] (1,6)--(3,5);
\draw [thick] (7,1)--(3,5);
\draw [thick] (3,5)--(3,9);
\draw [thick] (4,4)--(9,4);
\node [below] at (1,1) {\tiny$(c_1,c_2)$};
\node [above] at (7,9) {\tiny$(c_1+b_1,c_2+b_2)$};
\node [above] at (1,9) {\tiny$(c_1,c_2+b_2)$};
\draw[black,fill=black] (7,1) circle (1ex);
\node [below] at (7,1) {$p_{a_2}^*$};
\draw[black,fill=black] (1,6) circle (1ex);
\node [left] at (1,6) {$p_{a_1}$};
\node at (6,6) {\large$W$};
\node at (2.5,2.5) {\large$Z$};
\draw[black,fill=black] (3,5) circle (1ex);
\node [left] at (3,5) {\footnotesize$P$};
\draw[black,fill=black] (4,4) circle (1ex);
\node [below] at (4,4) {\footnotesize$Q$};
\node at (5.75,3.25) {\normalsize$B$};
\node [right] at (6,3) {\footnotesize$(1,a_2)$};
\node at (2,8) {\normalsize$A$};
\node at (2,6.5) {\footnotesize$(a_1,1)$};
\node [right] at (6,2) {\footnotesize$a_2=1$};
\end{tikzpicture}}&
\subfloat[]{\label{fig:4a-4b-(b)}\begin{tikzpicture}[scale=0.4,font=\small,axis/.style={very thick, ->, >=stealth'}]
\draw [axis,thick,-] (1,1)--(9,1);
\draw [axis,thick,-] (9,1)--(9,9);
\draw [axis,thick,-] (9,9)--(1,9);
\draw [axis,thick,-] (1,9)--(1,1);
\draw [thick] (1,6.25)--(2.75,4.75);
\draw [thick] (6.5,1)--(2.75,4.75);
\draw [thick,dotted] (1,6)--(3,5);
\draw [thick,dotted] (7,1)--(3,5);
\draw [thick] (2.75,4.75)--(2.75,9);
\draw [thick,dotted] (4,4)--(9,4);
\node [below] at (1,1) {\tiny$(c_1,c_2)$};
\node [above] at (7,9) {\tiny$(c_1+b_1,c_2+b_2)$};
\node [above] at (1,9) {\tiny$(c_1,c_2+b_2)$};
\draw[black,fill=black] (1,6.25) circle (1ex);
\node [left] at (1,6.25) {$p_{a_1}'$};
\draw[black,fill=black] (6.5,1) circle (1ex);
\node [below] at (6.5,1) {$p$};
\node at (6.5,4) {\large$W'$};
\node at (5.75,7) {\scriptsize$\bar{\mu}(W')$};
\node at (6.5,6) {\scriptsize$\leq\bar{\mu}(W\cup B)$};
\draw[black,fill=black] (2.75,4.75) circle (1ex);
\node [left] at (2.75,4.75) {\footnotesize$P'$};
\draw[black,fill=black] (3,5) circle (1ex);
\node [above] at (3,5) {\footnotesize$P$};
\draw[black,fill=black] (4,4) circle (1ex);
\node [above] at (4,4) {\footnotesize$Q$};
\end{tikzpicture}}
\end{tabular}
\caption{Illustration of the transition from Figure \ref{fig:a} to Figure \ref{fig:b}. If we increase $p_{a_1}$ to $p_{a_1}'$, then by claim 5, the point $P$ moves south-west, and the value of $\bar{\mu}(W)$ decreases.}\label{fig:4a-4b}
\end{figure}
\begin{figure}[h!]
\centering
\begin{tabular}{cc}
\subfloat[]{\label{fig:4b-4c-(a)}\begin{tikzpicture}[scale=0.4,font=\small,axis/.style={very thick, ->, >=stealth'}]
\draw [axis,thick,-] (1,1)--(9,1);
\draw [axis,thick,-] (9,1)--(9,9);
\draw [axis,thick,-] (9,9)--(1,9);
\draw [axis,thick,-] (1,9)--(1,1);
\draw [thick] (1,5)--(5,1);
\draw [thick,dotted] (7,1)--(4,4);
\draw [thick,dotted] (3,3)--(3,9);
\draw [thick,dotted] (4,4)--(9,4);
\node [below] at (1,1) {\tiny$(c_1,c_2)$};
\node [above] at (7,9) {\tiny$(c_1+b_1,c_2+b_2)$};
\node [above] at (1,9) {\tiny$(c_1,c_2+b_2)$};
\draw[black,fill=black] (1,5) circle (1ex);
\node [left] at (1,5) {$p_{a_1}*$};
\node at (6.5,4) {\large$W$};
\node at (2,2) {\large$Z$};
\node at (2,8) {\normalsize$A$};
\node at (2,6.5) {\footnotesize$(a_1,1)$};
\node at (2.25,5.5) {\footnotesize$a_1=1$};
\draw[black,fill=black] (3,3) circle (1ex);
\node [left] at (3,3) {\footnotesize$P$};
\draw[black,fill=black] (4,4) circle (1ex);
\node [above] at (4,4) {\footnotesize$Q$};
\end{tikzpicture}}&
\subfloat[]{\label{fig:4b-4c-(b)}\begin{tikzpicture}[scale=0.4,font=\small,axis/.style={very thick, ->, >=stealth'}]
\draw [axis,thick,-] (1,1)--(9,1);
\draw [axis,thick,-] (9,1)--(9,9);
\draw [axis,thick,-] (9,9)--(1,9);
\draw [axis,thick,-] (1,9)--(1,1);
\draw [thick,dotted] (1,5)--(3,3);
\draw [thick] (1,4.5)--(4.5,1);
\draw [thick,dotted] (7,1)--(4,4);
\draw [thick,dotted] (3,3)--(3,9);
\draw [thick,dotted] (4,4)--(9,4);
\node [below] at (1,1) {\tiny$(c_1,c_2)$};
\node [above] at (7,9) {\tiny$(c_1+b_1,c_2+b_2)$};
\node [above] at (1,9) {\tiny$(c_1,c_2+b_2)$};
\draw[black,fill=black] (1,4.5) circle (1ex);
\node [left] at (1,4.5) {$p*$};
\node at (3,6) {\large$W'$};
\node at (5.75,7) {\scriptsize$\bar{\mu}(W')$};
\node at (6.5,6) {\scriptsize$\leq\bar{\mu}(W\cup A)$};
\draw[black,fill=black] (3,3) circle (1ex);
\node [right] at (3,3) {\footnotesize$P$};
\draw[black,fill=black] (4,4) circle (1ex);
\node [above] at (4,4) {\footnotesize$Q$};
\end{tikzpicture}}
\end{tabular}
\caption{Illustration of the transition from Figure \ref{fig:b} to Figure \ref{fig:c}. If we decrease $p_{a_1}$, then by claim 7, $\bar{\mu}(W)$ decreases.}\label{fig:4b-4c}
\end{figure}
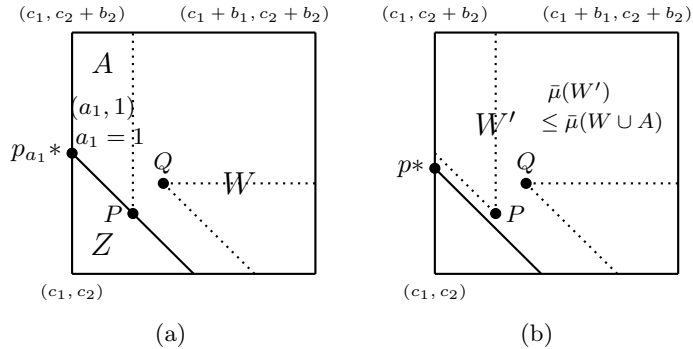

The above claims prove that whenever $r_i\leq p_{a_i}^*$, $i=1,2$, the solution $(p_{a_1},p_{a_2})$ obtained by solving (\ref{eqn:pa1-pa2}) and (\ref{eqn:big-W}) gives rise to a valid partition of the support set $D$. Furthermore, the structure of the optimal mechanism turns out to be one of the four structures depicted in Figure \ref{fig:gen-structure-1}, based on whether $a_i$ is less than $1$ or not.

\subsubsection{{\bf Case 2: $r_1\leq p_{a_1}^*$ but $r_2>p_{a_2}^*$}}

In this case, we have $a_1\leq 1$ and $a_2>1$ at $p_{a_i}=r_i$. Thus the technique of increasing both $p_{a_i}$ (from $r_i$) does not help here, because by Observation \ref{obs:pa1-increase}, increasing $p_{a_2}$ only increases $a_2$ further. So we use an alternate method. We fix $p_{a_2}=r_2$ (and thus fixing the corresponding $Q$ and $s_2(z_2)$), but change the line joining $Q$ and $p_{a_2}$ to a line of slope $-1$, as illustrated in Figure \ref{fig:case2}. In other words, we modify the structure, which was as in Figure \ref{fig:a} with $p_{a_i}=r_i$, to that in Figure \ref{fig:b} with $p_{a_1}=r_1$ and $p=P_1+P_2$. We claim that $\bar{\mu}(W)\geq 0$ for this case. This can be verified as follows. Substituting $p_{a_1}=r_1$ in (\ref{eqn:pa1-(b)}), the LHS$=-\bar{\mu}(W)$, becomes
\begin{multline*}
  -\bar{\mu}(W)=-(2b_2(b_1+c_1)-c_2(b_1+3 c_1))\\(b_2(2b_1+3c_1)^2+8b_2^2(b_1+c_1)+16b_1b_2c_2+12b_2c_1c_2+6b_1c_2^2)
\end{multline*}
which is nonpositive when $c_2\leq2b_2\frac{b_1+c_1}{b_1+3c_1}$. Thus for $c_i,b_i$ satisfying (\ref{eqn:c1-c2-small}), $\bar{\mu}(W)\geq 0$ holds when $p_{a_1}=r_1$.

We now increase $p_{a_1}$ starting from $r_1$. Notice that the case is now similar to Case 1, from bullet 6 onwards. The optimal mechanism is as in Figure \ref{fig:b} if $\bar{\mu}(W)=0$ occurs when $p_{a_1}\leq p_{a_1}^*$, and is as in Figure \ref{fig:c} otherwise. Thus, when we start from State 2 in the decision tree in Figure \ref{fig:dec-tree}, we obtain the procedure to find the optimal mechanism as mentioned above.
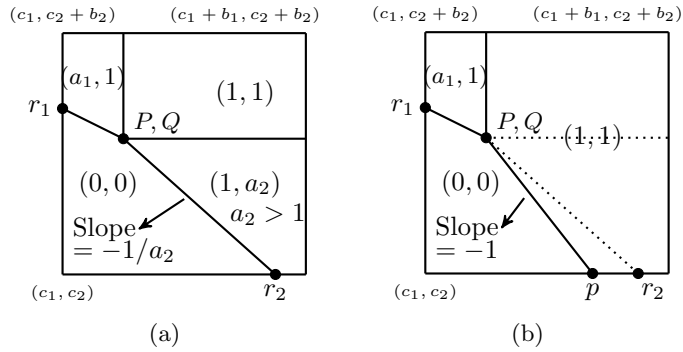
\begin{figure}[h!]
\centering
\begin{tabular}{ccc}
\subfloat[]{\label{fig:case2-(a)}\begin{tikzpicture}[scale=0.4,font=\small,axis/.style={very thick, ->, >=stealth'}]
\draw [axis,thick,-] (1,1)--(9,1);
\draw [axis,thick,-] (9,1)--(9,9);
\draw [axis,thick,-] (9,9)--(1,9);
\draw [axis,thick,-] (1,9)--(1,1);
\draw [thick] (1,6.5)--(3,5.5);
\draw [thick] (3,5.5)--(8,1);
\draw [thick] (3,5.5)--(3,9);
\draw [thick] (3,5.5)--(9,5.5);
\node [below] at (1,1) {\tiny$(c_1,c_2)$};
\node [above] at (7,9) {\tiny$(c_1+b_1,c_2+b_2)$};
\node [above] at (1,9) {\tiny$(c_1,c_2+b_2)$};
\draw[black,fill=black] (8,1) circle (1ex);
\node [below] at (8,1) {$r_2$};
\draw[black,fill=black] (1,6.5) circle (1ex);
\node [left] at (1,6.5) {$r_1$};
\node at (7,7) {$(1,1)$};
\node at (2.5,4) {$(0,0)$};
\draw[black,fill=black] (3,5.5) circle (1ex);
\node [right] at (3,6) {\footnotesize$P,Q$};
\node [right] at (1,2.5) {Slope};
\node [right] at (1,1.75) {$=-1/a_2$};
\draw [axis,thick,->] (5,3.5)--(3.5,2.5);
\node at (2,7.5) {\footnotesize$(a_1,1)$};
\node at (7,4) {$(1,a_2)$};
\node at (7.75,3) {$a_2>1$};
\end{tikzpicture}}&
\subfloat[]{\label{fig:case2-(b)}\begin{tikzpicture}[scale=0.4,font=\small,axis/.style={very thick, ->, >=stealth'}]
\draw [axis,thick,-] (1,1)--(9,1);
\draw [axis,thick,-] (9,1)--(9,9);
\draw [axis,thick,-] (9,9)--(1,9);
\draw [axis,thick,-] (1,9)--(1,1);
\draw [thick] (1,6.5)--(3,5.5);
\draw [thick,dotted] (3,5.5)--(8,1);
\draw [thick,dotted] (3,5.5)--(9,5.5);
\draw [thick] (3,5.5)--(6.5,1);
\draw [thick] (3,5.5)--(3,9);
\node [below] at (1,1) {\tiny$(c_1,c_2)$};
\node [above] at (7,9) {\tiny$(c_1+b_1,c_2+b_2)$};
\node [above] at (1,9) {\tiny$(c_1,c_2+b_2)$};
\draw[black,fill=black] (8,1) circle (1ex);
\node [below] at (8.5,1) {$r_2$};
\draw[black,fill=black] (6.5,1) circle (1ex);
\node [below] at (6.5,1) {$p$};
\draw[black,fill=black] (1,6.5) circle (1ex);
\node [left] at (1,6.5) {$r_1$};
\node at (6.5,5.5) {$(1,1)$};
\node at (2.5,4) {$(0,0)$};
\draw[black,fill=black] (3,5.5) circle (1ex);
\node [right] at (3,6) {\footnotesize$P,Q$};
\node [right] at (1,2.5) {Slope};
\node [right] at (1,1.75) {$=-1$};
\draw [axis,thick,->] (4.25,3.5)--(3.5,2.5);
\node at (2,7.5) {\footnotesize$(a_1,1)$};
\end{tikzpicture}}
\end{tabular}
\caption{An illustration of change of slope in Case 2. The structure is as in Figure \ref{fig:a} in (a), and as in Figure \ref{fig:b} in (b).}\label{fig:case2}
\end{figure}

\subsubsection{{\bf Case 3: $r_2\leq p_{a_2}^*$ but $r_1>p_{a_1}^*$}}

This case is the same as Case 2, except that the optimal mechanism is as in Figure \ref{fig:f} instead of Figure \ref{fig:b}. We start from State 3 in the decision tree in Figure \ref{fig:dec-tree} to obtain the procedure to find the optimal mechanism.

\subsubsection{{\bf Case 4: $r_i>p_{a_i}^*$, $i=1,2$}}

Fixing $p_{a_i}=r_i$ yields $a_i>1$. Thus the technique of increasing $p_{a_i}$ from $r_i$ is not useful here because, by Observation \ref{obs:pa1-increase}, increasing $p_{a_i}$ only increases $a_i$ further. We thus use an alternate method. We first assume $p_{a_1}^*\leq p_{a_2}^*$, and do the following.  (The argument is symmetric when $p_{a_2}^*<p_{a_1}^*$.)

\begin{itemize}
 \item We force the line joining $Q$ and $p_{a_2}$ to have a slope of $-1$, as in Case 2. See the illustration in Figure \ref{fig:case-4-(a)}. Then, $\bar{\mu}(W)$ remains nonnegative by the same arguments in Case 2. Also, the structure obtained is that depicted in Figure \ref{fig:b}, with $p_{a_1}=r_1$ and $p=P_1+P_2$. But we have $a_1>1$ because $p_{a_1}>p_{a_1}^*$.
 \item We now fix $p_{a_2}=r_2$ (which fixes the corresponding $Q$ and $s_2(z_2)$), and decrease $p_{a_1}$ from $r_1$ to $p_{a_1}^*$. See the illustration in Figure \ref{fig:case-4-(b)}. This decrease moves $P$ towards north-east, decreases $a_1$ to $1$, and increases $\bar{\mu}(W)$. The first two assertions in the previous sentence follow from Observation \ref{obs:pa1-increase} and the last assertion can be shown using a proof similar to the proof of claim 3 in Case 1. Observe that the line $z_1+z_2=c_1+c_2+p_{a_1}^*$ remains within $D$ because $p_{a_1}^*\leq p_{a_2}^*\leq b_1$. The structure obtained is that depicted in Figure \ref{fig:c}, with $p=p_{a_1}^*$.
 \item We now fix $p_{a_1}=p_{a_1}^*$ (thus fixing the corresponding $P$ and $s_1(z_1)$), and decrease $p$ from $p_{a_1}^*$ to $p^*$. Then, the partition moves from Figure \ref{fig:case-4-(b)} to \ref{fig:case-4-(c)}. Furthermore, by the same argument as in claim 7, Case 1, the value of $\bar{\mu}(W)$ decreases, and equals $0$ at $p=p^*$.
\end{itemize}

\begin{figure}[h!]
\centering
\begin{tabular}{ccc}
\subfloat[]{\label{fig:case-4-(a)}\begin{tikzpicture}[scale=0.32,font=\footnotesize,axis/.style={very thick, ->, >=stealth'}]
\draw [axis,thick,-] (1,1)--(9,1);
\draw [axis,thick,-] (9,1)--(9,9);
\draw [axis,thick,-] (9,9)--(1,9);
\draw [axis,thick,-] (1,9)--(1,1);
\draw [thick] (1,8)--(3.5,3.5);
\draw [thick] (3.5,3.5)--(8,1);
\draw [thick,dotted] (3.5,3.5)--(6,1);
\draw [thick] (3.5,3.5)--(3.5,9);
\draw [thick] (3.5,3.5)--(9,3.5);
\node [below] at (1,1) {\tiny$(c_1,c_2)$};
\node [above] at (7,9) {\tiny$(c_1+b_1,c_2+b_2)$};
\node [above] at (1,9) {\tiny$(c_1,c_2+b_2)$};
\draw[black,fill=black] (8,1) circle (1ex);
\node [below] at (8,1) {$r_2$};
\draw[black,fill=black] (6,1) circle (1ex);
\node [below] at (6,1) {$p$};
\draw[black,fill=black] (1,8) circle (1ex);
\node [left] at (1,8) {$r_1$};
\node at (6,6) {\large$W$};
\node at (2,3.5) {\large$Z$};
\draw[black,fill=black] (3.5,3.5) circle (1ex);
\node [right] at (3.5,4) {\footnotesize$P,Q$};
\node [right] at (1.25,2) {\tiny Slope};
\node [right] at (1.25,1.5) {\tiny$=-1$};
\draw [axis,thick,->] (4.75,2)--(3.75,1.5);
\node [right] at (5.5,3) {\tiny Slope$=$};
\node [right] at (5.75,2.25) {\tiny$-1/a_2$};
\node [right] at (6.75,1.5) {\tiny$<-1$};
\draw [axis,thick,->] (5.25,2.5)--(5.75,3);
\node at (2.75,6.5) {\large$A$};
\end{tikzpicture}}&
\subfloat[]{\label{fig:case-4-(b)}\begin{tikzpicture}[scale=0.32,font=\footnotesize,axis/.style={very thick, ->, >=stealth'}]
\draw [axis,thick,-] (1,1)--(9,1);
\draw [axis,thick,-] (9,1)--(9,9);
\draw [axis,thick,-] (9,9)--(1,9);
\draw [axis,thick,-] (1,9)--(1,1);
\draw [thick,dotted] (1,8)--(3.5,3.5);
\draw [thick,dotted] (3.5,3.5)--(8,1);
\draw [thick] (4,4)--(4,9);
\draw [thick,dotted] (3.5,3.5)--(9,3.5);
\draw [thick] (7,1)--(1,7);
\node [below] at (1,1) {\tiny$(c_1,c_2)$};
\node [above] at (7,9) {\tiny$(c_1+b_1,c_2+b_2)$};
\node [above] at (1,9) {\tiny$(c_1,c_2+b_2)$};
\draw[black,fill=black] (1,7) circle (1ex);
\node [left] at (1,7) {$p_{a_1}^*$};
\draw[black,fill=black] (1,8) circle (1ex);
\node [left] at (1,8) {$r_1$};
\draw[black,fill=black] (8,1) circle (1ex);
\node [below] at (8,1) {$r_2$};
\node at (6,6) {\large$W'$};
\node at (2,3.5) {\large$Z'$};
\draw[black,fill=black] (4,4) circle (1ex);
\node [right] at (4,4) {\footnotesize$P$};
\draw[black,fill=black] (3.5,3.5) circle (1ex);
\node [below] at (3.5,3.5) {\footnotesize$Q$};
\node at (3,6.5) {\large$A'$};
\end{tikzpicture}}&
\subfloat[]{\label{fig:case-4-(c)}\begin{tikzpicture}[scale=0.32,font=\footnotesize,axis/.style={very thick, ->, >=stealth'}]
\draw [axis,thick,-] (1,1)--(9,1);
\draw [axis,thick,-] (9,1)--(9,9);
\draw [axis,thick,-] (9,9)--(1,9);
\draw [axis,thick,-] (1,9)--(1,1);
\draw [thick,dotted] (3.5,3.5)--(8,1);
\draw [thick,dotted] (4,4)--(4,9);
\draw [thick,dotted] (3.5,3.5)--(9,3.5);
\draw [thick,dotted] (1,7)--(4,4);
\draw [thick] (1,4.5)--(4.5,1);
\node [below] at (1,1) {\tiny$(c_1,c_2)$};
\node [above] at (7,9) {\tiny$(c_1+b_1,c_2+b_2)$};
\node [above] at (1,9) {\tiny$(c_1,c_2+b_2)$};
\draw[black,fill=black] (1,7) circle (1ex);
\node [left] at (1,7) {$p_{a_1}^*$};
\draw[black,fill=black] (1,4.5) circle (1ex);
\node [left] at (1,4.5) {$p^*$};
\draw[black,fill=black] (8,1) circle (1ex);
\node [below] at (8,1) {$r_2$};
\node at (6,6) {\large$W''$};
\node at (2,2) {\large$Z''$};
\draw[black,fill=black] (4,4) circle (1ex);
\node [right] at (4,4) {\footnotesize$P$};
\draw[black,fill=black] (3.5,3.5) circle (1ex);
\node [below] at (3.5,3.5) {\footnotesize$Q$};
\end{tikzpicture}}\\
\subfloat[]{\label{fig:case-4-(d)}\begin{tikzpicture}[scale=0.32,font=\footnotesize,axis/.style={very thick, ->, >=stealth'}]
\draw [axis,thick,-] (1,1)--(9,1);
\draw [axis,thick,-] (9,1)--(9,9);
\draw [axis,thick,-] (9,9)--(1,9);
\draw [axis,thick,-] (1,9)--(1,1);
\draw [thick,dotted] (3.5,3.5)--(8,1);
\draw [thick,dotted] (4,4)--(4,9);
\draw [thick,dotted] (3.5,3.5)--(9,3.5);
\draw [thick,dotted] (1,7)--(4,4);
\draw [thick] (1,6.5)--(6.5,1);
\node [below] at (1,1) {\tiny$(c_1,c_2)$};
\node [above] at (7,9) {\tiny$(c_1+b_1,c_2+b_2)$};
\node [above] at (1,9) {\tiny$(c_1,c_2+b_2)$};
\draw[black,fill=black] (1,7) circle (1ex);
\node [left] at (1,7.5) {$p_{a_1}^*$};
\draw[black,fill=black] (1,6.5) circle (1ex);
\node [left] at (1,6) {$p^*$};
\draw[black,fill=black] (8,1) circle (1ex);
\node [below] at (8,1) {$r_2$};
\draw[black,fill=black] (4,4) circle (1ex);
\node [right] at (4,4) {\footnotesize$P$};
\draw[black,fill=black] (3.5,3.5) circle (1ex);
\node [below] at (3.5,3.5) {\footnotesize$Q$};
\end{tikzpicture}}&
\subfloat[]{\label{fig:case-4-(e)}\begin{tikzpicture}[scale=0.32,font=\footnotesize,axis/.style={very thick, ->, >=stealth'}]
\draw [axis,thick,-] (1,1)--(9,1);
\draw [axis,thick,-] (9,1)--(9,9);
\draw [axis,thick,-] (9,9)--(1,9);
\draw [axis,thick,-] (1,9)--(1,1);
\draw [thick,dotted] (3.5,5)--(7.5,1);
\draw [thick,dotted] (4,4)--(4,9);
\draw [thick,dotted] (3.5,5)--(9,5);
\draw [thick,dotted] (1,7)--(4,4);
\draw [thick] (1,6.5)--(6.5,1);
\node [below] at (1,1) {\tiny$(c_1,c_2)$};
\node [above] at (7,9) {\tiny$(c_1+b_1,c_2+b_2)$};
\node [above] at (1,9) {\tiny$(c_1,c_2+b_2)$};
\draw[black,fill=black] (1,7) circle (1ex);
\node [left] at (1,7.5) {$p_{a_1}^*$};
\draw[black,fill=black] (1,6.5) circle (1ex);
\node [left] at (1,6) {$p^*$};
\draw[black,fill=black] (7.5,1) circle (1ex);
\node [below] at (7.5,1) {$p_{a_2}^*$};
\draw[black,fill=black] (4,4) circle (1ex);
\node [below] at (4,4) {\footnotesize$P$};
\draw[black,fill=black] (3.5,5) circle (1ex);
\node [above] at (3.5,5) {\footnotesize$Q$};
\end{tikzpicture}}&
\subfloat[]{\label{fig:case-4-(f)}\begin{tikzpicture}[scale=0.32,font=\footnotesize,axis/.style={very thick, ->, >=stealth'}]
\draw [axis,thick,-] (1,1)--(9,1);
\draw [axis,thick,-] (9,1)--(9,9);
\draw [axis,thick,-] (9,9)--(1,9);
\draw [axis,thick,-] (1,9)--(1,1);
\draw [thick,dotted] (4.25,4.25)--(7.5,1);
\draw [thick,dotted] (4,4)--(4,9);
\draw [thick,dotted] (4.25,4.25)--(9,4.25);
\draw [thick,dotted] (1,7)--(4,4);
\draw [thick] (1,6.5)--(6.5,1);
\node [below] at (1,1) {\tiny$(c_1,c_2)$};
\node [above] at (7,9) {\tiny$(c_1+b_1,c_2+b_2)$};
\node [above] at (1,9) {\tiny$(c_1,c_2+b_2)$};
\draw[black,fill=black] (1,7) circle (1ex);
\node [left] at (1,7.5) {$p_{a_1}^*$};
\draw[black,fill=black] (1,6.5) circle (1ex);
\node [left] at (1,6) {$p^*$};
\draw[black,fill=black] (7.5,1) circle (1ex);
\node [below] at (7.5,1) {$p_{a_2}^*$};
\draw[black,fill=black] (4,4) circle (1ex);
\node [left] at (3.5,4) {\footnotesize$P$};
\draw[black,fill=black] (4.25,4.25) circle (1ex);
\node [above] at (4.75,4.25) {\footnotesize$Q$};
\end{tikzpicture}}
\end{tabular}
\caption{An illustration of the steps enumerated in Case 4. (a) We force the line joining $Q$ and $r_2$ to have a slope of $-1$; (b) Fixing $p_{a_2}=r_2$, we decrease $p_{a_1}$ from $r_1$ to $p_{a_1}^*$; (c) Fixing $p_{a_1}=r_1$, we decrease $p$ from $p_{a_1}^*$ to $p^*$; The point $Q$ may lie either outside $Z$ as in (c), or inside $Z$ as in (d). In case it is as in (d), we decrease $p_{a_2}$ from $r_2$ to $p_{a_2}^*$. When we do so, the lines $s_1(z_1)$ and $s_2(z_2)$ may either overlap as in (e), or, may remain separate as in (f).}\label{fig:case-4}
\end{figure}
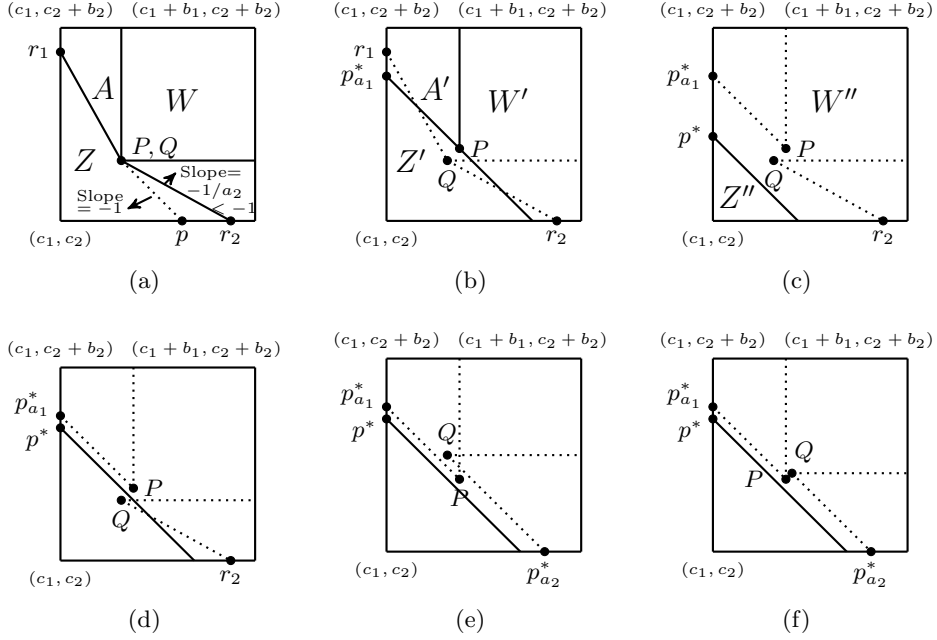

We thus obtain a structure which is as in Figure \ref{fig:c}. But to prove the optimality of the structure, we require the critical points $P$ and $Q$ to be above the line $z_1+z_2=c_1+c_2+p^*$. Observe that the point $Q$ may possibly lie below the line $z_1+z_2=c_1+c_2+p^*$, as illustrated in Figure \ref{fig:case-4-(d)}. In case that happens, we modify $Q$ and the corresponding $s_2(z_2)$ as follows. We decrease $p_{a_2}$ from $r_2$ to $p_{a_2}^*$. The point $Q$ thus moves towards north-east, and is on or above the line $z_1+z_2=c_1+c_2+p^*$, since $p^*\leq p_{a_1}^*\leq p_{a_2}^*$.

Now the partitions with respect to $p_{a_1}^*$ and $p_{a_2}^*$ may overlap (as illustrated in Figure \ref{fig:case-4-(e)}). The overlap occurs whenever $Q$ lies to the north-west of $P$. In case of no overlap (as illustrated in Figure \ref{fig:case-4-(f)}), it can be argued that the optimal mechanism is as in Figure \ref{fig:c}, by constructing $\gamma$ in the same way as in the proof of Proposition \ref{prop:known}. In case of an overlap, we construct $\gamma$ (on $W\times W$) as follows. We first fix $\theta:=\bar{\mu}_+^W+\bar{\alpha}$, $\bar{\alpha}$ constructed with $p_{a_i}=p_{a_i}^*$. The following lemma shows that if $\theta$ satisfies certain conditions, then there exists a $\gamma$ on $W\times W$ that satisfies the dual constraint and the conditions in Lemma \ref{lem:conditions}.
\begin{lemma}\label{lem:DDT-equiv}
Consider a measure $\theta$ on $W$ that satisfies the following conditions:
\begin{equation}\label{eqn:theta-measure}
  (i)\,\theta\succeq_1\bar{\mu}_-^W;\,(ii)\,\theta\succeq_{cvx}\bar{\mu}_+^W;\,(iii)\int_W\|x\|_1\,d\theta=\int_W\|x\|_1\,d\bar{\mu}_+^W.
\end{equation}
Then, there exists a measure $\gamma$ on $W\times W$ such that
\begin{multline*}
  (a)\,\gamma_1-\gamma_2\succeq_{cvx}\bar{\mu}^W;\,(b)\int_Wu\,d(\gamma_1-\gamma_2-\bar{\mu}^W)=0;\\(c)\,u(z)-u(z')=\|z-z'\|_1,\gamma-a.e.
\end{multline*}
\end{lemma}

Let $\theta = \bar{\mu}_+^W + \bar{\alpha}$. We now show that $\theta$ satisfies the conditions in Lemma \ref{lem:DDT-equiv}. Conditions (ii) and (iii) only involve verifying if $\bar{\alpha}\succeq_{cvx}0$ and $\int_D\|x\|_1\,d\bar{\alpha}\geq 0$, both of which are established in Proposition \ref{prop:cvx}. The following lemma proves that $\theta$ satisfies condition (i).
\begin{lemma}\label{lem:gc1-W}
Consider the structure in Figure \ref{fig:c}. Let $s_i(z_i)$ with respect to $\bar{\mu}+\bar{\alpha}$ satisfy $s_i(z_i)>\max(c_{-i},c_1+c_2+p^*-z_i)$ for every $z_i\in[c_i,c_i+b_i)$. Then, the following hold:
\begin{enumerate}
 \item[(a)] $(\bar{\mu}+\bar{\alpha})(W\cap\{z_1\geq c_1+t_1\}\cap\{z_2\geq c_2+t_2\})\geq 0$ for every $t_i\in[0,b_i]$.
 \item[(b)] $\theta-\bar{\mu}_-^W=\bar{\mu}^W+\bar{\alpha}\succeq_1 0.$
\end{enumerate}
\end{lemma}

This shows that the dual variable $\gamma$ can be constructed even in the case of an overlap, and so the optimal mechanism is as in Figure \ref{fig:c}, whenever $r_i>p_{a_i}^*$. Thus, when we start from State 4 in the decision tree in Figure \ref{fig:dec-tree}, we obtain the optimal mechanism as in Figure \ref{fig:c}.

We have thus computed the optimal mechanism for all $c_i,b_i$ that satisfy (\ref{eqn:c1-c2-small}). The results are summarized by the following theorem.

\begin{theorem}\label{thm:str-1}
If $c_i, b_i$ satisfy (\ref{eqn:c1-c2-small}), then the optimal mechanism has one of the four structures depicted in Figure \ref{fig:gen-structure-1}. The exact structure and the values of $(p_{a_i},a_i,p)$ are given as follows.
\begin{enumerate}
 \item[(a)] The optimal mechanism is as in Figure \ref{fig:a}, if there exists $p_{a_i}\in[r_i,p_{a_i}^*]$, $i=1,2$, solving (\ref{eqn:pa1-pa2}) and (\ref{eqn:big-W}) simultaneously. The values of $(a_i,p)$ are found from (\ref{eqn:for-a1-a2}) and (\ref{eqn:for-p1}).
 \item[(b)] The optimal mechanism is as in Figure \ref{fig:b}, if (i) the condition in (a) does not hold, and (ii) there exists $p_{a_1}\in[r_1,p_{a_1}^*]$ that solves (\ref{eqn:pa1-(b)}). The values of $(a_1,p)$ are found from (\ref{eqn:for-a1-a2}) and (\ref{eqn:for-p1}). 
 
 The optimal mechanism is as in Figure \ref{fig:f}, if (i) and (ii) above hold with the roles of the indices 1 and 2 swapped. The values of $(a_2,p)$ are found analogously.
 \item[(c)] The optimal mechanism is as in Figure \ref{fig:c} with $p=p^*$, if conditions in $(a)-(b)$ do not hold.
\end{enumerate}
\end{theorem}

\subsection{Optimal mechanisms when $\frac{c_1}{b_1}$ is small but $\frac{c_2}{b_2}$ is large}\label{SUB:GC2}
Consider the case when $c_1\leq b_1$, but $c_2>2b_2(b_1+c_1)/(b_1+3c_1)$. Notice that $c_2>b_2$ for any $c_1\leq b_1$, and thus this case violates (\ref{eqn:c1-c2-small}). When $c_2$ crosses $2b_2(b_1+c_1)/(b_1+3c_1)$, we claim that the critical point $P$ moves outside the support $D$. From (\ref{eqn:P-Q-new}), we have $P_2\geq c_2$ iff $p_{a_1}\leq 2b_2-c_2$. So to prove the preceding claim, it suffices to prove that when $\{c_1\leq b_1,c_2>2b_2(b_1+c_1)/(b_1+3c_1)\}$, the solution in the interval $[(2b_2-c_2)/3,b_2]$ that solves (\ref{eqn:pa1-(b)}) is at least $2b_2-c_2$.

Substituting $p_{a_1}=2b_2-c_2$ in the left-hand side of (\ref{eqn:pa1-(b)}), we get $2b_2(2b_2(b_1+3c_1)-c_2(b_1+c_1))<0$ when $c_2>2b_2(b_1+c_1)/(b_1+3c_1)$. Substituting $p_{a_1}=(2b_2-c_2)/3$, we get $-\frac{3}{8}b_2c_1(b_2+c_2)<0$, and substituting $p_{a_1}=b_2$, we have $(b_2+c_2)(8b_2(b_1-c_1)+(c_2-b_2)(c_2+3b_2))/8>0$ when $c_1\leq b_1$ and $c_2\geq b_2$. So the solution for (\ref{eqn:pa1-(b)}) in the interval $[(2b_2-c_2)/3,b_2]$ is at least $2b_2-c_2$. We have proved our claim.

At $c_2=2b_2(b_1+c_1)/(b_1+3c_1)$, we have $P_2=c_2$, i.e., $P$ touches the bottom boundary of $D$. The structure in Figure \ref{fig:b} coincides with that in Figure \ref{fig:d}. We anticipate that adding on the top-left boundary a shuffling measure whose density is linear in $[c_1,c_1+p_{a_1}/a_1]$ and is a constant in $[c_1+p_{a_1}/a_1,c_1+p]$ would yield the structure in Figure \ref{fig:d}. We do the following.

\begin{itemize}
\item We construct a shuffling measure $\bar{\beta}$ with parameters $p_{a_1}$, $a_1$, and $p$, and find the conditions on the parameters for $\bar{\beta}\succeq_{cvx}0$ to hold.
\item We then identify the canonical partition associated with $\bar{\mu}+\bar{\beta}$. Observe that it only involves computing $s_1(\cdot)$, since we do not have region $B$ in Figure \ref{fig:d}.
\item To prove the validity of the partition, we (i) analyze how the parameters $a_1$ and $p$ vary when $p_{a_1}$ varies, and then (ii) fix $p_{a_1}=b_2$ and decrease it until we get $p=(b_1-c_1)/2$.
\item For $c_2\in[2b_2(b_1+c_1)/(b_1+3c_1),2b_2(b_1/(b_1-c_1))^2]$, we prove that if $a_1\leq 1$ at the stopping point, then the optimal mechanism is as in Figure \ref{fig:d}; otherwise it is as in Figure \ref{fig:c'}.
\item For $c_2\geq2b_2(b_1/(b_1-c_1))^2$, we proceed as follows. We first prove that if $c_2=2b_2(b_1/(b_1-c_1))^2$, then $p=(b_1-c_1)/2$ occurs when $p_{a_1}=a_1=0$. Then we proceed to prove that the optimal mechanism is as in Figure \ref{fig:e} whenever $c_2\geq2b_2(b_1/(b_1-c_1))^2$.
\end{itemize}

We now split this case into two subcases: (i) $c_2\in[2b_2(b_1+c_1)/(b_1+3c_1),2b_2(b_1/(b_1-c_1))^2]$, i.e., $\frac{c_2}{b_2}$ has large but not too large values, and (ii) $c_2\geq2b_2(b_1/(b_1-c_1))^2$, i.e., $\frac{c_2}{b_2}$ is very large.

\subsubsection{Optimal mechanisms when $\frac{c_1}{b_1}$ is small and $\frac{c_2}{b_2}$ has large but not too large values}
\label{subsubsec:small-intermediate}

Let $c_1\leq b_1$, $c_2\in[2b_2(b_1+c_1)/(b_1+3c_1),2b_2(b_1/(b_1-c_1))^2]$. We begin by constructing the shuffling measure $\bar{\beta}$ (see Figure \ref{fig:measure1}). Define the density $\beta_s$ in the interval $\tilde{D}^{(1)}:=[c_1,c_1+p]\times\{c_2+b_2\}$ as
\begin{equation}\label{eqn:beta}
  \beta_s(x,c_2+b_2):=(2b_2+(3a_1(x-c_1)-c_2-3p_{a_1})\mathbf{1}(x\leq c_1+p_{a_1}/a_1))/(b_1b_2).
\end{equation}

Define $\beta_p:=c_1(b_2-p_{a_1})/(b_1b_2)\delta_{(c_1,c_2+b_2)}$, a point measure at $(c_1,c_2+b_2)$ with mass $c_1(b_2-p_{a_1})/(b_1b_2)$. Define
$$
\bar{\beta}(A):=\int_{c_1}^{c_1+p}\mathbf{1}_A(x,c_2+b_2)\beta_s(x,c_2+b_2)+\beta_p(A\cap(c_1,c_2+b_2))
$$
for all measurable sets $A\subseteq\tilde{D}^{(1)}$. Observe that the densities of $\bar{\beta}$ and $\bar{\alpha}$ differ only in the interval $[c_1+p_{a_1}/a_1,c_1+p]\times\{c_2+b_2\}$. A jump occurs when $\beta_s$ reaches $2b_2-c_2$.

\begin{figure}
\centering
\begin{tikzpicture}[scale=0.4,font=\small,axis/.style={very thick, ->, >=stealth'}]
\draw[thick,<->] (0,-5)--(0,5);
\draw[thick,->] (-1,0)--(12,0);
\node[right] at (0,5) {$\bm{\bar{\beta}(x)}$};
\node[above] at (12,0) {$\bm{x}$};
\draw[axis,->] (2,0)--(2,1);
\node[left] at (0,1) {$c_1(b_2-p_{a_1})$};
\draw[black,fill=black] (0,1) circle (1ex);
\draw[axis,-] (2,-4)--(8,2);
\node[left] at (0,-4) {$2b_2-c_2-3p_{a_1}$};
\draw[black,fill=black] (0,-4) circle (1ex);
\node[left] at (0,2) {$2b_2-c_2$};
\draw[black,fill=black] (0,2) circle (1ex);
\draw[axis,-] (8,3)--(10,3);
\node[left] at (0,3) {$2b_2$};
\draw[black,fill=black] (0,3) circle (1ex);
\node[below] at (2,0) {$c_1$};
\draw[black,fill=black] (2,0) circle (1ex);
\node[below] at (7,0) {\scriptsize$c_1+p_{a_1}/a_1$};
\draw[black,fill=black] (8,0) circle (1ex);
\node[below] at (11,0) {\scriptsize$c_1+p$};
\draw[black,fill=black] (10,0) circle (1ex);
\end{tikzpicture}
\caption{The measure $\bar{\beta}$.}\label{fig:measure1}
\end{figure}
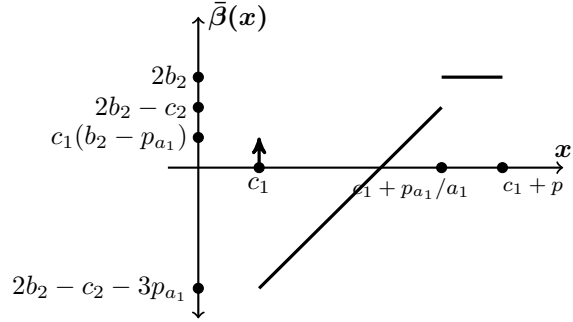

We now identify the canonical partition of $D$ with respect to $\bar{\mu}+\bar{\beta}$, along the same steps used in Section \ref{SUB:GC1}. The construction is illustrated in Figure \ref{fig:illust-sec5-2}.
\begin{figure}[h!]
\centering
\begin{tabular}{ccc}
\subfloat[]{\begin{tikzpicture}[scale=0.4,font=\small,axis/.style={very thick, ->, >=stealth'}]
\draw [axis,thick,-] (1,1)--(9,1);
\draw [axis,thick,-] (9,1)--(9,9);
\draw [axis,thick,-] (9,9)--(1,9);
\draw [axis,thick,-] (1,9)--(1,1);
\draw [thick,dotted] (1,3)--(4,1);
\node [right] at (2,2.5) {$s_1(z_1)$};
\node [below] at (1,1) {\tiny$(c_1,c_2)$};
\node [above] at (7,9) {\tiny$(c_1+b_1,c_2+b_2)$};
\node [above] at (1,9) {\tiny$(c_1,c_2+b_2)$};
\draw[black,fill=black] (4,1) circle (1ex);
\node [below] at (5.5,1) {\tiny$(c_1+p_{a_1}/a_1,c_2)$};
\draw[black,fill=black] (1,3) circle (1ex);
\node [above,rotate=90] at (1,3.5) {\tiny$(c_1,c_2+p_{a_1})$};
\end{tikzpicture}}&
\subfloat[]{\begin{tikzpicture}[scale=0.4,font=\small,axis/.style={very thick, ->, >=stealth'}]
\draw [axis,thick,-] (1,1)--(9,1);
\draw [axis,thick,-] (9,1)--(9,9);
\draw [axis,thick,-] (9,9)--(1,9);
\draw [axis,thick,-] (1,9)--(1,1);
\draw [thick,dotted] (1,3)--(4,1);
\draw [thick] (5,1)--(5,9);
\draw[black,fill=black] (5,1) circle (1ex);
\node [below] at (5.5,1) {\tiny$(c_1+p,c_2)$};
\end{tikzpicture}}&
\subfloat[]{\begin{tikzpicture}[scale=0.4,font=\small,axis/.style={very thick, ->, >=stealth'}]
\draw [axis,thick,-] (1,1)--(9,1);
\draw [axis,thick,-] (9,1)--(9,9);
\draw [axis,thick,-] (9,9)--(1,9);
\draw [axis,thick,-] (1,9)--(1,1);
\draw [thick,dotted] (1,3)--(4,1);
\draw [thick] (5,1)--(9,1);
\draw[black,fill=black] (5,1) circle (1ex);
\node [below] at (5,1) {\tiny$((c_1+b_1)/2,c_2)$};
\path[fill=gray!50,opacity=.9] (1,3) to (4,1) to (1,1) to (1,3);
\node at (1.5,1.5) {\normalsize$Z$};
\path[fill=gray!50,opacity=.6] (1,3) to (4,1) to (5,1) to (5,9) to (1,9) to (1,3);
\node at (3,5) {\huge$A$};
\path[fill=gray!50,opacity=.3] (5,1) to (9,1) to (9,9) to (5,9) to (5,1);
\node at (7,5) {\huge$W$};
\end{tikzpicture}}
\end{tabular}
\caption{An illustration of (a) the construction of $s_1(z_1)$, (b) the construction of $W$, and (c) the canonical partition.}\label{fig:illust-sec5-2}
\end{figure}
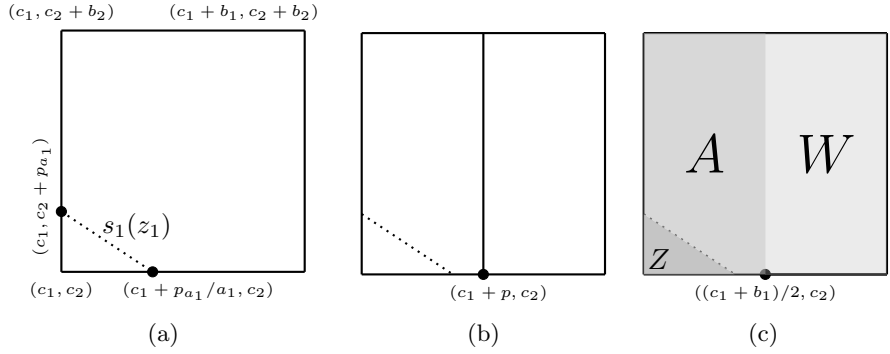

\begin{enumerate}
\item[(a)] The outer boundary function $s_1:[c_1,c_1+p]\rightarrow[c_2,c_2+b_2]$, with respect to $\bar{\mu}+\bar{\beta}$, is computed using the same expressions in (\ref{eqn:si-zi-small}) but with $\alpha_s^{(i)}$ and $\alpha_p^{(i)}$ replaced respectively by $\beta_s$ and $\beta_p$. It is given by
$$
  s_1(z_1)=c_2+(p_{a_1}-a_1(z_1-c_1))_+.
$$
\item[(b)] We now choose the critical price $p$ so that $\bar{\mu}(W)=0$. From an examination of Figure \ref{fig:d} and simple geometry considerations, it is immediate that $\bar{\mu}(W)=0$ iff $b_2(c_1+b_1)+(b_1-p)(-2b_2)=0$, which yields $p=(b_1-c_1)/2$.
\item[(c)] We thus have the following canonical partition of $D$.
\begin{align*}
Z&=\{(z_1,z_2):z_1\leq p_{a_1}/a_1,\,z_2\leq s_1(z_1)\},\\
W&=[c_1+p,c_1+b_1]\times[c_2,c_2+b_2],\\
A&=D\backslash (W\cup Z).
\end{align*}
\end{enumerate}

We now derive the parameters $p_{a_1}$ and $a_1$, by imposing $\bar{\beta}\succeq_{cvx}0$. All proofs of this subsection are relegated to \ref{app:b}.
\begin{proposition}\label{prop:gc2}
Consider the structure in Figure \ref{fig:d}. Then, $\bar{\beta}\succeq_{cvx}0$ is satisfied, if
\begin{equation}\label{eqn:p-a1-pa1}
  p=\frac{3p_{a_1}^2/(2a_1)+c_2p_{a_1}/a_1-c_1(b_2-p_{a_1})}{2b_2}=\frac{p_{a_1}}{a_1}\sqrt{\frac{p_{a_1}+c_2}{2b_2}}.
\end{equation}
If $p=(b_1-c_1)/2$, then $\bar{\beta}\succeq_{cvx}0$ holds when $a_1$ satisfies
\begin{equation}\label{eqn:a1-4d}
  a_1=p_{a_1}\left(\frac{\frac{3}{2}p_{a_1}+c_2}{b_1b_2-c_1p_{a_1}}\right)=\frac{p_{a_1}}{b_1-c_1}\sqrt{\frac{2(p_{a_1}+c_2)}{b_2}},
\end{equation}
and if $p_{a_1}$ is a solution to
\begin{multline}\label{eqn:pa1-4d}
  2b_1^2b_2^2c_2-c_2^2b_2(b_1-c_1)^2+p_{a_1}(2b_1^2b_2^2-4b_1b_2c_1c_2-3c_2b_2(b_1-c_1)^2)\\+p_{a_1}^2(2c_1^2c_2-4b_1b_2c_1-9b_2(b_1-c_1)^2/4)+p_{a_1}^3(2c_1^2)=0.
\end{multline}
\end{proposition}

We have thus identified the parameters of the canonical partition. We now verify if the values of $(p_{a_1},a_1,p)$ so obtained give rise to a valid partition. A partition is valid, if (i) $p_{a_1}\in[0,b_2]$, (ii) $p_{a_1}/a_1\leq p$, and (iii) $a_i\in[0,1]$.

To find the values of $c_i, b_i$ for which the partition is valid, we make a series of claims, proofs of which are in \ref{app:b}.
\begin{enumerate}
 \item At $p_{a_1}=b_2$, the values of $a_1$, $p_{a_1}/a_1$, and $p$ that satisfy (\ref{eqn:p-a1-pa1}) are given by $a_1=\infty$, and $p_{a_1}/a_1=p=0$. See Figure \ref{fig:5.2.1-1}.
 \item So we fix $p_{a_1}=b_2$ and decrease $p_{a_1}$. Suppose $a_1$ and $p$ are modified so as to satisfy (\ref{eqn:p-a1-pa1}). Then, decreasing $p_{a_1}$ from $b_2$ to $(2b_2-c_2)_+$ decreases $a_1$, but increases both $p_{a_1}/a_1$ and $p$. The change of $p_{a_1}/a_1$ and $a_1$ on decreasing $p_{a_1}$ is illustrated in Figure \ref{fig:5.2.1-2}.
 \item We decrease $p_{a_1}$ until $p=(b_1-c_1)/2$. Suppose $c_1\leq b_1$ and $c_2\in[2b_2(b_1+c_1)/(b_1+3c_1),2b_2(b_1/(b_1-c_1))^2]$. Then, the pair $(p_{a_1},a_1)$ that solves (\ref{eqn:p-a1-pa1}) when $p=(b_1-c_1)/2$ satisfies (i) $p_{a_1}\geq(2b_2-c_2)_+$ and (ii) $p_{a_1}/a_1\leq p$. See Figure \ref{fig:5.2.1-3}.
 \item If the value of $a_1$ so obtained is at most $1$, then the optimal mechanism is as in Figure \ref{fig:d}.
 \item In case $a_1>1$, then the optimal mechanism is as in Figure \ref{fig:c'}.
\end{enumerate}

\begin{figure}
\centering
\begin{tabular}{ccc}
\subfloat[]{\label{fig:5.2.1-1}\begin{tikzpicture}[scale=0.38,font=\small,axis/.style={very thick, ->, >=stealth'}]
\draw [axis,thick,-] (1,1)--(9,1);
\draw [axis,thick,-] (9,1)--(9,9);
\draw [axis,thick,-] (9,9)--(1,9);
\draw [axis,thick,-] (1,9)--(1,1);
\draw[black,fill=black] (1,1) circle (1ex);
\node [below] at (1,1) {$\frac{p_{a_1}}{a_1}$};
\node [left] at (1,1.5) {$p$};
\draw[black,fill=black] (1,9) circle (1ex);
\node [left] at (1,9) {$p_{a_1}$};
\node [right] at (1,5) {$a_1=\infty$};
\end{tikzpicture}}&
\subfloat[]{\label{fig:5.2.1-2}\begin{tikzpicture}[scale=0.38,font=\small,axis/.style={very thick, ->, >=stealth'}]
\draw [axis,thick,-] (1,1)--(9,1);
\draw [axis,thick,-] (9,1)--(9,9);
\draw [axis,thick,-] (9,9)--(1,9);
\draw [axis,thick,-] (1,9)--(1,1);
\draw [thick] (1,3.5)--(3.5,1);
\draw [thick] (4.5,1)--(4.5,9);
\draw [thick,dotted] (1,3)--(4,1);
\draw [thick,dotted] (5,1)--(5,9);
\draw[black,fill=black] (1,3) circle (1ex);
\node [left] at (1,3) {$p_{a_1}'$};
\draw[black,fill=black] (4,1) circle (1ex);
\node [below] at (4,1.25) {\scriptsize$\frac{p_{a_1}'}{a_1'}$};
\draw[black,fill=black] (5,1) circle (1ex);
\node [below] at (5.5,1) {\scriptsize$p'$};
\draw[black,fill=black] (1,3.5) circle (1ex);
\node [left] at (1,4) {$p_{a_1}$};
\draw[black,fill=black] (3.5,1) circle (1ex);
\node [below] at (2.75,1) {\scriptsize$\frac{p_{a_1}}{a_1}$};
\draw[black,fill=black] (4.5,1) circle (1ex);
\node [below] at (4.75,1) {\scriptsize$p$};
\node [right] at (1,5) {$a_1'<a_1$};
\end{tikzpicture}}&
\subfloat[]{\label{fig:5.2.1-3}\begin{tikzpicture}[scale=0.38,font=\small,axis/.style={very thick, ->, >=stealth'}]
\draw [axis,thick,-] (1,1)--(9,1);
\draw [axis,thick,-] (9,1)--(9,9);
\draw [axis,thick,-] (9,9)--(1,9);
\draw [axis,thick,-] (1,9)--(1,1);
\draw [thick] (1,3)--(4,1);
\draw [thick] (5,1)--(5,9);
\draw[black,fill=black] (1,3) circle (1ex);
\node [left] at (1,3) {$p_{a_1}$};
\draw[black,fill=black] (4,1) circle (1ex);
\node [below] at (3.5,1) {$\frac{p_{a_1}}{a_1}$};
\draw[black,fill=black] (5,1) circle (1ex);
\node [below] at (5.5,1) {$\frac{b_1-c_1}{2}$};
\node [right] at (1,5) {$p_{a_1}\geq(2b_2-c_2)_+$};
\end{tikzpicture}}
\end{tabular}
\caption{An illustration of the steps enumerated in Section 5.2.1. (a) If $p_{a_1}=b_2$, then $a_1=\infty$, and $p_{a_1}/a_1=p=0$; (b) If $p_{a_1}$ decreases to $p_{a_1}'$, then $a_1'<a_1$, but $p_{a_1}'/a_1'>p_{a_1}/a_1$ and $p'>p$; (c) If $p=(b_1-c_1)/2$, then $p_{a_2}\geq(2b_2-c_2)_+$.}\label{fig:illust-sec5.2.1}
\end{figure}
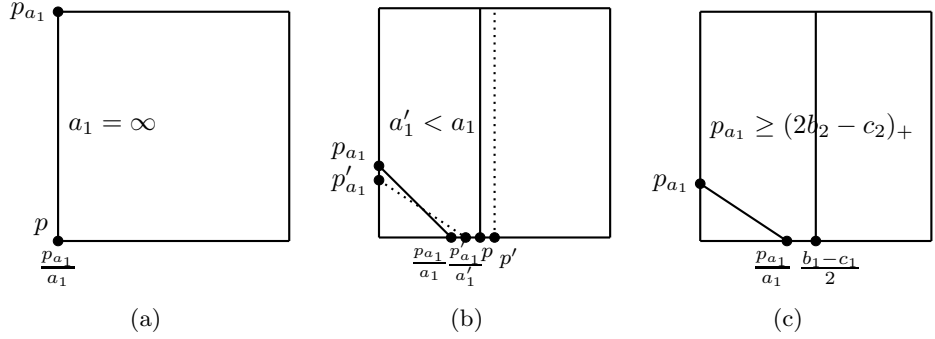

We now assert the following.
\begin{theorem}\label{thm:gc4}
Let $\{c_1\leq b_1$, $c_2\in[2b_2(b_1+c_1)/(b_1+3c_1),2b_2(b_1/(b_1-c_1))^2]\}$. Then, there exists $p_{a_1}\in[(2b_2-c_2)_+,b_2]$ that solves (\ref{eqn:pa1-4d}). If $a_1$ found from (\ref{eqn:a1-4d}) is at most $1$, then the optimal mechanism is as in Figure \ref{fig:d}, with $(p_{a_1},a_1)$ as found above, and $p=(b_1-c_1)/2$. Instead, if $a_1>1$, then the optimal mechanism is as in Figure \ref{fig:c'}, with $p=p^*$.
\end{theorem}

\subsubsection{Optimal mechanisms when $\frac{c_1}{b_1}$ is small but $\frac{c_2}{b_2}$ is very large}
Let $c_1\leq b_1$ and $c_2\geq 2b_2(b_1/(b_1-c_1))^2$. Consider $c_2=2b_2(b_1/(b_1-c_1))^2$. We start from $p_{a_1}=b_2$ yet again, decrease $p_{a_1}$, and modify $a_1$ and $p$ according to (\ref{eqn:p-a1-pa1}).
\begin{claim}\label{clm:beta-eh-1}
 At $p_{a_1}=0$, we claim that $a_1=0$ and $p=(b_1-c_1)/2$.
\end{claim}
The structure is now the same as depicted in Figure \ref{fig:e}.

We now construct a shuffling measure that is a variant of $\bar{\beta}$. Observe from (\ref{eqn:a1-4d}) that $\lim_{p_{a_1}\rightarrow 0}(p_{a_1}/a_1)=b_1b_2/c_2$. So, when $c_2=2b_2(b_1/(b_1-c_1))^2$, the density $\beta_s$ (in (\ref{eqn:beta})) becomes a step function with values $(2b_2-c_2)/(b_1b_2)$ and $2b_2/(b_1b_2)$ with the jump occurring at $(c_1+b_1b_2/c_2,c_2+b_2)$. We fix this density for every $c_2\geq 2b_2(b_1/(b_1-c_1))^2$, and call it $\beta_{es}$. We also define a point measure $\beta_{ep}:=c_1b_2/(b_1b_2)\delta_{(c_1,c_2+b_2)}$. The resulting measure is
$$
\bar{\beta_e}(A):=\int_{c_1}^{c_1+p}\mathbf{1}_A(x,c_2+b_2)\beta_{es}(x,c_2+b_2)\,dx+\beta_{ep}(A\cap(c_1,c_2+b_2))
$$
for all measurable sets $A\subseteq\tilde{D}^{(1)}$. The components of $\bar{\beta}_e$ are given by
\begin{align}
  \beta_{es}(x,c_2+b_2)&=\begin{cases}(2b_2-c_2)/(b_1b_2)&\mbox{if }x\in[c_1,c_1+b_1b_2/c_2]\\
   2b_2/(b_1b_2)&\mbox{if }x\in(c_1+b_1b_2/c_2,(c_1+b_1)/2],\end{cases}\nonumber\\
  \beta_{ep}(c_1,c_2+b_2)&=c_1b_2/(b_1b_2).\label{eqn:beta-eh}
\end{align}

The canonical partition of the support $D$ with respect to $\bar{\mu}+\bar{\beta}_e$ is the same as constructed in Section 5.2.1 (as illustrated in Figure \ref{fig:illust-sec5-2}), with $p_{a_1}=a_1=0$, and $p_{a_1}/a_1=b_1b_2/c_2$.

We now verify if $\bar{\beta}_e\succeq_{cvx}0$.
\begin{proposition}\label{prop:beta-eh}
Suppose $c_1\leq b_1$ and $c_2\geq 2b_2(b_1/(b_1-c_1))^2$. Then, $\bar{\beta}_e$ satisfies (i) $\bar{\beta}_e\succeq_{cvx}0$ and (ii) $\int_{\tilde{D}^{(1)}}u\,d\bar{\beta}_e=0$ for constant $u$.
\end{proposition}

We now assert that the following theorem holds, proof of which is relegated to \ref{app:b}.
\begin{theorem}\label{thm:gc5}
The optimal mechanism is as in Figure \ref{fig:e}, when $\{c_1\leq b_1, c_2\geq 2b_2(b_1/(b_1-c_1))^2\}$.
\end{theorem}

Put together, Theorem \ref{thm:gc4} of Section \ref{subsubsec:small-intermediate} and Theorem \ref{thm:gc5} characterize the optimal mechanisms for $c_i, b_i$ satisfying $\{c_1\leq b_1, c_2\geq 2b_2(b_1+c_1)/(b_1+3c_1)\}$.

\subsection{Optimal mechanisms when $\frac{c_2}{b_2}$ is small but $\frac{c_1}{b_1}$ is large}\label{SUB:GC3}

Consider the case when $c_2<b_2$ but $c_1>2b_1(b_2+c_2)/(b_2+3c_2)$. We define the density $\beta_s$ in the interval $\tilde{D}^{(2)}:\{c_1+b_1\}\times[c_2,c_2+p]$ as
\begin{equation}\label{eqn:beta-sym}
  \beta_s(c_1+b_1,x):=(2b_1+(3a_2(x-c_2)-c_1-3p_{a_2})\mathbf{1}(x\leq c_2+p_{a_2}/a_2))/(b_1b_2).
\end{equation}
We also define $\beta_p:=c_2(b_1-p_{a_2})/(b_1b_2)\delta_{(c_1+b_1,c_2)}$, a point measure at $(c_1+b_1,c_2)$ with mass $c_2(b_1-p_{a_2})/(b_1b_2)$. Define
$$
\bar{\beta}(A):=\int_{c_2}^{c_2+p}\mathbf{1}_A(c_1+b_1,x)\beta_s(c_1+b_1,x)+\beta_p(A\cap(c_1+b_1,c_2))
$$
for all measurable sets $A\subseteq\tilde{D}^{(2)}$. With symmetric arguments, we have the following assertions similar to Proposition \ref{prop:gc2},  Theorem \ref{thm:gc4}, and Theorem \ref{thm:gc5}.

\subsubsection{Optimal mechanisms when $\frac{c_2}{b_2}$ is small and $\frac{c_1}{b_1}$ has large but not too large values}

\vspace*{.1in}

The following is the analog of Proposition \ref{prop:gc2}.

\begin{proposition}\label{prop:gc2-sym}
Consider the structure in Figure \ref{fig:g}. If $p=(b_2-c_2)/2$, then $\bar{\beta}\succeq_{cvx}0$ is satisfied, if $a_2$ satisfies
\begin{equation}\label{eqn:a2-4d}
  a_2=p_{a_2}\left(\frac{\frac{3}{2}p_{a_2}+c_1}{b_1b_2-c_2p_{a_2}}\right)=\frac{p_{a_2}}{b_2-c_2}\sqrt{\frac{2(p_{a_2}+c_1)}{b_1}}
\end{equation}
and if $p_{a_2}$ is a solution to
\begin{multline}\label{eqn:pa2-4d}
  2b_1^2b_2^2c_1-c_1^2b_1(b_2-c_2)^2+p_{a_2}(2b_1^2b_2^2-4b_1b_2c_1c_2-3c_1b_1(b_2-c_2)^2)\\+p_{a_2}^2(2c_2^2c_1-4b_1b_2c_2-9b_1(b_2-c_2)^2/4)+p_{a_2}^3(2c_2^2)=0.
\end{multline}
\end{proposition}

We also have the analog of Theorem \ref{thm:gc4}.

\begin{theorem}\label{thm:gc4-sym}
Let $\{c_2\leq b_2$, $c_1\in[2b_1(b_2+c_2)/(b_2+3c_2),2b_1(b_2/(b_2-c_2))^2]\}$. Then, there exists $p_{a_2}\in[(2b_1-c_1)_+,b_1]$ that solves (\ref{eqn:pa2-4d}). If $a_2$ found from (\ref{eqn:a2-4d}) is at most $1$, then the optimal mechanism is as in Figure \ref{fig:g}, with $(p_{a_2},a_2)$ as found above, and $p=(b_2-c_2)/2$. Instead, if $a_2>1$, then the optimal mechanism is as in Figure \ref{fig:c'}, with $p=p^*$.
\end{theorem}

\subsubsection{Optimal mechanisms when $\frac{c_2}{b_2}$ is small but $\frac{c_1}{b_1}$ is very large}

\vspace*{.1in}
The next assertion is the analog of Theorem \ref{thm:gc5}.

\begin{theorem}\label{thm:gc5-sym}
The optimal mechanism is as in Figure \ref{fig:h} when $\{c_2\leq b_2, c_1\geq 2b_1(b_2/(b_2-c_2))^2\}$.
\end{theorem}

\subsection{Optimal mechanisms when $\frac{c_1}{b_1}$ and $\frac{c_2}{b_2}$ are large}\label{SUB:GC4}

\vspace*{.1in}

When $c_1\geq b_1$ and $c_2\geq b_2$, we have the following theorem.
\begin{theorem}\label{thm:figc}
 The optimal mechanism is as in Figure \ref{fig:c} when $\{c_1\geq b_1, c_2\geq b_2\}$.
\end{theorem}
A simple way to prove this theorem is to use \citet[Prop.~3]{MHJ15}. The proposition asserts that for negative power rate distributions, if $z_i-(1-F_i(z_i))/f_i(z_i)\geq 0$ for all $z_i\in D$, $i=1,2$, then pure bundling is the optimal mechanism. The conditions of the proposition clearly hold when $c_i\geq b_i, i=1,2$, and thus the theorem is proved.

Alternatively, we can prove this theorem via the duality approach used in this paper. We fix $\theta=\bar{\mu}^W_+$ and show that $\theta$ satisfies all the conditions of Lemma \ref{lem:DDT-equiv}. This shows the existence of a dual measure $\gamma$ that satisfies the dual constraints and the complementary slackness constraints. The complete proof is in \ref{app:b}.

\section{Extension to Other Distributions}\label{sec:extension}
Optimal mechanism for uniform distribution over any rectangle was found only using $\bar{\alpha}$ and its variants as shuffling measures. We can now ask if there is a generalization of $\bar{\alpha}$ for other distributions. For distributions with constant negative power rate, we anticipate that the optimal mechanisms would be of the same form as in uniform distributions (based on the result of \citet{WT14}), and that they can be arrived at using similar $\bar{\alpha}$. We propose using a ``generalized'' $\bar{\alpha}$, whose components are
\begin{align}
  \alpha_s^{(1)}(z_1,c_2+b_2)&=f_1(z_1)(-\Delta(1-F_2(p_{a_1}-a_1(z_1-c_1))))\nonumber\\&\hspace*{1in}-f_1(z_1)(c_2+b_2)f_2(c_2+b_2),\,z\in[c_1,P_1],\nonumber\\
  \alpha_p^{(1)}(c_1,c_2+b_2)&=c_1f_1(c_1)(1-F_2(p_{a_1})),\label{eqn:gen-shuffle}
\end{align}
along with similarly defined $\alpha_s^{(2)}$ and $\alpha_p^{(2)}$. Recall that $\Delta$ is the power rate of the distribution defined as $-3-z_1f_1'(z_1)/f_1(z_1)-z_2f_2'(z_2)/f_2(z_2)$. Observe that the $\bar{\alpha}$ mentioned above reduces to the $\bar{\alpha}$ used in the case of uniform distributions, when we substitute $f_i(z_i)=1$, and $\Delta=-3$ (see (\ref{eqn:alpha_s})).

We have six parameters to be computed: $p_{a_1}$, $p_{a_2}$, $a_1$, $a_2$, $P_1$, and $Q_2$. Defining $\tilde{D}^{(1)}:=[c_1,P_1]\times\{c_2+b_1\}$ and $\tilde{D}^{(2)}:=\{c_1+b_1\}\times[c_2,Q_2]$, we use the following six equations for computing these parameters.
\begin{itemize}
\item We require $\alpha^{(1)}\succeq_{cvx}0$. Using Proposition \ref{prop:cvx}, we impose $\bar{\alpha}^{(1)}(\tilde{D}^{(1)})=\int_{[c_1,P_1]}x\,\bar{\alpha}^{(1)}(dx,c_2+b_2)=0$.
\item Similarly, we require $\alpha^{(2)}\succeq_{cvx}0$. We thus impose $\bar{\alpha}^{(2)}(\tilde{D}^{(2)})=\int_{[c_2,Q_2]}x\,\bar{\alpha}^{(2)}(c_1+b_1,dx)=0$.
\item We require that the critical points $P$ and $Q$ be connected by a line $z_1+z_2=p$. We thus impose $P_1+P_2=Q_1+Q_2$ or $P_1+c_2+p_{a_1}-a_1(P_1-c_1)=c_1+p_{a_2}-(Q_2-c_2)/a_2$.
\item We finally impose $\bar{\mu}(W)=0$.
\end{itemize}

If there exists a tuple $(p_{a_1},p_{a_2},a_1,a_2,P_1,Q_2)$ that solves these six equations simultaneously and forms a valid canonical partition, then we can assert that the optimal mechanism is as depicted in Figure \ref{fig:a}, using the same arguments in section \ref{SUB:GC1}.

We now consider an example of a distribution with negative constant power rate. Let $f_1(z)=f_2(z)=2z/(2c+1)$ when $z\in[c,c+1]$. The power rate $\Delta=-5$, a negative constant. We now have the following theorem.
\begin{theorem}\label{THM:EXTENSION}
 Let $f_i(z)=2z/(2c+1)$, $z\in[c,c+1]$. When $c=0$, the optimal mechanism is as in Figure \ref{fig:a}, with $p_{a_1}=\sqrt{0.6}$, $a_1=0$, and $p=1.09597$. When $c=0.1$, the optimal mechanism is the same, with $p_{a_1}=0.79615$, $a_1=0.23198$, and $P_1=0.364655$.
\end{theorem}

The detailed proof is in \ref{app:c}. We first characterize the optimal mechanism for $c=0$ using the method outlined in Section \ref{sec:zero}, without any shuffling measures. We then consider the case $c=0.1$ where we introduce the shuffling measure $\bar{\alpha}$ in (\ref{eqn:gen-shuffle}), compute the parameters by solving the above six equations, and prove that the canonical partition obtained is valid.

For other structures, one could follow a similar procedure as done for Figure \ref{fig:a} and one could obtain parameter ranges for which those structures are optimal. On the basis of the above development, we conjecture that the generalized shuffling measure in (\ref{eqn:gen-shuffle}) and its variants should suffice for the class of distributions having negative constant power rates.

\section{Conclusion}\label{sec:conclude}
We solved the problem of computing the optimal solution in the two-item one-buyer setting, when the valuations of the buyer are uniformly distributed in an arbitrary rectangle on the positive quadrant. Our results show that a wide range of structures arise out of different values of $(c_1,c_2)$. When the buyer guarantees that his valuations are at least $(c_1,c_2)$, the seller offers different menus based on the guaranteed $(c_1,c_2)$ and the upper bounds $(c_1+b_1,c_2+b_2)$. The seller also finds it optimal to sell the items individually for certain values of $(c_1,c_2)$, to sell the items as a bundle for certain other values, and to pose the menu as an interplay of individual sale and a bundled sale for the remaining values of $(c_1,c_2)$.

Our solutions used the duality approach in \cite{DDT17}. We constructed a ``shuffling measure'' for various distributions that are not bordered by the coordinate axes, thereby extending the pool of solvable problems using the dual method. Our results for the linear distribution function suggests that the method of constructing shuffling measures could be used to solve the problem for a wider class of distributions. We conjecture that our method works for all distributions with constant negative power rate.

Another approach based on the idea of {\em virtual valuation} was used by \citet{Pav11} to compute the optimal solution when the buyer's valuations are uniformly distributed in $[c,c+1]^2$, for arbitrary nonnegative values of $c$. In a companion paper where we consider the unit-demand setting \cite{TRN17b}, we highlight the sensitivity of the structure of the shuffling measure to the parameters. For that setting, we do not as yet fully understand what shuffling measures will work. In that paper, we follow Pavlov's approach to compute the optimal solution when $z\sim\mbox{Unif}[c,c+b_1]\times[c,c+b_2]$, for arbitrary nonnegative values of $(c,b_1,b_2)$. The two works, taken together, provide concrete examples of nontrivial optimal structures and thus provide insight on optimal mechanisms likely to occur under more general settings. They also help us understand the pros and cons of the dual and the virtual valuation methods under different situations.

Our results show the existence of multi-dimensional optimal mechanisms with no exclusion region. The optimal mechanism has no exclusion region when the minimum valuation for one of the items is small and that for the other item is very large. This is interesting since the results of \citet{Arm96} and \citet{PBBK14} state that the optimal mechanisms in multi-dimensional setting have a nontrivial exclusion region under some sufficient conditions. The setting in our paper does not satisfy their sufficient conditions.

The optimal mechanisms in the uniform case considered in this paper are stochastic when both $\frac{c_1}{b_1}$ and $\frac{c_2}{b_2}$ are small, and are deterministic when either $\frac{c_1}{b_1}$ or $\frac{c_2}{b_2}$ is sufficiently large. We can now ask the following question: can we characterize the set of distributions for which deterministic mechanisms are optimal? The work in \cite{MV06} answered this under the assumption that the support set is $[0,1]^2$. Characterizing them for general support sets is a direction of future work.

Another direction is to characterize those distributions for which deterministic mechanisms give a constant-factor approximation. The works in \cite{CHK07,CMS15} considered the unit-demand setting, and proved that if the distributions on the items are either uncorrelated or have a certain positive correlation, then deterministic mechanisms provide at least a constant-fraction of the optimal revenue. Characterizing such distributions for the setting we consider in this paper is another direction for future work.

\section*{Acknowledgements}
This work was supported by the Defence Research and Development Organisation [Grant no. DRDO0667] under the DRDO-IISc Frontiers Research Programme.

\appendix
\section*{APPENDIX}
\section{Proofs from Section \ref{SUB:GC1}}\label{app:a}

\textbf{Proof of Proposition \ref{prop:cvx}:} Consider $h$ to be the affine shift of any increasing convex function $g$ (i.e., $h=\beta_1g+\beta_2, \beta_1>0, \beta_2\in\mathbb{R}$) such that $h(c_1+l_1)=c_1+l_1$, and $h(c_1+l_2)=c_1+l_2$. Observe that $h(x)\leq x$ when the density of $\alpha$ is nonpositive, and $h(x)\geq x$ when the density is nonnegative. Now we have
\begin{align*}
  &\int_{[c_1,c_1+m_1]}g\,d\alpha\\ &= \frac{1}{\beta_1}\int_{[c_1,c_1+m_1]}h\,d\alpha\\
  & = \frac{1}{\beta_1}\left(\int_{[c_1,c_1+m_1]}(h(x)-x)\,\alpha(dx)+\int_{[c_1,c_1+m_1]}x\,\alpha(dx)\right)\\
  & =\frac{1}{\beta_1}\left(\int_{[c_1,c_1+m_1]}(h(x)-x)\,\alpha(dx)\right)\\
  &\geq 0
\end{align*}
where the first equality follows from $\alpha([c_1,c_1+m_1])=0$, the third equality follows from $\int_{[c_1,c_1+m_1]}x\,\alpha(dx)=0$, and the last inequality follows because $\sgn{(h(x)-x)}=\sgn{(\alpha(x))}$ for every $x\in[c_1,c_1+m_1]$. Hence the result.

\vspace*{10pt}
{\bf Proof of Corollary \ref{cor:cvx}:}  We now find the values of $m_1$ and $a_1$ for which $\bar{\alpha}^{(1)}(\tilde{D}^{(1)})=\int_{[c_1,c_1+m_1]}(x-c_1)\,\bar{\alpha}^{(1)}(dx,c_2+b_2)=0$ holds.
$$
  \bar{\alpha}^{(1)}(\tilde{D}^{(1)})=((2b_2-c_2-3p_{a_1})m_1+(3/2)a_1m_1^2+c_1(b_2-p_{a_1}))/(b_1b_2)=0
$$
\begin{multline}
  \int_{[c_1,c_1+m_1]}(x-c_1)\,\bar{\alpha}^{(1)}(dx,c_2+b_2)=((2b_2-c_2-3p_{a_1})m_1^2/2+a_1m_1^3)/(b_1b_2)=0
\end{multline}
Solving both of these equations simultaneously, we obtain $m_1=4c_1(b_2-p_{a_1})/(c_2-2b_2+3p_{a_1})$ and $a_1=(c_2-2b_2+3p_{a_1})^2/(8c_1(b_2-p_{a_1}))$. So $\bar{\alpha}^{(1)}(\tilde{D}^{(1)})=\int_{[c_1,c_1+m_1]}(x-c_1)\,\bar{\alpha}^{(1)}(dx,c_2+b_2)=0$ holds at these values of $m_1$ and $a_1$. We also have $\int_{\tilde{D}^{(1)}} f\,d\bar{\alpha}^{(1)}=0$ for any affine $f$, because $f(x)$ can be written as $\beta_1x+\beta_2$.

$\bar{\alpha}^{(1)}\succeq_{cvx}0$ follows from Proposition \ref{prop:cvx}.\qed

\vspace*{10pt}
{\bf Proofs of Claims for Case 1}

\begin{enumerate}
\item The lower bounds of $r_i$ are clear because for every $c_i,b_i\geq 0$, the inequalities $(2b_i+5c_i)\geq(2b_i+3c_i)$ and $(2b_i-3c_i)\leq(2b_i+3c_i)$ hold. Now we prove that if $c_i, b_i$ satisfy (\ref{eqn:c1-c2-small}), then $\frac{2b_2(2b_1+5c_1)-c_2(2b_1-3c_1)}{3(2b_1+3c_1)}\leq b_2$ holds. Rewriting this condition, we have $b_2(2b_1-c_1)+c_2(2b_1-3c_1)\geq 0$.

When $\{c_2\leq b_2, c_1\leq2b_1\frac{b_2+c_2}{b_2+3c_2}\}$ is satisfied, we have
\begin{multline*}
  b_2(2b_1-c_1)+c_2(2b_1-3c_1)\\\geq 2b_1b_2\left(1-\frac{b_2+c_2}{b_2+3c_2}\right)+2b_1c_2\left(1-3\frac{b_2+c_2}{b_2+3c_2}\right)\geq 0
\end{multline*}
where the first inequality follows from $c_1\leq2b_1\frac{b_2+c_2}{b_2+3c_2}$. When $\{c_1\leq b_1, c_2\leq 2b_2\frac{b_1+c_1}{b_1+3c_1}\}$ is satisfied, the condition trivially holds for $c_1\leq\frac{2b_1}{3}$. When $c_1\in[\frac{2b_1}{3},b_1]$, we have
$$
  b_2(2b_1-c_1)+c_2(2b_1-3c_1)\geq 3b_2(2b_1+3c_1)(b_1-c_1)/(b_1+3c_1) \geq 0,
$$
where the first inequality follows from $c_2\leq2b_2\frac{b_1+c_1}{b_1+3c_1}$, and the second inequality from $c_1\leq b_1$. This establishes $r_1\in[(2b_2-c_2)/3,b_2]$. Establishing $r_2\in[(2b_1-c_1)/3, b_1]$ is analogous.

Substituting $p_{a_i}=r_i$ in (\ref{eqn:P-Q-new}), we obtain
$$
  P_i=Q_i=c_i+\frac{4b_i(b_{-i}+c_{-i})-2c_i(b_{-i}+3c_{-i})}{3(2b_{-i}+3c_{-i})}.
$$
The points $P$ and $Q$ thus coincide, when $p_{a_i}=r_i$. Substituting $p_{a_i}=r_i$ in (\ref{eqn:big-W}), the LHS equals $-\frac{4}{3}c_1c_2(b_1b_2+b_1c_2+b_2c_1)\leq 0$. Recalling that (\ref{eqn:big-W}) was constructed by equating $-\bar{\mu}(W)=0$, we conclude that $\bar{\mu}(W)\geq 0$ when $p_{a_i}=r_i$.\qed
\item Starting from $p_{a_i}=r_i$, we increase $p_{a_1}$ and adjust $p_{a_2}$ so that $P_1+P_2=Q_1+Q_2$ holds (i.e., (\ref{eqn:pa1-pa2}) is satisfied). Assume $p_{a_2}$ decreases. By Observation \ref{obs:pa1-increase}, an increase in $p_{a_1}$ moves $P$ towards south-west, causing a decrease in $P_1+P_2$. Similarly, a decrease in $p_{a_2}$ moves $Q$ towards north-east, causing an increase in $Q_1+Q_2$. But when $p_{a_i}=r_i$, we have $P_1+P_2=Q_1+Q_2$, by the previous claim. So $P_1+P_2=Q_1+Q_2$ fails to satisfy when $p_{a_1}$ increases and $p_{a_2}$ decreases, leading to a contradiction. Thus $p_{a_2}$ must increase.

We proceed to prove that $P$ lies to the north-west of $Q$ when $p_{a_i}\in[r_i,p_{a_i}^*]$. We first note that $a_1\leq 1$ and $a_2\leq 1$ (from (\ref{eqn:for-a1-a2})) for all $p_{a_i}\in[r_i,p_{a_i}^*]$. To prove $P$ lies to the north-west of $Q$, we show that $P_1-Q_1$ decreases when both $p_{a_1}$ and $p_{a_2}$ are increased so as to satisfy (\ref{eqn:pa1-pa2}).

We first find the expression for $\frac{\partial p_{a_2}}{\partial p_{a_1}}$ when $P_1+P_2=Q_1+Q_2$ (i.e., equation (\ref{eqn:pa1-pa2})) is satisfied. Differentiating on both sides of $P_1+P_2=Q_1+Q_2$, we have
\begin{multline*}
  \left(\frac{4c_1(c_2+b_2)}{(c_2-2b_2+3p_{a_1})^2}+\frac{1}{2}\right)\partial p_{a_1}=\left(\frac{4c_2(c_1+b_1)}{(c_1-2b_1+3p_{a_2})^2}+\frac{1}{2}\right)\partial p_{a_2}\\\Rightarrow\frac{\partial p_{a_2}}{\partial p_{a_1}}=\frac{\frac{4c_1(c_2+b_2)}{(c_2-2b_2+3p_{a_1})^2}+\frac{1}{2}}{\frac{4c_2(c_1+b_1)}{(c_1-2b_1+3p_{a_2})^2}+\frac{1}{2}}.
\end{multline*}
Differentiating $P_1-Q_1$ w.r.t. $p_{a_1}$, we have
\begin{multline*}
  \frac{\partial}{\partial p_{a_1}}(P_1-Q_1)=-\frac{4c_1(c_2+b_2)}{(c_2-2b_2+3p_{a_1})^2}+\frac{1}{2}\frac{\partial p_{a_2}}{\partial p_{a_1}}\\=\frac{(1/2)}{\frac{8c_2(c_1+b_1)}{(c_1-2b_1+3p_{a_2})^2}+1}\left(1-\frac{64c_1c_2(c_1+b_1)(c_2+b_2)}{(c_2-2b_2+3p_{a_1})^2(c_1-2b_1+3p_{a_2})^2}\right)\leq 0
\end{multline*}
where the second equality follows by substitution of $\frac{\partial p_{a_2}}{\partial p_{a_1}}$, and the inequality follows because $a_1a_2\leq 1$ from (\ref{eqn:for-a1-a2}) implies $(c_2-2b_2+3p_{a_1})^2(c_1-2b_1+3p_{a_2})^2\leq 64c_1c_2(b_2-p_{a_1})(b_1-p_{a_2}) \leq 64c_1c_2(b_2+c_2)(b_1+c_1)$. With $P_1-Q_1=0$ at $p_{a_1}=r_1$, and the derivative w.r.t. $p_{a_1}$ is nonpositive whenever $a_1a_2\leq 1$, we have $P_1-Q_1\leq 0$ for all $p_{a_1}\in[r_1,p_{a_1}^*]$. The other condition $Q_2-P_2\leq 0$ automatically holds because $P_1+P_2=Q_1+Q_2$. Hence $P$ lies to the north-west of $Q$.\qed
\item If $P$ is to the north-west of $Q$, the illustration in Figure \ref{fig:illust1} shows that $\bar{\mu}(W)$ decreases when either $p_{a_1}$ or $p_{a_2}$ increases.\qed
\begin{figure}
\centering
\begin{minipage}{.46\textwidth}
\centering
\includegraphics[height=4.5cm,width=5cm]{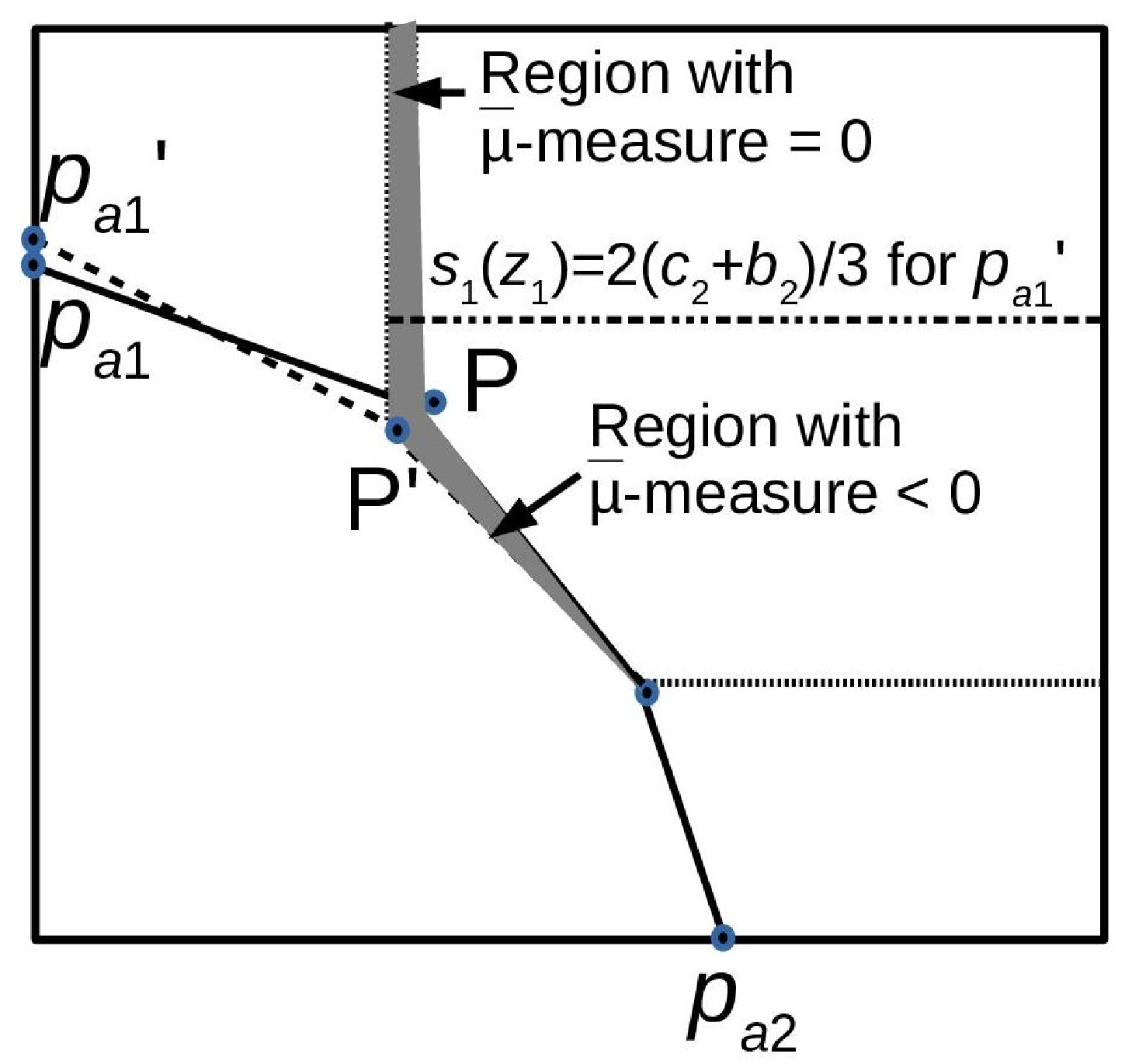}
\caption{An illustration to indicate $\bar{\mu}(W)$ decreases when $p_{a_1}$ increases to $p_{a_1}'$.}\label{fig:illust1}
\end{minipage}
\begin{minipage}{.03\textwidth}
\hspace*{.03\textwidth}
\end{minipage}
\begin{minipage}{.46\textwidth}
\centering
\includegraphics[height=4.5cm,width=5cm]{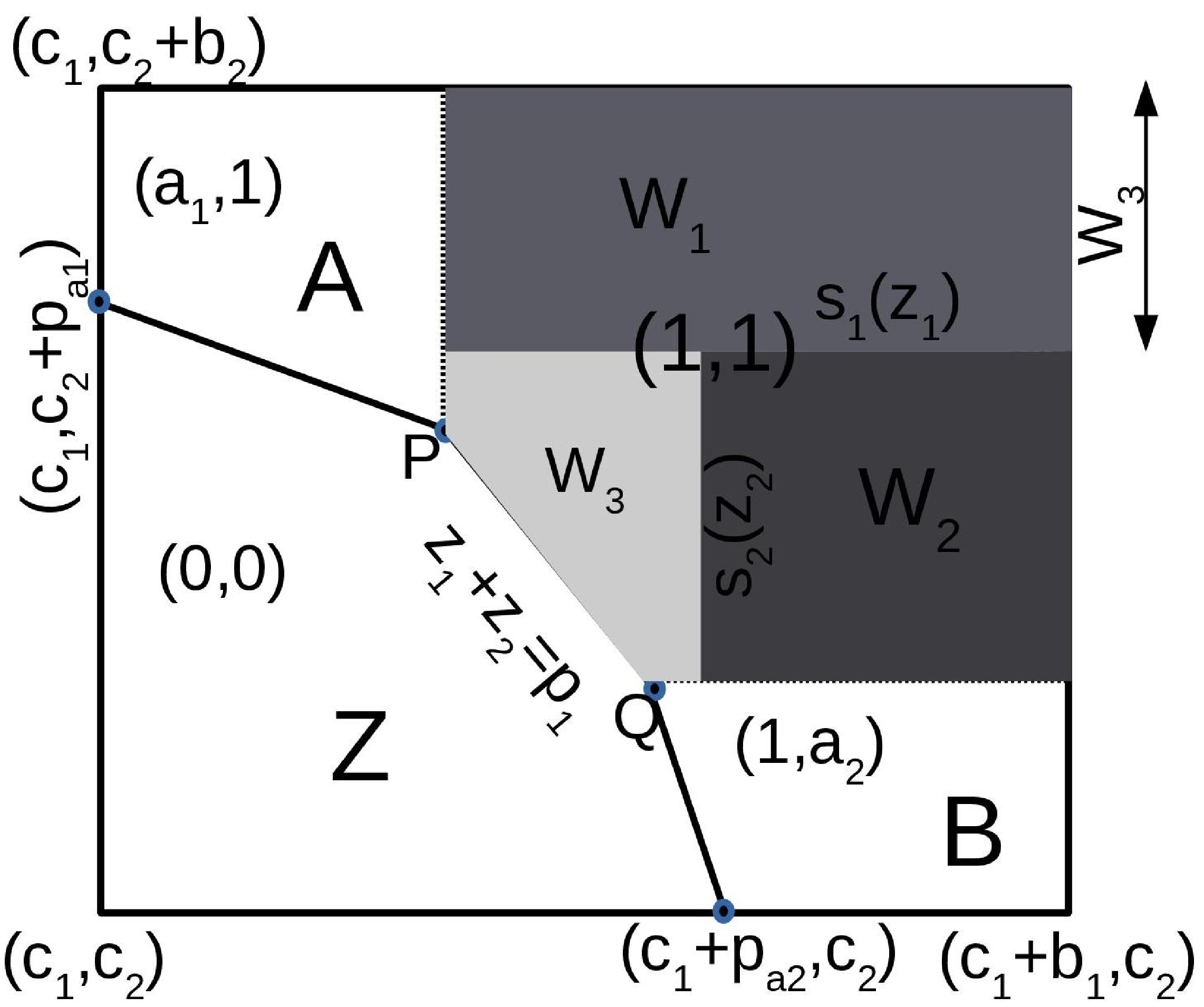}
\caption{The structure in Figure \ref{fig:a} shown with $D$ and $W$ partitioned.}\label{fig:W-partition}
\end{minipage}
\end{figure}
\item We first prove that $P$ and $Q$ in Figure \ref{fig:a} lies within $D$. This is true if (a) $p_{a_i}\leq b_{-i}$, (b) $a_i\in[0,1]$, and (c) $P$ is to the north-west of $Q$. Condition (a) holds because $p_{a_i}^*\leq b_{-i}$ and $p_{a_i}\leq p_{a_i}^*$. Condition (b) holds because $p_{a_i}\leq p_{a_i}^*\leq b_{-i}\Rightarrow a_i\in[0,1]$. Condition (c) holds by claim 2.

To prove that the optimal mechanism is as in Figure \ref{fig:a}, we first partition $W=D\backslash(Z\cup A\cup B)$ into three regions:
\begin{multline*}
 W_1:[P_1,c_1+b_1)\times[s_1(z_1),c_2+b_2],\,W_2:[s_2(z_2),c_1+b_1]\times[Q_2,s_1(z_1)],\\W_3:\mbox{ the remaining region}.
\end{multline*}
See Figure \ref{fig:W-partition}. We now construct $u$ and $\gamma$ in the same way as in the proof of Proposition \ref{prop:known}, with $\bar{\mu}$ replaced by $(\bar{\mu}+\bar{\alpha})$.\qed
\item By Observation \ref{obs:pa1-increase}, we know that increasing $p_{a_1}$ moves $P$ towards south-west. The claim that $\bar{\mu}(W)$ reduces when $p_{a_1}$ increases, can be shown by an illustration similar to that in Figure \ref{fig:illust1}.\qed
\item We first show that the point $P$ is within $D$. Observe that the value of $p$ reduces from $p_{a_2}^*$ by construction. So $p_{a_2}^*\leq b_1\Rightarrow p\leq b_1$. Further, $p_{a_1}\leq p_{a_1}^*$ implies $p_{a_1}\leq b_2$, both of which imply $a_1\in[0,1]$. All these imply that $P_1\in[c_1,c_1+b_1]$ and $P_2\leq c_2+b_2$. It remains to show that $P_2\geq c_2$. From (\ref{eqn:P-Q-new}), we have $P_2\leq c_2$ iff $p_{a_1}\leq 2b_2-c_2$. We now show that $p_{a_1}\leq 2b_2-c_2$.

We know that $\bar{\mu}(W)=0$ occurs when $p_{a_1}$ solves (\ref{eqn:pa1-(b)}). Substituting $p_{a_1}=(2b_2-c_2)/3$, we obtain LHS$=-\frac{8}{3}b_2c_1(b_2+c_2)<0$, and substituting $p_{a_1}=2b_2-c_2$, we obtain LHS$=2b_2(2b_2(b_1+c_1)-c_2(b_1+3c_1))$. Observe that $c_i,b_i$ that satisfy (\ref{eqn:c1-c2-small}) also satisfy $(2b_2(b_1+c_1)-c_2(b_1+3c_1))\geq 0$. Thus for the cases under consideration, the value of $p_{a_1}$ that satisfies $\bar{\mu}(W)=0$ is at most $2b_2-c_2$.

To prove that the optimal mechanism is as in Figure \ref{fig:b}, we construct the dual variable $\gamma$ similar to the proof of Proposition \ref{prop:known}.\qed
\item A decrease in $p$ continues to decrease $\bar{\mu}(W)$. This can be observed as follows. A decrease in $p$ only removes negative $\bar{\mu}$-measure from $Z$, and thus $\bar{\mu}(Z)$ increases. So $\bar{\mu}(W)$ must decrease to satisfy $\bar{\mu}(D)=0$.\qed
\item The claim holds by symmetry.\qed
\end{enumerate}

{\bf Proof of Lemma \ref{lem:DDT-equiv}:} Consider $\theta\succeq_{1}\bar{\mu}_-^W$ to hold. Strassen's theorem \cite[Thm.~6.3]{DDT13} states that under this condition, there exists a joint measure $\gamma$ such that its marginals  $\gamma_1$ and $\gamma_2$ equal $\theta$ and $\bar{\mu}_-^W$ respectively, and $z\geq z'$ (component-wise) holds $\gamma(z,z')$-a.e. We now verify if this $\gamma$ satisfies the duality conditions, and conditions in Lemma \ref{lem:conditions}.
\begin{enumerate}
 \item[(a)] $\theta\succeq_{cvx}\bar{\mu}_+^W$ holds by assumption (ii) in the lemma. Thus we have
 $$
   \gamma_1-\gamma_2=(\theta-\bar{\mu}_-^W)\succeq_{cvx}(\bar{\mu}_+^W-\bar{\mu}_-^W)=\bar{\mu}^W.
 $$
 \item[(b)] When $z\in W$, we have $u(z)=\|z\|_1-k$ for some constant $k$. So,
$$
  \int_Wu\,d(\gamma_1-\gamma_2-\bar{\mu}^W)=\int_Wu\,d(\theta-\bar{\mu}_-^W-\bar{\mu}^W)=\int_W(\|x\|_1-k)\,d(\theta-\bar{\mu}_+^W).
$$
 $\theta(W)=\bar{\mu}_+(W)$ is an easy consequence of the assumption $\theta\succeq_{cvx}\bar{\mu}_+^W$. So we have
$$
  \int_Wu\,d(\gamma_1-\gamma_2-\bar{\mu}^W)=\int_W\|z\|_1\,d(\theta-\bar{\mu}_+^W)-k(\theta(W)-\bar{\mu}_+(W))=0
$$
where the last equality holds by assumption (iii) in the lemma.
 \item[(c)] $u(z)=\|z\|_1-k$ for some constant $k$, if $z\in W$. So $u(z)-u(z')=\|z-z'\|_1$ for every $\{z,z'\in W:z\geq z'\}$. But $z\geq z'$ holds $\gamma(z,z')$-a.e. So, $u(z)-u(z')=(z_1-z_1')+(z_2-z_2')=\|z-z'\|_1$ holds $\gamma$-a.e.
\end{enumerate}\qed

{\bf Proof of Lemma \ref{lem:gc1-W}:}
\begin{figure}
\centering
\begin{minipage}{.46\textwidth}
\centering
\includegraphics[height=5cm,width=5.5cm]{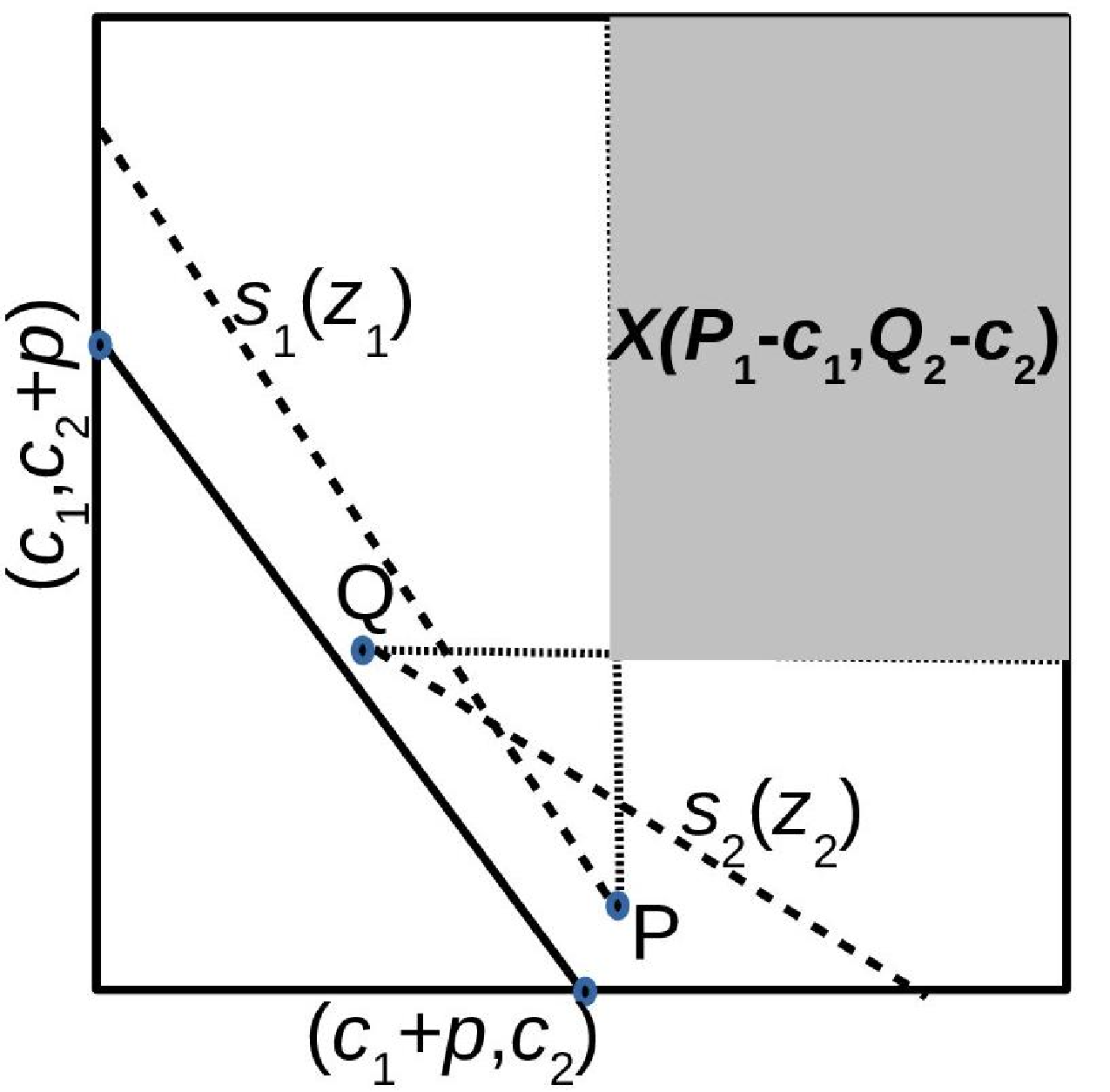}
\caption{An illustration of a possible $s_i(z_i)$ that satisfies the assumptions in Lemma \ref{lem:gc1-W}.}\label{fig:illust-gc1-W}
\end{minipage}
\begin{minipage}{.03\textwidth}
\hspace*{.03\textwidth}
\end{minipage}
\begin{minipage}{.46\textwidth}
\centering
\includegraphics[height=5cm,width=5.5cm]{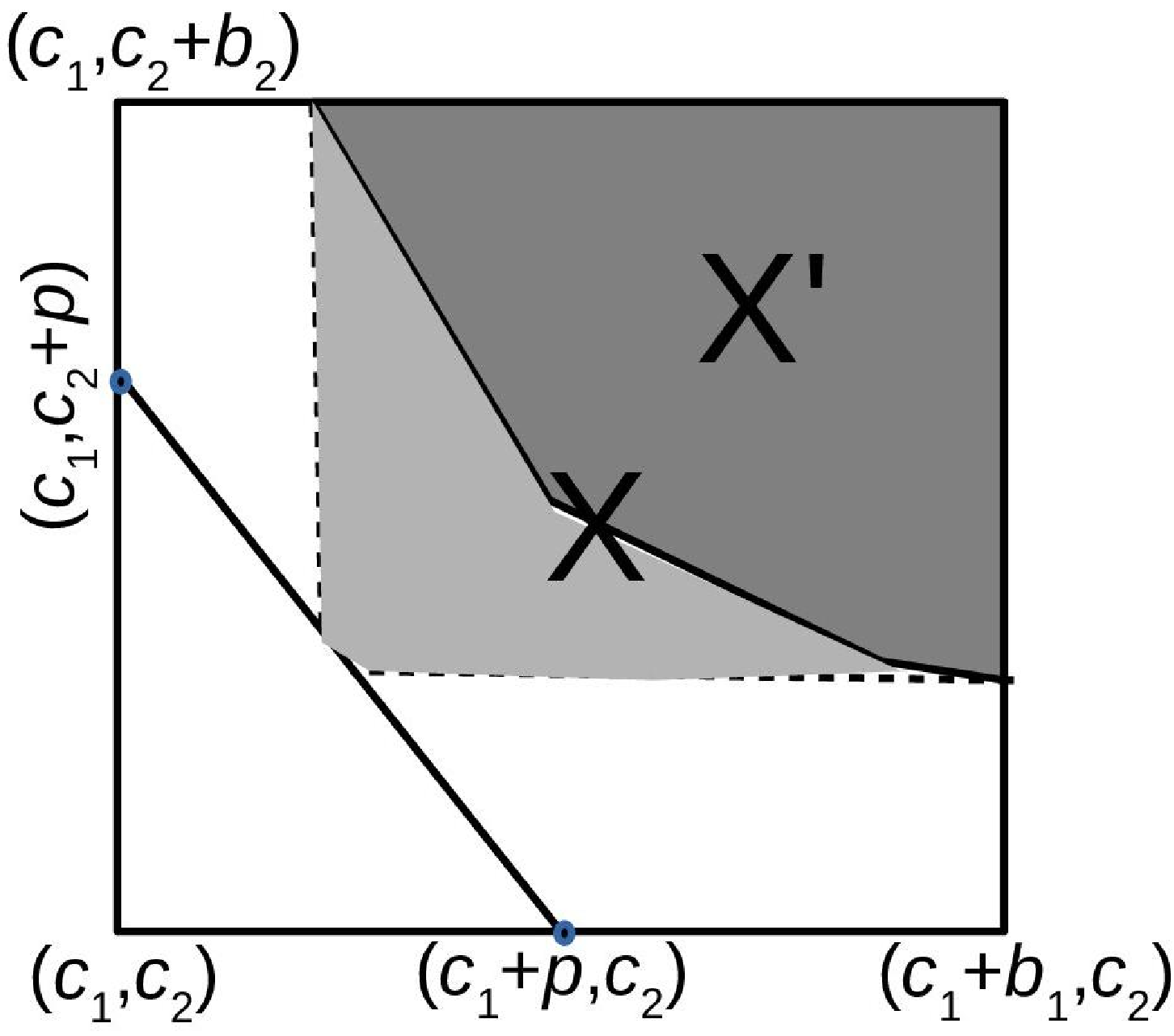}
\caption{Constructing an increasing set with boundaries parallel to the axes.}\label{fig:x-x'}
\end{minipage}
\end{figure}

The assumption in the lemma is that the function $s_1(z_1)$ is always above the lines $z_1+z_2=c_1+c_2+p$ and $z_2=c_2$, and the function $s_2(z_2)$ is always to the right of the lines $z_1+z_2=c_1+c_2+p$ and $z_1=c_1$. See Figure \ref{fig:illust-gc1-W} for an example. Defining $X(t_1,t_2):=W\cap\{z_1\geq c_1+t_1, z_2\geq c_2+t_2\}$, we now proceed to prove that $(\bar{\mu}+\bar{\alpha})(X(t_1,t_2))\geq 0$ for every $t_i\in[0,b_i]$.

We start with $t_2=0$, i.e., we first show that $(\bar{\mu}+\bar{\alpha})(X(t_1,0))\geq 0$ holds for every $t_1\in[0,b_1]$. We first note that $X(t_1,0)=W\backslash\{z_1<c_1+t_1\}$. Recall that $s_1(z_1)$ is defined as the point at which the integral of the densities of $\bar{\mu}+\bar{\alpha}$ (from $c_2+b_2$) vanishes. As $s_1(z_1)$ lies above the line $z_1+z_2=c_1+c_2+p$ for each $z_1$, we have $(\bar{\mu}+\bar{\alpha})^W(\{z_1<c_1+t_1\})$ to be an integral of negative numbers in each vertical line, when $t_1\in[0,b_1)$. So we have
$$
  (\bar{\mu}+\bar{\alpha})(X(t_1,0))=(\bar{\mu}+\bar{\alpha})(W\backslash\{z_1<c_1+t_1\})=0-(\mbox{a negative number})\geq 0
$$
for every $t_1\in[0,b_1)$. When $t_1=b_1$, we have $(\bar{\mu}+\bar{\alpha})(X(b_1,0))=\bar{\mu}_+(\{c_1+b_1\}\times[c_2,c_2+b_2])\geq 0$. Thus $(\bar{\mu}+\bar{\alpha})(X(t_1,t_2))\geq 0$ holds for every $(t_1,t_2)\in([0,b_1],0)$. The proof when $(t_1,t_2)\in(0,[0,b_2])$ is analogous.

We now show that $(\bar{\mu}+\bar{\alpha})(X(t_1,t_2))\geq 0$ holds for $t_1,t_2>0$ as well. We have four cases.
\begin{itemize}
 \item Consider $t_1+t_2\leq p^*$. Then, $X(t_1,t_2)$ contains portions of the line $z_1+z_2=c_1+c_2+p^*$, and we have
 \begin{align*}
   &(\bar{\mu}+\bar{\alpha})(X(t_1,t_2))\\&=(\bar{\mu}+\bar{\alpha})(X(t_1,0))+(\bar{\mu}+\bar{\alpha})(X(0,t_2))-(\bar{\mu}+\bar{\alpha})(W)\\&\geq 0.
 \end{align*}
 \item Consider $t_1+t_2>p^*$, and $(t_1,t_2)=(P_1-c_1,Q_2-c_2)$ This refers to the shaded region in Figure \ref{fig:illust-gc1-W}. Observe that when $t_1\geq(2b_2-c_2)/3$ or $t_2\geq(2b_1-c_1)/3$, the integral of the densities of $\bar{\mu}+\bar{\alpha}$ is nonnegative in each vertical (or horizontal) line of $X(t_1,t_2)$, and thus $(\bar{\mu}+\bar{\alpha})(X(t_1,t_2))\geq 0$ holds trivially. So we consider $P_1-c_1\leq(2b_1-c_1)/3$ and $Q_2-c_2\leq(2b_2-c_2)/3$. Now we have
 \begin{align*}
  &(\bar{\mu}+\bar{\alpha})(X(P_1-c_1,Q_2-c_2))\\
  &=\bar{\mu}([P_1,c_1+b_1)\times[\frac{2}{3}(c_2+b_2),c_2+b_2])
  \\&\hspace*{.5in}+\bar{\mu}([\frac{2}{3}(c_1+b_1),c_1+b_1]\times[Q_2,\frac{2}{3}(c_2+b_2)])\\
  &\hspace{1in}+\mu_s(\{c_1+b_1\}\times[\frac{2}{3}(c_2+b_2),c_2+b_2])
  \\&\hspace*{1.5in}+\bar{\mu}([P_1,\frac{2}{3}(c_1+b_1)]\times[Q_2,\frac{2}{3}(c_2+b_2)])\\
  &=\mu_s(\{c_1+b_1\}\times[2(c_2+b_2)/3,(c_2+b_2)])
  \\&\hspace*{1in}+\bar{\mu}([P_1,2(c_1+b_1)/3]\times[Q_2,2(c_2+b_2)/3])\\
  &\geq(c_1+b_1)(c_2+b_2)/3-3(2(c_1+b_1)/3-Q_1)(2(c_2+b_2)/3-P_2)\\
  &=(c_1+b_1)(c_2+b_2)/3-3((c_1+p_{a_2}^*)/2
  \\&\hspace*{1in}-(c_1+b_1)/3)((c_2+p_{a_1}^*)/2-(c_2+b_2)/3)\\
  &\geq(c_1+b_1)(c_2+b_2)/3-3(c_1+b_1)/6(c_2+b_2)/6\\
  &\geq0
 \end{align*}
 where the second equality here holds because $s_1(z_1)=2(c_2+b_2)/3$ in the interval $[P_1,c_1+b_1)$, and thus $\bar{\mu}([P_1,c_1+b_1]\times[2(c_2+b_2)/3,(c_2+b_2)])$ is an integral of zero measures in each vertical line; an analogous argument holds for $\bar{\mu}([2(c_1+b_2)/3,(c_2+b_1)]\times[Q_2,2(c_2+b_2)/3])$; the first inequality holds because we have $P_2\leq Q_2$ and $Q_1\leq P_1$ in the overlapping case; the third equality follows by substitution of $P_2$ and $Q_1$; and the second inequality follows from $p_{a_i}^*\leq b_{-i}$.
 
 We now consider the case when $t_1+t_2>p^*$ and $\{t_1\in[P_1-c_1,(2b_2-c_2)/3],t_2\in[Q_2-c_2,(2b_1-c_1)/3]\}$. We have
 \begin{align*}
   (\bar{\mu}+\bar{\alpha})(X(t_1,t_2))&=(\bar{\mu}+\bar{\alpha})(X(P_1-c_1,Q_2-c_2))\\&\hspace*{.2in}-\bar{\mu}([P_1,c_1+t_1]\times[2(c_2+b_2)/3,(c_2+b_2)])\\&\hspace*{.2in}-\bar{\mu}([2(c_1+b_1)/3,c_1+b_1]\times[Q_2,c_2+t_2])\\&\hspace*{1.5in}-(\mbox{a negative number})\geq 0.
 \end{align*}
 where the inequality holds because $s_1(z_1)=2(c_2+b_2)/3$ in the interval $[P_1,c_1+t_1]$, and thus $\bar{\mu}([P_1,c_1+t_1]\times[2(c_2+b_2)/3,(c_2+b_2)])$ is an integral of zero measures in each vertical line; an analogous argument holds for $\bar{\mu}([2(c_1+b_1)/3,(c_1+b_1)]\times[Q_2,c_2+t_2])$.
 \item Consider $t_1+t_2>p^*$, $t_1\geq P_1-c_1$ but $t_2\leq Q_2-c_2$. Let $c_1+t_1\geq s_2(c_2)$. Recall that the density of $\bar{\alpha}$ is the lowest at $z_2=c_2$, and thus $s_2(z_2)$ is maximized at $z_2=c_2$. So for any $t_2\in[0,Q_2-c_2]$, the integral of the densities of $\bar{\mu}+\bar{\alpha}$ is nonnegative in each horizontal line, and thus $(\bar{\mu}+\bar{\alpha})(X(t_1,t_2))\geq 0$.
 
 Let $c_1+t_1<s_2(c_2)$. When $c_2+t_2\in[s_2^{-1}(c_1+t_1),Q_2]$, we have
  \begin{multline*}
    (\bar{\mu}+\bar{\alpha})(X(t_1,t_2))=(\bar{\mu}+\bar{\alpha})(X(t_1,Q_2-c_2))\\+(\bar{\mu}+\bar{\alpha})([c_1+t_1,c_1+b_1]\times[c_2+t_2,Q_2])\geq 0
  \end{multline*}
  where the last inequality occurs because the integral of the densities of $(\bar{\mu}+\bar{\alpha})$ is nonnegative in each horizontal line for $z_2\in[c_2+t_2,Q_2]$. When $t_2\in((p^*-t_1)_+,s_2^{-1}(c_1+t_1)-c_2)$, the integral of the densities of $\bar{\mu}+\bar{\alpha}$ is nonpositive in each horizontal line for $z_2\leq c_2+t_2$, and thus
  $$
    (\bar{\mu}+\bar{\alpha})(X(t_1,t_2))=(\bar{\mu}+\bar{\alpha})(X(t_1,(p^*-t_1)_+))-(\mbox{a negative number})\geq 0.
  $$
 \item Consider $t_1+t_2>p^*$ and $t_1\leq P_1-c_1$. When $c_2+t_2\in[s_1(c_1+t_1),c_2+b_2]$, we have
  \begin{multline*}
    (\bar{\mu}+\bar{\alpha})(X(t_1,t_2))=(\bar{\mu}+\bar{\alpha})(X(P_1-c_1,t_2))\\+(\bar{\mu}+\bar{\alpha})([c_1+t_1,P_1]\times[c_2+t_2,c_2+b_2])\geq 0
  \end{multline*}
  where the last inequality occurs because the integral of the densities of $(\bar{\mu}+\bar{\alpha})$ is nonnegative in each vertical line for $z_1\in[c_1+t_1,P_1]$. When $t_2\in((p^*-t_1)_+,s_1(c_1+t_1)-c_2)$, the integral of the densities of $\bar{\mu}+\bar{\alpha}$ is nonpositive in each vertical line for $z_1\leq c_1+t_1$, and thus
  $$
    (\bar{\mu}+\bar{\alpha})(X(t_1,t_2))=(\bar{\mu}+\bar{\alpha})(X((p^*-t_2)_+,t_2))-(\mbox{a negative number})\geq 0.
  $$
\end{itemize}

We have thus shown that $(\bar{\mu}+\bar{\alpha})(X(t_1,t_2))\geq 0$ for every $t_i\in[0,b_i]$. We use this to show that $\bar{\mu}^W+\bar{\alpha}\succeq_1 0$. We first note that $\bar{\mu}^W+\bar{\alpha}\succeq_1 0$ holds if $(\bar{\mu}^W+\bar{\alpha})(X')\geq 0$ for every increasing set $X'$ \cite[Chap.~6]{SS07}. So we now show that if $(\bar{\mu}+\bar{\alpha})(X(t_1,t_2))\geq 0$ holds for every $t_i\in[0,b_i]$, then $(\bar{\mu}+\bar{\alpha})^W(X')\geq 0$ also holds for every increasing set $X'$.

Consider the increasing set $X'$ in Figure \ref{fig:x-x'}. We construct $X\supseteq X'$ such that the boundaries consist of lines parallel to the axes $z_1$ and $z_2$. Observe that such a construction is possible for any increasing $X'$. The measure $\bar{\mu}$ is negative at the interior, and thus we have $(\bar{\mu}^W+\bar{\alpha})(X')\geq(\bar{\mu}^W+\bar{\alpha})(X)$. So $(\bar{\mu}^W+\bar{\alpha})(X)\geq 0\Rightarrow(\bar{\mu}^W+\bar{\alpha})(X')\geq 0$.\qed

\section{Proofs from Section \ref{SUB:GC2}}\label{app:b}

\textbf{Proof of Proposition \ref{prop:gc2}:} By Proposition \ref{prop:cvx}, we know that $\bar{\beta}\succeq_{cvx}0$ if $\bar{\beta}(\tilde{D}^{(1)})=\int_{[c_1,c_1+p]}x\,\bar{\beta}(dx,c_2+b_2)=0$. We now find the conditions on $(p_{a_1},a_1,p)$ for both of these equations to hold.
\begin{align*}
  \bar{\beta}(\tilde{D}^{(1)})&=\frac{2b_2-c_2-3p_{a_1}}{b_1b_2}\frac{p_{a_1}}{a_1}\\&\hspace*{.5in}+\frac{3a_1}{b_1b_2}\frac{p_{a_1}^2}{2a_1^2}+\frac{2b_2}{b_1b_2}\left(p-\frac{p_{a_1}}{a_1}\right)+\frac{c_1}{b_1b_2}(b_2-p_{a_1})\\&=\frac{c_1(b_2-p_{a_1})}{b_1b_2}-\frac{c_2p_{a_1}}{a_1b_1b_2}-\frac{3p_{a_1}^2}{2a_1b_1b_2}+\frac{2b_2p}{b_1b_2}=0\\\Rightarrow p&=\frac{\frac{3p_{a_1}^2}{2a_1}+\frac{c_2p_{a_1}}{a_1}-c_1(b_2-p_{a_1})}{2b_2}.
\end{align*}
\begin{align*}
  &\int_{[c_1,c_1+p]}(x-c_1)\,\bar{\beta}(dx,c_2+b_2)\\&=\frac{1}{b_1b_2}\left((2b_2-c_2-3p_{a_1})\frac{p_{a_1}^2}{2a_1^2}+a_1\frac{p_{a_1}^3}{a_1^3}+b_2p^2-b_2\frac{p_{a_1}^2}{a_1^2}\right)\\&=\frac{1}{b_1b_2}\left(b_2p^2-(c_2+p_{a_1})\frac{p_{a_1}^2}{2a_1^2}\right)=0\Rightarrow p=\frac{p_{a_1}}{a_1}\sqrt{\frac{(p_{a_1}+c_2)}{2b_2}}.
\end{align*}
Thus if the parameters $(p_{a_1},p,a_1)$ satisfy
$$
  p=\frac{3p_{a_1}^2/(2a_1)+c_2p_{a_1}/a_1-c_1(b_2-p_{a_1})}{2b_2}=\frac{p_{a_1}}{a_1}\sqrt{\frac{p_{a_1}+c_2}{2b_2}},
$$
then, $\bar{\beta}$ satisfies $\bar{\beta}(\tilde{D}^{(1)})=\int_{[c_1,c_1+p]}x\,\bar{\beta}(dx,c_2+b_2)=0$. We also have $\int_{\tilde{D}^{(1)}}f\,d\bar{\beta}=0$ for any affine function $f(x)=\beta_1x+\beta_2$.

Now, imposing $\bar{\mu}(W)=0$, we easily derive $p=(b_1-c_1)/2$. We now substitute $p$ in (\ref{eqn:p-a1-pa1}), and obtain
$$
  a_1=p_{a_1}\left(\frac{\frac{3}{2}p_{a_1}+c_2}{b_1b_2-c_1p_{a_1}}\right)=\frac{p_{a_1}}{b_1-c_1}\sqrt{\frac{2(p_{a_1}+c_2)}{b_2}}.
$$
Rearranging, squaring, and simplifying, we have the following cubic equation in $p_{a_1}$.
\begin{multline*}
  p_{a_1}^3(2c_1^2)+p_{a_1}^2(2c_1^2c_2-4b_1b_2c_1-9b_2(b_1-c_1)^2/4)\\+p_{a_1}(2b_1^2b_2^2-4b_1b_2c_1c_2-3c_2b_2(b_1-c_1)^2)+2b_1^2b_2^2c_2-c_2^2b_2(b_1-c_1)^2=0.
\end{multline*}
So if $p_{a_1}$ is a solution to (\ref{eqn:pa1-4d}), and $a_1$ satisfies (\ref{eqn:a1-4d}), then the conditions $\bar{\beta}\succeq_{cvx}0$ and $\int_{\tilde{D}^{(1)}}f\,d\bar{\beta}=0$ for any affine $f$, are satisfied.\qed

\vspace*{10pt}
{\bf Proofs of Claims}
\begin{enumerate}
 \item We rearrange (\ref{eqn:p-a1-pa1}) as follows, to derive an expression for $a_1$.
\begin{multline}\label{eqn:a1-4d-ini}
  p_{a_1}\sqrt{2b_2(p_{a_1}+c_2)}=3p_{a_1}^2/2+c_2p_{a_1}-c_1a_1(b_2-p_{a_1})\\\Rightarrow a_1=\frac{p_{a_1}^2/2+p_{a_1}\left(c_2+p_{a_1}-\sqrt{2b_2(c_2+p_{a_1})}\right)}{c_1(b_2-p_{a_1})}.
\end{multline}
Using the expression for $a_1$, we write $p_{a_1}/a_1$ as follows.
\begin{equation}\label{eqn:pa1/a1-4d}
  \frac{p_{a_1}}{a_1}=\frac{c_1(b_2-p_{a_1})}{p_{a_1}/2+c_2+p_{a_1}-\sqrt{2b_2(c_2+p_{a_1})}}
\end{equation}
Using the expression for $p_{a_1}/a_1$, we write $p$ as follows.
\begin{equation}\label{eqn:p-4d}
  p=\frac{p_{a_1}}{a_1}\sqrt{\frac{p_{a_1}+c_2}{2b_2}}=\frac{c_1(b_2-p_{a_1})}{p_{a_1}\sqrt{\frac{b_2}{2(c_2+p_{a_1})}}+\sqrt{2b_2(c_2+p_{a_1})}-2b_2}.
\end{equation}
Substituting $p_{a_1}=b_2$ in (\ref{eqn:a1-4d-ini}), (\ref{eqn:pa1/a1-4d}), and (\ref{eqn:p-4d}), we obtain $a_1=\infty$, and $p_{a_1}/a_1=p=0$.\qed
 \item Consider (\ref{eqn:a1-4d-ini}). The term $\left(c_2+p_{a_1}-\sqrt{2b_2(c_2+p_{a_1})}\right)$, and its differential w.r.t. $p_{a_1}$, are both nonnegative when $p_{a_1}\in[(2b_2-c_2)_+,b_2]$. So, as $p_{a_1}$ decreases from $b_2$ to $(2b_2-c_2)_+$, the numerator decreases and the denominator increases, resulting in an decrease of $a_1$.
 
 Consider (\ref{eqn:pa1/a1-4d}). By the same reasoning as above, $p_{a_1}/a_1$ increases as $p_{a_1}$ decreases.
 
 Consider (\ref{eqn:p-4d}). The differential of the term $p_{a_1}\sqrt{\frac{b_2}{2(c_2+p_{a_1})}}$ w.r.t. $p_{a_1}$, and the term $\sqrt{2b_2(c_2+p_{a_1})}-2b_2$, are both nonnegative when $p_{a_1}=[(2b_2-c_2)_+,b_2]$. So, as $p_{a_1}$ decreases, the numerator increases and the denominator decreases, resulting in an increase of $p$. Hence the result.\qed
 \item We just showed that $p$ increases as $p_{a_1}$ decreases. So it suffices to prove that $p_{a_1}=(2b_2-c_2)_+$ results in $p\geq(b_1-c_1)/2$. We first consider the case when $2b_2-c_2\geq 0$. At $p_{a_1}=2b_2-c_2$, we have
$$
  a_1=(2b_2-c_2)^2/(2c_1(c_2-b_2));\,p_{a_1}/a_1=p=2c_1(c_2-b_2)/(2b_2-c_2).
$$
We now show that $2c_1\frac{c_2-b_2}{2b_2-c_2}\geq\frac{b_1-c_1}{2}$ is true. Rearranging, we get $c_2(b_1+3c_1)-2b_2(c_1+b_1)\geq 0$, which holds when $c_2\geq 2b_2(b_1+c_1)/(b_1+3c_1)$.

We now consider the case when $2b_2-c_2\leq 0$. At $p_{a_1}=(2b_2-c_2)_+=0$, we have
$$
  a_1=0;\,p_{a_1}/a_1=c_1b_2/(c_2-\sqrt{2b_2c_2});\,p=c_1b_2/(\sqrt{2b_2c_2}-2b_2).
$$
We now show that $\frac{c_1b_2}{\sqrt{2b_2c_2}-2b_2}\geq\frac{b_1-c_1}{2}$ is true. Rearranging and squaring on both sides, we get $c_2\leq 2b_2(b_1/(b_1-c_1))^2$, which holds by the assumption in the claim. This proves that there exist a $p_{a_1}\in[(2b_2-c_2)_+,b_2]$.

When $p_{a_1}\geq (2b_2-c_2)_+$, it is clear from (\ref{eqn:p-4d}) that $p_{a_1}/a_1\leq p$.\qed
 \item We first split $W$ into the following three regions:
\begin{align}
 W_1&:[(c_1+b_1)/2,c_1+b_1)\times\{\{c_2\}\cup[c_2+2b_2/3,c_2+b_2]\},\nonumber\\
 W_2&:[2(c_1+b_1)/3,c_1+b_1]\times(c_2,c_2+2b_2/3],\nonumber\\
 W_3&:\{[(c_1+b_1)/2,2(c_1+b_1)/3]\times(c_2,c_2+2b_2/3]\}\nonumber\\&\hspace*{1in}\cup\{\{c_1+b_1\}\times[c_2+2b_2/3,c_2+b_2]\}.\label{eqn:W-4d}
\end{align}

We now set up the (dual variable) $\gamma$ measure as follows. First, let $\gamma_1:=\gamma_1^Z+\gamma_1^{D\backslash Z}$, with $\gamma_1^Z=\bar{\mu}^Z$, and $\gamma_1^{D\backslash Z}=(\bar{\mu}^{D\backslash Z}+\bar{\beta})_+$. So $\gamma_1$ is supported on $(\partial D^+\cup Z)$. Next we specify the transition kernel $\gamma(\cdot~|~z)$ for $z\in(\partial D^+\cup Z)$.
\begin{enumerate}
 \item[(a)] For $z\in Z$, we define $\gamma(\cdot~|~z)=\delta_z(\cdot)$. We interpret this as the mass being retained at each $z\in Z$.
 \item[(b)] For $z\in[c_1,c_1+p_{a_1}/a_1]\times\{c_2+b_2\}$, $\gamma(\cdot~|~z)$ is defined by the uniform probability density on the line $\{z_1\}\times[s_1(z_1),c_2+b_2)$, and zero elsewhere. We interpret this as a transfer of $(\mu_s+\beta_s)(z)$ from the boundary to the above line when $z_1\in(c_1,c_1+p_{a_1}/a_1)$, and a transfer of $\beta_p(z)$ when $z_1=c_1$.
 \item[(c)] For $z\in[c_1+p_{a_1}/a_1,c_1+p]\times\{c_2+b_2\}$, we define $\gamma(x~|~z)=\mu(x)/(\mu_s+\beta_s)(z)$ when $x\in\{z_1\}\times(c_2,c_2+b_2)$, $\mu_s(x)/(\mu_s+\beta_s)(z)$ when $x=(z_1,c_2)$, and zero otherwise. We interpret this as a transfer of $(\mu_s+\beta_s)(z)$ from the boundary to the above line.
 \item[(d)] For $z\in W_1\cap\partial D^+$, we define $\gamma(x~|~z)=\mu(x)/\mu_s(z)$ when $x=\{z_1\}\times[c_2+2b_2/3,c_2+b_2)$, $\mu_s(x)/\mu_s(z)$ when $x=(z_1,c_2)$, and zero otherwise. Again, we interpret this as a transfer of $\mu_s(z)$ from the boundary to the above line.
 \item[(e)] For $z\in W_2\cap\partial D^+$, $\gamma(\cdot~|~z)$ is defined by the uniform probability density on the line $[2(c_1+b_1)/3,c_1+b_1)\times\{z_2\}$, and zero elsewhere. Again, we interpret this as a transfer of $\mu_s(z)$ from the boundary to the above line.
 \item[(f)] For $z\in W_3\cap\partial D^+$, $\gamma(\cdot~|~z)$ is defined as follows. The total mass $\mu_s(W_3\cap\partial D^+)$ is spread uniformly on $W_3\backslash\partial D^+$ with equal contribution from each $z$ in $W_3\cap\partial D^+$.
\end{enumerate}

We then define $\gamma(F)=\int_{(z,z')\in F}\gamma_1(dz)\gamma(dz'~|~z)$ for any measurable $F\in D\times D$. It is then easy to check, by virtue of the choices of $s_1(z_1)$, $Z$, and the matchings in (a)--(f), that $\gamma_2^Z=\gamma(Z,\cdot)=\bar{\mu}^Z$, and $\gamma_2^{D\backslash Z}=\gamma(D\backslash Z,\cdot)=(\bar{\mu}^{D\backslash Z}+\bar{\beta})_-$. Thus $(\gamma_1-\gamma_2)^Z=0$ and $(\gamma_1-\gamma_2)^{D\backslash Z}=\bar{\mu}^{D\backslash Z}+\bar{\beta}$.

We now define the allocation $q(z)$ to be
$$
  q(z)=\begin{cases}(0,0)&\mbox{if }z\in Z\\(a_1,1)&\mbox{if }z\in A\\(1,1)&\mbox{if }z\in W\end{cases}
$$
The primal variable $u(z)$ can be derived from this, since $\nabla u =q$. Verifying if $u$ and $\gamma$ satisfy all the required constraints is now similar to the verifications in the proof of Proposition \ref{prop:known}. The optimal mechanism thus is as in Figure \ref{fig:d}.\qed
\end{enumerate}

\begin{figure}
\centering
\begin{tabular}{cc}
\subfloat[]{\label{fig:5b-5c-case1}\includegraphics[height=5cm, width=5.5cm]{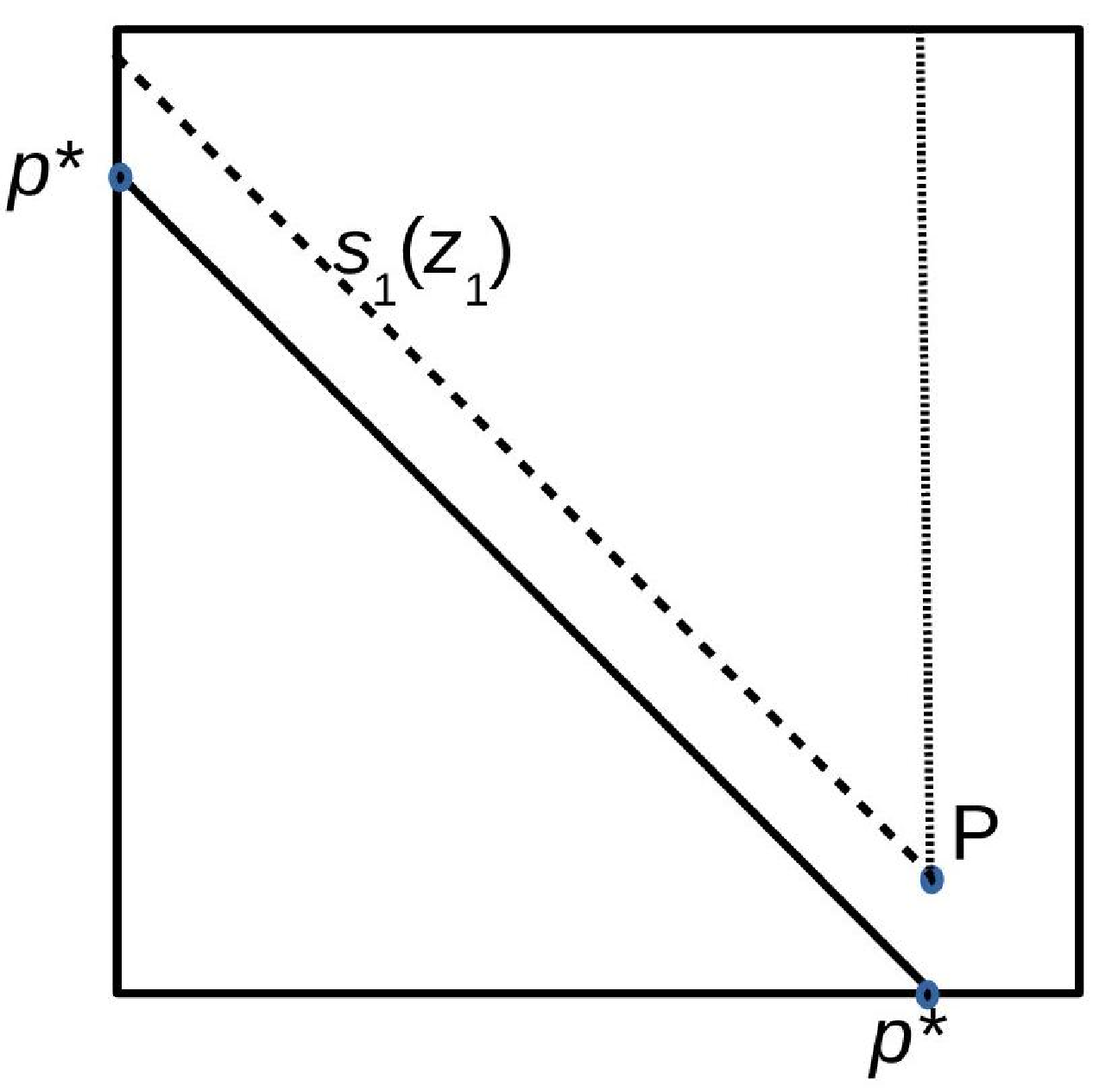}}&
\subfloat[]{\label{fig:5b-5c-case2}\includegraphics[height=5cm, width=5.5cm]{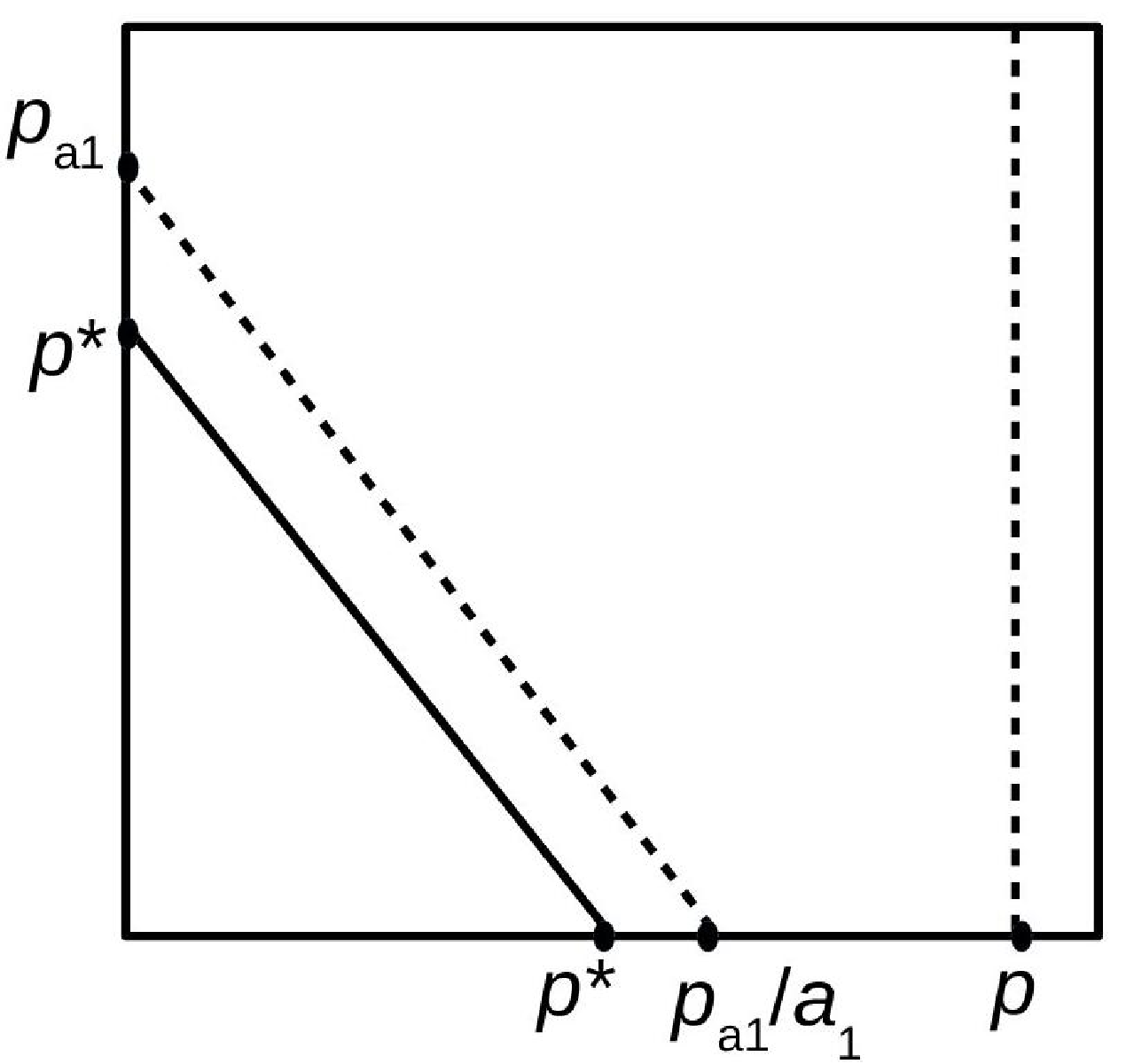}}
\end{tabular}
\caption{The function $s_1(z_1)$, when (a) $p_{a_1}$ reaches $(2b_2-c_2)$ at $a_1>1$ (b) $a_1$ reaches $1$ at $p_{a_1}>(2b_2-c_2)$.}\label{fig:5b-5c}
\end{figure}

\subsection*{Proof of Claim 5, Section \ref{SUB:GC2}:}
Consider the case where $a_1>1$ when $p=(b_1-c_1)/2$. Then starting from the $p_{a_1}$ that led to $p=(b_1-c_1)/2$, we continue to decrease $p_{a_1}$, and modify $a_1$ and $p$ according to (\ref{eqn:p-a1-pa1}). We stop either when $p_{a_1}$ reaches $2b_2-c_2$ or when $a_1$ reaches $1$. We thus have two cases.

\vspace*{10pt}
{\bf Case a: $p_{a_1}$ reaches $(2b_2-c_2)$ at $a_1>1$}

In this case, we have $p_{a_1}/a_1=p=2c_1(c_2-b_2)/(2b_2-c_2)$. Observe that the structure is then the same as depicted in Figure \ref{fig:b}, with $P_2=c_2$. We now decrease $p_{a_1}$ until $p_{a_1}=p_{a_1}^*$, but modify the point $P$ according to (\ref{eqn:P-Q-new}), i.e., we use $\bar{\alpha}^{(1)}$ as the shuffling measure. The structure is now the same as depicted in Figure \ref{fig:c'}. We fix $P$ and $s_1(z_1)$ corresponding to $p_{a_1}=p_{a_1}^*$, and decrease $p$ until $p=p^*$. Observe that the procedure is similar to the procedure adopted in the proof of claim in Case 4, Section \ref{SUB:GC1}.

Figure \ref{fig:5b-5c-case1} shows the structure in Figure \ref{fig:c'} with $P$ and $s_1(z_1)$ corresponding to $p_{a_1}^*$. Observe that $P_1$ constructed thus, may be more than $c_1+b_1$. We get around this issue by ignoring $P$ for time being, continue decreasing $p_{a_1}$ until it reaches $p_{a_1}^*$, and then further decrease $p$ to $p^*$. We fix $P$ and $s_1(z_1)$ corresponding to the $p_{a_1}$ for which $P_1=c_1+b_1$. We claim that the value of $p^*$ so obtained is at most $b_1$. This can be observed as follows. $p^*$ equals $b_1$ when $c_1+c_2=b_2-(3/2)b_1$. Also, $p^*$ decreases with increase in $(c_1+c_2)$. So $p^*\leq b_1$ if $c_1+c_2\geq b_2-(3/2)b_1$. But from the theorem statement, we have $c_2\geq b_2$. So $p^*\leq b_1$ holds.

We now argue that the optimal mechanism is as in Figure \ref{fig:c'}. We do so by fixing $\theta=\bar{\mu}_+^W+\bar{\alpha}^{(1)}$, and showing that this value of $\theta$ satisfies the assumptions in Lemma \ref{lem:DDT-equiv}. Conditions (ii) and (iii) of Lemma \ref{lem:DDT-equiv} only involve verifying $\bar{\alpha}^{(1)}\succeq_{cvx}0$ and $\int_D\|x\|_1\,d\bar{\alpha}^{(1)}\geq 0$, both of which are established in Proposition \ref{prop:cvx}. The following lemma proves that $\theta$ satisfies condition (i).
\begin{lemma}\label{lem:gc2-W-case1}
Consider the structure in Figure \ref{fig:c'}. Let $s_1(z_1)$ associated with $\bar{\mu}+\bar{\alpha}^{(1)}$ satisfy $s_1(z_1)>\max(c_2,c_1+c_2+p^*-z_1)$ for every $z_1\in[c_1,c_1+b_1)$. Further, let $c_2\geq b_2$. Then, the following hold:
\begin{enumerate}
 \item[(a)] $(\bar{\mu}+\bar{\alpha}^{(1)})(W\cap\{z_1\geq c_1+t_1\}\cap\{z_2\geq c_2+t_2\})\geq 0$ for every $t_i\in[0,b_i]$.
 \item[(b)] $\theta-\bar{\mu}_-^W=\bar{\mu}^W+\bar{\alpha}^{(1)}\succeq_1 0.$
\end{enumerate}
\end{lemma}
\begin{proof}
Recall that $X(t_1,t_2)=W\cap\{z_1\geq t_1+c_1, z_2\geq t_2+c_2\}$ from the proof of Lemma \ref{lem:gc1-W}. So we must prove that $(\bar{\mu}+\bar{\alpha}^{(1)})(X(t_1,t_2))\geq 0$ for every $t_i\in[0,b_i]$. We start with $t_1=0$. Observe that the set $X(0,t_2)=W\backslash\{z_2<c_2+t_2\}$. We now prove that $(\bar{\mu}+\bar{\alpha}^{(1)})^W(\{z_2<c_2+t_2\})=\bar{\mu}^W(\{z_2<c_2+t_2\})\leq 0$ holds for every $t_2\in[0,b_2]$, whenever $c_2\geq b_2$.

Let $t_2\in[0,p^*]$. Then, $\bar{\mu}^W(\{z_2<c_2+t_2\})=-c_2(b_1-p^*)+t_2(c_1-2b_1+3p^*-3t_2/2)$. We will show that $\max_{t_2\in[0,p^*]}(-c_2(b_1-p^*)+t_2(c_1-2b_1+3p^*-3t_2/2))\leq 0$. We note that this proves the inequality for every $t_2\in[0,p^*]$. The maximum is attained at $t_2=t_2^*:=(c_1-2b_1+3p^*)/3$, provided $t_2^*\in[0,p^*]$. Assume $t_2^*\in[0,p^*]$. Then,
\begin{align*}
  \bar{\mu}^W(\{z_2<c_2+t_2^*\})&=-c_2(b_1-p^*)+(c_1-2b_1+3p^*)^2/6\\&=b_1(b_2-c_2-2p^*)+(2b_1-c_1)^2/6.
\end{align*}
The last equality uses $3(p^*)^2/2+p^*(c_1+c_2)-b_1b_2=0$. By our assumption that $t_2^*\geq 0$, we have $p^*\geq(2b_1-c_1)/3$. We also have $c_2\geq b_2$ by assumption in the lemma. Using both of these inequalities, we have
$$
  b_1(b_2-c_2-2p^*)+(2b_1-c_1)^2/6\leq b_1(b_2-c_2)-2b_1(2b_1-c_1)/3+(2b_1-c_1)^2/6\leq 0.
$$

Now consider $t_2^*\notin[0,p^*]$. Then the maximum of $(-c_2(b_1-p^*)+t_2(c_1-2b_1+3p^*-3t_2/2))$ is attained either at $t_2=0$ or $t_2=p$. It is attained at $t_2=0$ if $p^*<(2b_1-c_1)/3$, and at $t_2=p$ if $c_2>2b_2$. But $\bar{\mu}^W(\{z_2<c_2\})=0$, and
$$
  \bar{\mu}^W(\{z_2<c_2+p^*\})=-c_2(b_1-p^*)+p^*(c_1-2b_1+3p^*/2)=b_1(b_2-c_2-2p^*)\leq 0
$$
where the last equality follows from $3(p^*)^2/2+p^*(c_1+c_2)-b_1b_2=0$, and the last inequality follows from $c_2\geq b_2$. We have thus shown that $\bar{\mu}^W(\{z_2<c_2+t_2\})\leq 0$ for every $t_2\in[0,p^*]$.

Let $t_2\in[p^*,b_2]$. Then, $\bar{\mu}^W(\{z_2<c_2+t_2\})=\bar{\mu}^W(\{z_2<c_2+p^*\})-2b_1(t_2-p^*)\leq\bar{\mu}^W(\{z_2<c_2+p^*\})\leq 0$. So $(\bar{\mu}+\bar{\alpha}^{(1)})(X(t_1,t_2))\geq 0$ holds when $(t_1,t_2)\in(0,[0,b_2])$.

The steps for the remaining part of the proof is as follows.
\begin{itemize}
 \item Prove that $(\bar{\mu}+\bar{\alpha}^{(1)})(X(t_1,0))\geq 0$ holds for any $t_1\in[0,b_1]$. This is exactly the same as the proof in Lemma \ref{lem:gc1-W}.
 \item Prove that $(\bar{\mu}+\bar{\alpha}^{(1)})(X(t_1,t_2))\geq 0$ for $t_1,t_2>0$. This is done by considering three cases: $t_1+t_2\leq p^*$, and $\{t_1+t_2>p^*, t_1\geq P_1-c_1\}$ and $\{t_1+t_2>p^*, t_1\leq P_1-c_1\}$. The proof proceeds in the same way as the proof of Lemma \ref{lem:gc1-W}, but with fewer cases.
 \item Prove that $(\bar{\mu}^W+\bar{\alpha}^{(1)})\succeq_1 0$ holds when $(\bar{\mu}+\bar{\alpha})(X(t_1,t_2))\geq 0$ holds for all $t_i\in[0,b_i]$. This has been established in the proof of Lemma \ref{lem:gc1-W}.
\end{itemize}
\end{proof}

\begin{remark}\label{rem:1}
The steps of the proof indicate that $\bar{\mu}(X(0,t_2))\geq 0$ holds for all $t_2\in[0,b_2]$, if one of the following holds.
\begin{enumerate}
 \item[(a)] $p^*\leq(2b_1-c_1)/3$.
 \item[(b)] $b_1(c_2-b_2+2p^*)-((2b_1-c_1)_+)^2/6\geq 0$.
\end{enumerate}
An analogous statement is true for $\bar{\mu}(X(t_1,0))\geq 0$ to hold.
\end{remark}

{\bf Case b: $a_1$ reaches $1$ at $p_{a_1}>(2b_2-c_2)_+$}

In this case, we have $p_{a_1}>(2b_2-c_2)_+$, and thus $p_{a_1}/a_1<p$ from (\ref{eqn:p-a1-pa1}). Define $\hat{p}_{a_1}$ as the value of $p_{a_1}$ for which $a_1=1$ occurs. We fix the points $\hat{p}_{a_1}$, $\hat{p}_{a_1}/a_1(=\hat{p}_{a_1})$, and $p$ that corresponds to $a_1=1$, and then decrease $p_{a_1}$ until $p_{a_1}=p^*$. See Figure \ref{fig:5b-5c-case2}.

Observe that $p$ that corresponds to $a_1=1$, may be more than $b_1$. In such a case, we fix $p_{a_1}$ and $p_{a_1}/a_1$ at the point where $p=b_1$. This value of $p_{a_1}$ is more than $\hat{p}_{a_1}$, by Claim 2, Section \ref{SUB:GC2}. The following lemma shows that in such a case, either $p^*\leq p_{a_1}/a_1$ holds, or $p^*$ is such that $\bar{\mu}^W\succeq_1 0$ is true.
\begin{lemma}\label{lem:p^*-mu-bar}
Suppose $c_1\leq b_1$, and $c_2\in[2b_2(b_1+c_1)/(b_1+3c_1),2b_2(b_1/(b_1-c_1))^2]$. Let $p_{a_1}$ and $a_1$ be obtained by solving (\ref{eqn:p-a1-pa1}) for $p=b_1$. Then, $p^*$ is such that $\bar{\mu}^W\succeq_1 0$, whenever $p^*\geq p_{a_1}/a_1$.
\end{lemma}
\begin{proof} We know from Remark \ref{rem:1} that both $\bar{\mu}(X(0,t_2))$ and $\bar{\mu}(X(t_1,0))\geq 0$ hold for any $t_i\in[0,b_i]$, if $c_i, b_i$ belong to the following set.
\begin{multline}\label{eqn:mu-inc-dominant-12}
  \{((b_2(c_1-b_1+2p^*)-((2b_2-c_2)_+)^2/6\geq 0)\cup(p^*\leq(2b_2-c_2)/3))\\\cap((b_1(c_2-b_2+2p^*)-((2b_1-c_1)_+)^2/6\geq 0)\cup(p^*\leq(2b_1-c_1)/3))\}
\end{multline}
We now claim that $\bar{\mu}^W\succeq_1 0$ holds for all those $c_i, b_i$. To prove our claim, we first prove that $\bar{\mu}(X(t_1,t_2))\geq 0$ for all $t_1,t_2>0$, and then point to the proof of Lemma \ref{lem:gc1-W} that it is equivalent to $\bar{\mu}^W\succeq_1 0$. We consider the following cases.
\begin{itemize}
 \item Let $t_1+t_2\leq p^*$. Then, $X(t_1,t_2)$ contains a portion of the line $z_1+z_2=c_1+c_2+p^*$. So we have,
 $$
   \bar{\mu}(X(t_1,t_2))=\bar{\mu}(X(t_1,0))+\bar{\mu}(X(0,t_2))-\bar{\mu}(W)\geq 0.
 $$
 \item Let $t_1+t_2>p^*$. Observe that the integral of densities of $\bar{\mu}$ on each vertical line of $X(t_1,t_2)$ is a constant, except on the lines $z_1=c_1$ and $z_1=c_1+b_1$. If the constant is positive, then $\bar{\mu}(X(t_1,t_2))\geq 0$ holds because it is just an integral of positive numbers. If the constant is negative, then
 $$
   \bar{\mu}(X(t_1,t_2))=\bar{\mu}(X((p^*-t_2)_+,t_2))-(\mbox{a negative number})\geq 0
 $$
\end{itemize}
Thus $\bar{\mu}(X(t_1,t_2))\geq 0$ for all $t_i\in[0,b_i]$, and from the proof of Lemma \ref{lem:gc1-W}, we have $\bar{\mu}^W\succeq_1 0$. We have proved our claim.

We observe that the $c_i, b_i$ in the statement of the lemma satisfy $b_1(c_2-b_2+2p^*)-((2b_1-c_1)_+)^2/6\geq 0$, if $p^*\geq(2b_1-c_1)/3$. This is because every $c_2$ given in the lemma is at least $b_2$, and thus we have
\begin{multline*}
  b_1(c_2-b_2+2p^*)-((2b_1-c_1)_+)^2/6\\\geq b_1(c_2-b_2)+2b_1(2b_1-c_1)/3-((2b_1-c_1)_+)^2/6\geq 0.
\end{multline*}
So it suffices to prove that $c_i, b_i$ in the statement of the lemma belong to the set
\begin{equation}\label{eqn:mu-inc-dominant-1}
 \{(b_2(c_1-b_1+2p^*)-((2b_2-c_2)_+)^2/6\geq 0)\cup(p^*\leq(2b_2-c_2)/3)\},
\end{equation}
whenever $p^*\geq p_{a_1}/a_1$.

Consider $c_2\geq 2b_2$. Then, $c_i,b_i$ satisfy (\ref{eqn:mu-inc-dominant-1}) if $p^*\geq(b_1-c_1)/2$. We now prove that $p_{a_1}/a_1\geq(b_1-c_1)/2$ for all $c_i, b_i$ given in the lemma. At $p=b_1$, we have $p_{a_1}/a_1=b_1\sqrt{2b_2/(p_{a_1}+c_2)}$, from (\ref{eqn:p-a1-pa1}). So $p_{a_1}/a_1\geq(b_1-c_1)/2$ occurs iff $p_{a_1}+c_2\leq 8b_2(b_1/(b_1-c_1))^2$. But $c_2\leq 2b_2(b_1/(b_1-c_1))^2$ holds by the assumption in the lemma, and $p_{a_1}\leq b_2$. So $p_{a_1}/a_1\geq(b_1-c_1)/2$ holds for all $c_i, b_i$ given in the lemma, and since $p^*\geq p_{a_1}/a_1$ by assumption in the lemma, we have $p^*\geq(b_1-c_1)/2$.

When $c_2<2b_2$, we have two cases. Consider $b_2\leq 2b_1$. At $p=b_1$, we have $p_{a_1}/a_1=b_1\sqrt{2b_2/(p_{a_1}+c_2)}$, from (\ref{eqn:p-a1-pa1}). So $p_{a_1}/a_1\geq(b_1-c_1)/2+(2b_2-c_2)^2/(12b_2)$ occurs iff $p_{a_1}+c_2\leq 2b_2b_1^2/((b_1-c_1)/2+(2b_2-c_2)^2/(12b_2))^2$. But for the case under consideration, we have
$$
  \frac{2b_2b_1^2}{\left(\frac{b_1-c_1}{2}+\frac{(2b_2-c_2)^2}{12b_2}\right)^2}\geq\frac{2b_2b_1^2}{\left(\frac{b_1-c_1}{2}+\frac{b_2}{12}\right)^2}\geq\frac{2b_2b_1^2}{\left(\frac{b_1}{2}+\frac{b_1}{6}\right)^2}\geq\frac{9}{2}b_2
$$
where the first inequality follows from $c_2\geq b_2$, and the second inequality follows from $b_2\leq 2b_1$. But $p_{a_1}+c_2\leq(9/2)b_2$ holds trivially, since $c_2<2b_2$. So for this case, $p_{a_1}+c_2\leq 2b_2b_1^2/((b_1-c_1)/2+(2b_2-c_2)^2/(12b_2))^2$ holds, and hence $p_{a_1}/a_1\geq(b_1-c_1)/2+(2b_2-c_2)^2/(12b_2)$ holds as well. When $p^*\geq p_{a_1}/a_1$, we have $p^*\geq(b_1-c_1)/2+(2b_2-c_2)^2/(12b_2)$, which implies that $c_i, b_i$ satisfy (\ref{eqn:mu-inc-dominant-1}).

Consider $b_2>2b_1$. We know that $c_i, b_i$ satisfy (\ref{eqn:mu-inc-dominant-1}) if $p^*\leq(2b_2-c_2)/3$. So let $p^*>(2b_2-c_2)/3$. We now derive an equivalent condition for $p^*>(2b_2-c_2)/3$.
\begin{multline*}
  (\sqrt{(c_1+c_2)^2+6b_1b_2}-c_1-c_2)/3>(2b_2-c_2)/3\\\Leftrightarrow (c_1+c_2)^2+6b_1b_2>(2b_2+c_1)^2\\
  \Leftrightarrow 2b_1b_2/(2b_2+2c_1+c_2)>(2b_2-c_2)/3.
\end{multline*}
At $p=b_1$, we have $p_{a_1}/a_1=b_1\sqrt{2b_2/(p_{a_1}+c_2)}$, from (\ref{eqn:p-a1-pa1}). So $p_{a_1}/a_1\geq(b_1-c_1)/2+(2b_2-c_2)^2/(12b_2)$ occurs iff $p_{a_1}+c_2\leq 2b_2b_1^2/((b_1-c_1)/2+(2b_2-c_2)^2/(12b_2))^2$. But for the case under consideration, we have
$$
  \frac{2b_2b_1^2}{\left(\frac{b_1-c_1}{2}+\frac{(2b_2-c_2)^2}{12b_2}\right)^2}\geq\frac{2b_2b_1^2}{\left(\frac{b_1-c_1}{2}+\frac{3b_2b_1^2}{(2b_2+2c_1+c_2)^2}\right)^2}\geq\frac{2b_2b_1^2}{\left(\frac{b_1}{2}+\frac{b_1^2}{3b_2}\right)^2}\geq\frac{9}{2}b_2
$$
where the first inequality follows from $2b_1b_2/(2b_2+2c_1+c_2)>(2b_2-c_2)/3$, the second inequality from $(c_1,c_2)\geq(0,b_2)$, and the last inequality from $b_2\geq 2b_1$. But $p_{a_1}+c_2\leq(9/2)b_2$ holds trivially, since $c_2<2b_2$. So for this case, $p_{a_1}+c_2\leq 2b_2b_1^2/((b_1-c_1)/2+(2b_2-c_2)^2/(12b_2))^2$ holds, and hence $p_{a_1}/a_1\geq(b_1-c_1)/2+(2b_2-c_2)^2/(12b_2)$ holds as well. When $p^*\geq p_{a_1}/a_1$, we have $p^*\geq(b_1-c_1)/2+(2b_2-c_2)^2/(12b_2)$, which implies that $c_i, b_i$ satisfy (\ref{eqn:mu-inc-dominant-1}).
\end{proof}

We now argue that the optimal mechanism is as in Figure \ref{fig:c'}. We fix $\theta=\bar{\mu}_+^W+\bar{\beta}$ in case $p^*\leq p_{a_1}/a_1$, or in case $p$ corresponding to $a_1=1$ is less than $b_1$. We fix $\theta=\bar{\mu}_+^W$ in case $p^*$ is such that $\bar{\mu}^W\succeq_1 0$. We show that both of these $\theta$ satisfy the assumptions in Lemma \ref{lem:DDT-equiv}.

For $\theta=\bar{\mu}_+^W$, conditions (ii) and (iii) are trivially true, and condition (i) is true because $\bar{\mu}^W\succeq_1 0$. For $\theta=\bar{\mu}_+^W+\bar{\beta}$, conditions (ii) and (iii) of Lemma \ref{lem:DDT-equiv} only involve proving $\bar{\beta}\succeq_{cvx}0$ and $\int_D\|x\|_1\,d\bar{\beta}\geq 0$, both of which are established in Proposition \ref{prop:gc2}. The following lemma proves that $\theta$ satisfies condition (i).
\begin{lemma}\label{lem:gc2-W-case2}
Consider the structure in Figure \ref{fig:c'}. Let $s_1(z_1)$ associated with $\bar{\mu}+\bar{\beta}$ satisfy $s_1(z_1)>c_1+c_2+p^*-z_1$ for every $z_1\in[c_1,c_1+p]$. Further, let $c_2\geq b_2$. Then, the following hold:
\begin{enumerate}
 \item[(a)] $(\bar{\mu}+\bar{\beta})(W\cap\{z_1\geq c_1+t_1\}\cap\{z_2\geq c_2+t_2\})\geq 0$ for every $t_i\in[0,b_i]$.
 \item[(b)] $\theta-\bar{\mu}_-^W=\bar{\mu}^W+\bar{\beta}\succeq_1 0.$
\end{enumerate}
\end{lemma}
The steps of the proof are as follows.
\begin{itemize}
 \item Prove that $(\bar{\mu}+\bar{\beta})(X(0,t_2))\geq 0$ holds for any $t_2\in[0,b_2]$. This has been proved in Lemma \ref{lem:gc2-W-case1}.
 \item Prove that $(\bar{\mu}+\bar{\beta})(X(t_1,0))\geq 0$ holds for any $t_1\in[0,b_1]$. This is similar to the proof in Lemma \ref{lem:gc1-W}, where each vertical line has a nonpositive integral of $(\bar{\mu}+\bar{\beta})^W$.
 \item Prove that $(\bar{\mu}+\bar{\beta})(X(t_1,t_2))\geq 0$ for $t_1,t_2>0$. This is done by considering three cases: $t_1+t_2\leq p^*$, $\{t_1+t_2>p^*, t_1\geq p\}$ and $\{t_1+t_2>p^*, t_1\leq p\}$. The proof proceeds in the same way as the proof of Lemma \ref{lem:gc1-W}, but with fewer cases.
 \item Prove that $\bar{\mu}^W+\bar{\beta}\succeq_1 0$ holds when $(\bar{\mu}+\bar{\beta})(X(t_1,t_2))\geq 0$ holds for every $t_i\in[0,b_i]$. This is already established in the proof of Lemma \ref{lem:gc1-W}.
\end{itemize}

This completes the proof of Claim 5, Section \ref{SUB:GC2}.\qed

\vspace*{10pt}
{\bf Proof of Claim \ref{clm:beta-eh-1}:} Substituting $p_{a_1}=0$ in (\ref{eqn:a1-4d-ini}), we have $a_1=0$. Similarly, substituting $p_{a_1}=0$ in (\ref{eqn:p-4d}), we have $p=c_1b_2/(\sqrt{2b_2c_2}-2b_2)$, and on substituting $c_2=2b_2(b_1/(b_1-c_1))^2$, we have $p=(b_1-c_1)/2$. We have proved our claim.\qed

\vspace*{10pt}
{\bf Proof of Proposition \ref{prop:beta-eh}:}
We first verify if $\bar{\beta}_e(\tilde{D}^{(1)})=0$.
$$
 \bar{\beta}_e(\tilde{D}^{(1)})=\frac{1}{b_1b_2}\left(c_1b_2+(2b_2-c_2)\frac{b_1b_2}{c_2}+2b_2\left(\frac{b_1-c_1}{2}-\frac{b_1b_2}{c_2}\right)\right)=0.
$$
Thus we have $\int_{\tilde{D}^{(1)}}u\,d\bar{\beta}_e=u\bar{\beta}_e(\tilde{D}^{(1)})=0$ for any constant $u$. Now we verify if $\int_{[c_1,(c_1+b_1)/2]}(x-c_1)\,\bar{\beta}_e(dx,c_2+b_2)\geq 0$.
\begin{multline*}
 \int_{[c_1,(c_1+b_1)/2]}(x-c_1)\,\bar{\beta}_e(dx,c_2+b_2)\\=\frac{1}{b_1b_2}\left((2b_2-c_2)\frac{(b_1b_2)^2}{2c_2^2}+2b_2\left(\frac{(b_1-c_1)^2}{8}-\frac{(b_1b_2)^2}{2c_2^2}\right)\right)\geq 0
\end{multline*}
where the last inequality occurs since $c_2\geq 2b_2(b_1/(b_1-c_1))^2$. The proof of $\bar{\beta}_e\succeq_{cvx}0$ is now similar to the proof of Proposition \ref{prop:cvx}.\qed

\vspace*{10pt}
{\bf Proof of Theorem \ref{thm:gc5}:}
The function $s_1(z_1)$ associated with $\bar{\mu}+\bar{\beta}_e$ gives rise to the following partition of $D$:
\begin{multline*}
Z=([c_1,c_1+b_1b_2/c_2]\times\{c_2\}),\,A=([c_1,(c_1+b_1)/2]\times[c_2,c_2+b_2])\backslash Z,\\W=D\backslash A\backslash Z.
\end{multline*}
Observe that $Z$ consists only of a portion of the bottom boundary of $D$. We now construct the allocation function as follows.
$$
  q(z)=\begin{cases}(0,1)&\mbox{if }z\in Z\cup A\\(1,1)&\mbox{if }z\in W.\end{cases}
$$
We have fixed $q(\cdot)=(0,1)$ in the region $Z$ instead of $(0,0)$. We now partition $W$ exactly as in (\ref{eqn:W-4d}), and construct the $\gamma$ function the same way as we constructed for Figure \ref{fig:d}. The $\gamma$ so constructed is feasible because we have (i) $(\gamma_1-\gamma_2)^{D\backslash Z}=\bar{\mu}^{D\backslash Z}+\bar{\beta}_e$, and (ii) $\bar{\beta}_e\succeq_{cvx}0$ from Proposition \ref{prop:beta-eh}.

The utility function $u$ is a constant in the interval $\tilde{D}^{(1)}=[c_1,(c_1+b_1)/2]\times\{c_2+b_2\}$ because $q_1=0$. So we have $\int_{\tilde{D}^{(1)}}u\,d\bar{\beta}_e=u\bar{\beta}_e(\tilde{D}^{(1)})=0$. Thus the condition $\int_Du\,d(\gamma_1-\gamma_2-\bar{\mu})=0$ is satisfied. The proof that $u(z)-u(z')=\|z-z'\|_1$ holds $\gamma$-a.e., is the same as in the proof of Proposition \ref{prop:known}. The optimal mechanism thus is as in Figure \ref{fig:e} for $\{c_1\leq b_1, c_2\geq 2b_2(b_1/(b_1-c_1))^2\}$.\qed

\vspace*{10pt}
\textbf{Proof of Theorem \ref{thm:figc}:} We prove this theorem by fixing $\theta=\bar{\mu}_+^W$, and showing that $\theta$ satisfies all the three conditions in Lemma \ref{lem:DDT-equiv}. Conditions (ii) and (iii) are trivially true. Condition (i) involves proving that $\bar{\mu}^W\succeq_1 0$. From the proof of Lemma \ref{lem:p^*-mu-bar}, we know that $\bar{\mu}^W\succeq_1 0$ holds when $c_i, b_i$ satisfy (\ref{eqn:mu-inc-dominant-12}). So it suffices to prove that $c_i\geq b_i$ satisfies (\ref{eqn:mu-inc-dominant-12}).

Suppose $p^*\geq(2b_1-c_1)/3$. Then we have
\begin{multline*}
  b_1(c_2-b_2+2p^*)-((2b_1-c_1)_+)^2/6\\\geq b_1(c_2-b_2)+2b_1(2b_1-c_1)/3-((2b_1-c_1)_+)^2/6\geq 0
\end{multline*}
where the last inequality follows from $c_2\geq b_2$. So when $c_2\geq b_2$, either $b_1(c_2-b_2+2p^*)-((2b_1-c_1)_+)^2/6\geq 0$ or $p^*\leq(2b_2-c_2)/3$ is satisfied. Similarly when $c_1\geq b_1$, either $b_2(c_1-b_1+2p^*)-((2b_2-c_2)_+)^2/6\geq 0$ or $p^*\leq(2b_2-c_2)/3$ is satisfied.\qed

\section{Proof of Theorem \ref{THM:EXTENSION}}\label{app:c}

Given $f_1(z)=f_2(z)=2z/(c+1),\,z\in[c,c+1]$, we compute the components of $\bar{\mu}$-measure as
\begin{align*}
  \mu(z_1,z_2)&=-5f_1(z_1)f_2(z_2)\mbox{ when }(z_1,z_2)\in [c,c+1]^2,\\
  \mu_s(c,z_2)&=-\frac{2c^2}{2c+1}f_2(z_2),\quad\mu_s(c+1,z_2)=\frac{2(c+1)^2}{2c+1}f_2(z_2),
  \\&\hspace*{2.5in}\mbox{ when }z_2\in[c,c+1],\\
  \mu_s(z_1,c)&=-\frac{2c^2}{2c+1}f_1(z_1),\quad\mu_s(z_1,c+1)=\frac{2(c+1)^2}{2c+1}f_1(z_1),
  \\&\hspace*{2.5in}\mbox{ when }z_1\in[c,c+1].\\
    \mu_p(c,c)&=1.
\end{align*}
We note that $1-F_i(c+z_i)=((c+1)^2-(c+z_i)^2)/(2c+1)$.

Let $c=0$. We now identify the canonical partition of $D$ using the same steps enumerated in Section \ref{sec:zero}.
\begin{enumerate}
\item[(a)] We compute the outer boundary functions $s_i(z_i)$ using (\ref{eqn:si-zi}), which is given by $s_i(z_i)=\sqrt{0.6}=0.7746$ for all $z_i\in[0,1)$.
\item[(b)] We now compute the critical price $p$ so that $\bar{\mu}(Z)=0$. This is equivalent to having $\bar{\mu}(W)=0$. From Figure \ref{fig:a}, we have $\bar{\mu}(W)=0$ iff
\begin{multline*}
  2(1-F_1(p)+1-F_2(p))-5(1-F_2(p))(1-F_1(p))\\+5\int_{p}^{\sqrt{0.6}}\int_{p}^{p+\sqrt{0.6}-z_2}f_1(z_1)f_2(z_2)\,dz_1\,dz_2=0.
\end{multline*}
Substituting the values of $\mu_s$, $f_1$, and $f_2$, we arrive at $-(5/6)p^4-(10/3)\sqrt{0.6}p^3+3p^2+2\sqrt{0.6}p-0.7=0$. Solving this bi-quadratic equation, we get $p=1.09597$ to be the only solution that is greater than $\sqrt{0.6}$.
\item[(c)] The critical points $P$ and $Q$ are given by $P=(0.32137,0.7746)$ and $Q=(0.7746,0.32137)$. We thus have the canonical partition of $D\backslash Z$:
\begin{multline*}
  A=[0,0.32137]\times[0.7746,1],\,B=[0.7746,1]\times[0,0.32137],\\W=D\backslash(A\cup B\cup Z).
\end{multline*}
\end{enumerate}

So the structure is now as in Figure \ref{fig:a}, with $p_{a_i}=\sqrt{0.6}$, $a_i=0$, and $p=1.09597$. To prove that this is indeed the optimal mechanism, we construct $\gamma(\cdot,\cdot)$ and $q(\cdot)$ as in the proof of Proposition \ref{prop:known}. Verifying that $\gamma$ and $u$ satisfy the primal constraints, dual constraints, and the conditions in Lemma \ref{lem:conditions}, is also the same as in proof of Proposition \ref{prop:known}.

For $c>0$, the outer boundary function $s_i(z_i)$ is not defined at $z_i=c$. We thus add the shuffling measure suggested in (\ref{eqn:gen-shuffle}), which after substitution of $\Delta$ and $f$, becomes
\begin{align*}
  \alpha_p(c,c+1)&=2c^2((c+1)^2-(c+p_{a_1})^2)/(2c+1)^2\\
  \alpha_s(z_1,c+1)&=2z_1(3(c+1)^2-5(c+p_{a_1}-a_1(z_1-c))^2)/(2c+1)^2,\\&\hspace*{3in}z_1\in[c,P_1].
\end{align*}
We add the shuffling measure $\bar{\alpha}$ at the interval $\tilde{D}^{(1)}$. Also, we fix $p_{a_1}=p_{a_2}$, $a_1=a_2$, $P_1=Q_2$, and $P_2=Q_1$, since the case we have considered is symmetric (i.e., $f_1=f_2$). So we add the same shuffling measure at the interval $\tilde{D}^{(2)}$.

We now identify the canonical partition of $D$ as follows: (a) Computing the outer boundary function $s_1(z_1)$, (ii) Imposing the constraint $\bar{\alpha}\succeq_{cvx}0$, and (iii) Imposing the constraint $\bar{\mu}(W)=0$. While constructing the canonical partition, we will use the fact that $s_1(z_1)=s_2(z_2)$, $p_{a_1}=p_{a_2}$, $a_1=a_2$, $P_1=Q_2$, and $P_2=Q_1$.
\begin{enumerate}
\item[(a)] The outer boundary function $s_1(z_1)$ with respect to $\bar{\mu}+\bar{\alpha}$ is computed using (\ref{eqn:si-zi-small}), and are given by
$$
  s_1(z_1)=\begin{cases}c+p_{a_1}-a_1(z_1-c)&\mbox{if }z_1\in[c,P_1]\\\sqrt{0.6}(c+1)&\mbox{if }z_1\in(P_1,c+1]\end{cases}
$$
 \item[(b)] We now impose the constraint $\bar{\alpha}\succeq_{cvx}0$. From Proposition \ref{prop:cvx}, we know that this constraint is satisfied if $\bar{\alpha}([c_1,P_1]\times\{c_2+b_2\})=\int_c^{P_1}(x-c_1)\,\bar{\alpha}(dx,c_2+b_2)=0$. We now find the conditions for both of these equations to hold.
\begin{equation}\label{eqn:marginal}
  2c^2((c+1)^2-(c+p_{a_1})^2)+\int_c^{P_1}2z_1(3(c+1)^2-5(c+p_{a_1}-a_1(z_1-c))^2)\,dz_1=0
\end{equation}
\begin{equation}\label{eqn:expectation}
  \int_c^{P_1}2z_1(z_1-c)(3(c+1)^2-5(c+p_{a_1}-a_1(z_1-c))^2)\,dz_1=0
\end{equation}
 \item[(c)] We now impose the constraint $\bar{\mu}(W)=0$.
\begin{multline}\label{eqn:mu-W=0}
  2\frac{2(c+1)^2}{(2c+1)^2}((c+1)^2-P_1^2)-5\left(\frac{(c+1)^2-P_1^2}{2c+1}\right)^2\\+5\int_{P_1}^{c+p_{a_1}-a_1(P_1-c)}\int_{P_1}^{2c+p_{a_1}+(P_1-c)(1-a_1)-z_2}f_1(z_1)f_2(z_2)\,dz_1\,dz_2=0.
\end{multline}
The parameters are found by solving (\ref{eqn:marginal}), (\ref{eqn:expectation}), and (\ref{eqn:mu-W=0}) simultaneously in $(p_{a_1}, a_1, P_1)$.
 \item[(d)] For c=0.1, we get the solution as $p_{a_1}=0.796151$, $a_1=0.231984$, and $P_1=c+0.264655=0.364655$. We also have $P_2=c+p_{a_1}-a_1(P_1-c)=c+0.734755=0.834755$, and the critical price $p=P_1+P_2=1.19941$. The canonical partition of $D$ is thus given by
\begin{align*}
Z&=\{(z_1,z_2):z_2\leq s_1(z_1)\}\cap\{(z_1,z_2):z_1\leq s_2(z_2)\}\\
 &\hspace*{1in}\cap\{(z_1,z_2):z_1+z_2\leq 2c+p\},\\
A&=([c_1,P_1]\times[s_1(z_1),c_2+b_2])\backslash Z,\\
B&=([s_2(z_2),c_1+b_1]\times[c_2,Q_2])\backslash Z,\\
W&=D\backslash(A\cup B\cup Z).
\end{align*}
\end{enumerate}

The allocation and the payment functions $(q(z),t(z))$ are given by (\ref{eqn:q-full}).

So for $c=0.1$, the canonical partition results in the structure given in Figure \ref{fig:a}. By constructing $\gamma(\cdot,\cdot)$ and $q(\cdot)$ the same way as in the proof of Proposition \ref{prop:known}, we prove that this is indeed the optimal mechanism.

One could verify that the structure of the optimal mechanism is the same for $c\in[0,0.250116]$.\qed

\section{The Weak Duality Result}\label{app:d}
In this section, we explain the weak duality relation between (\ref{eqn:dual}) and (\ref{eqn:primal}). Consider the primal problem
$$
\max_{\substack{{u(z)-u(z')\leq\|z-z'\|_1}\\{u\mbox{ cont, conv, inc}}}}\int_Du(z)\,d\bar{\mu}(z).
$$
The problem can be rewritten as
$$
\max_{(u\mbox{ cont, conv, inc})}\min_{\gamma\geq 0}\int_Du(z)\,d\bar{\mu}(z)+\int_{D\times D}(\|z-z'\|_1-u(z)+u(z'))\,d\gamma(z,z').
$$
This is because if $u(z)-u(z')>\|z-z'\|_1$, then the minimizer can choose $\gamma$ to make the second integral $-\infty$. Hence the maximizer would choose $u$ so that $u(z)-u(z')\leq\|z-z'\|_1$ for all $z,z' \in D$. The quantity $d\gamma(z,z')$ can be viewed as the price (or price measure) for violating the constraint $u(z)-u(z')\leq\|z-z'\|_1$.

The dual of this problem is
$$
\min_{\gamma\geq 0}\max_{(u\mbox{ cont, conv, inc})}\int_Du(z)\,d\bar{\mu}(z)+\int_{D\times D}(\|z-z'\|_1-u(z)+u(z'))\,d\gamma(z,z').
$$
Defining $\gamma_1(z)=\int_D\gamma(z,dz')$ and $\gamma_2(z')=\int_D\gamma(dz,z')$, we rewrite the dual problem as
$$
\min_{\gamma\geq 0}\max_{(u\mbox{ cont, conv, inc})}\int_Du(z)\,d(\bar{\mu}(z)-(\gamma_1(z)-\gamma_2(z)))+\int_{D\times D}\|z-z'\|_1\,d\gamma(z,z').
$$
We now claim that the problem above is the same as
$$
\min_{\gamma:\gamma_1-\gamma_2\succeq_{cvx}\bar{\mu}}\int_{D\times D}\|z-z'\|_1\,d\gamma(z,z').
$$
This is because if $\gamma$ is such that $\gamma_1-\gamma_2\nsucceq_{cvx}\bar{\mu}$, then the maximizer can choose $u$ such that $\int_Du(z)\,d(\bar{\mu}(z)-(\gamma_1(z)-\gamma_2(z)))>0$, and thus drive the first integral to $\infty$.

This weak duality result provides us with an understanding of how the dual arises and why $\gamma$ may be interpreted as prices for violating the primal constraint.

\section{An Alternate View on the Construction of Optimal Mechanism}\label{app:e}
Consider the structure of the mechanism to be as depicted in Figure \ref{fig:a}. Let the price for the menu items $(a_1,1)$, $(1,a_2)$, and $(1,1)$ be given by $t_{a_1}$, $t_{a_2}$ and $t_1$. Then, the expected revenue generated from this mechanism, $R$, is given by
$$
  R:=t_1\cdot\int_Wf(x)\,dx+t_{a_1}\cdot\int_Af(x)\,dx+t_{a_2}\cdot\int_Bf(x)\,dx.
$$
The parameters $t_{a_1}$, $t_{a_2}$, and $t_1$ can be written using the parameters $p_{a_1}$, $p_{a_2}$, and $p$ in Figure \ref{fig:a} as
$$
  t_{a_1}=a_1c_1+c_2+p_{a_1},\, t_{a_2}=a_2c_2+c_1+p_{a_2},\,t_1=c_1+c_2+p.
$$
Similarly, the points $P$ and $Q$ can be rewritten using the price parameters as
$$
  P=\left(\frac{t_1-t_{a_1}}{1-a_1},\frac{t_{a_1}-a_1t_1}{1-a_1}\right),\,Q=\left(\frac{t_{a_2}-a_2t_1}{1-a_2},\frac{t_1-a_2t_{a_2}}{1-a_2}\right).
$$
Substituting these values in the expression of $R$, we get
\begin{multline}\label{eqn:R-left}
  R=t_1\left(\left(c_1+b_1-\frac{t_1-t_{a_1}}{1-a_1}\right)\left(c_2+b_2-\frac{t_1-t_{a_2}}{1-a_2}\right)\right.\\
  \left.-\frac{1}{2}\left(\frac{(1-a_1)t_{a_2}+(1-a_2)t_{a_1}-(1-a_1a_2)t_1}{(1-a_1)(1-a_2)}\right)^2\right)\\
  +t_{a_1}\left(\left(c_2+b_2-\frac{t_{a_1}-a_1t_1}{1-a_1}\right)\left(\frac{t_1-t_{a_1}}{1-a_1}-c_1\right)-\frac{a_1}{2}\left(\frac{t_1-t_{a_1}}{1-a_1}-c_1\right)^2\right)\\
  +t_{a_2}\left(\left(c_1+b_1-\frac{t_{a_2}-a_2t_1}{1-a_2}\right)\left(\frac{t_1-t_{a_2}}{1-a_2}-c_2\right)-\frac{a_2}{2}\left(\frac{t_1-t_{a_2}}{1-a_2}-c_2\right)^2\right).
\end{multline}
Observe that the expression is not jointly concave in $(a_1,a_2,t_{a_1},t_{a_2},t_1)$. This is evident from the presence of product terms such as $(1-a_2)t_{a_1}$, $a_1t_1$, and $(1-a_1a_2)$. First order conditions will then only ensure local optimality.

It turns out that the first order conditions to maximize $R$ are exactly the conditions derived using the dual approach. We now have the following result.

\begin{proposition}\label{prop:opt-specific}
Consider the problem of maximizing $R$ over the allocation parameters $(a_1,a_2)$, and the price parameters $(t_{a_1},t_{a_2},t_1)$. The first order conditions turn out to be (\ref{eqn:pa1-pa2}) and (\ref{eqn:big-W}), with the parameters $(p_{a_1},p_{a_2})$ replaced by the respective allocation and price parameters used in this problem.
\end{proposition}
\begin{remark}
Note that the dual approach does provide a certificate for global optimality, but the first order approach does not.
\end{remark}

\begin{proof}
We know that the expression for the expected revenue is equal to $\int_Du\,d\bar{\mu}$, the objective function of the primal problem (\ref{eqn:primal}). We thus rewrite $R$ as
$$
  R=\int_A(a_1z_1+z_2-t_{a_1})\,\bar{\mu}(dz)+\int_B(z_1+a_2z_2-t_{a_2})\,\bar{\mu}(dz)+\int_W(z_1+z_2-t_1)\,\bar{\mu}(dz).
$$

Given that the mechanism is as in Figure \ref{fig:a}, the utility of the buyer is given by $u(z)=\max(0,a_1z_1+z_2-t_{a_1},z_1+a_2z_2-t_{a_2},z_1+z_2-t_1)$. The valuations $z\in A$ iff $a_1z_1+z_2-t_{a_1}\geq\max(0,z_1+a_2z_2-t_{a_2},z_1+z_2-t_1)$. Similarly, $z$ belongs to the respective region $Z$, $B$, or $W$, iff the expression relating to the allocation and price schedules of the region turns out to be the maximum.

We now use the same steps in the proof of \cite[Thm.~1]{MV06}, and obtain $\frac{\partial R}{\partial t_1}=\bar{\mu}(W)$, $\frac{\partial R}{\partial t_{a_1}}=\bar{\mu}(A)$, $\frac{\partial R}{\partial t_{a_2}}=\bar{\mu}(B)$, $\frac{\partial R}{\partial a_1}=\int_Az_1\,\bar{\mu}(dz)$, and $\frac{\partial R}{\partial a_2}=\int_Bz_2\,\bar{\mu}(dz)$.

We now claim that $\bar{\mu}(A)=-\bar{\alpha}^{(1)}([c_1,P_1]\times\{c_2+b_2\})$. Recall from (\ref{eqn:si-zi-small}) that the outer boundary function $s_1(\cdot)$ with respect to $\bar{\mu}+\bar{\alpha}^{(1)}$ satisfies $\int_{s_1(z_1)}^{c_2+b_2}\mu(z_1,dz_2)+(\mu_s+\alpha^{(1)}_s)(z_1,c_2+b_2)=0$ for all $z_1\in(c_1,P_1]$, and $\int_{s_1(c_1)}^{c_2+b_2}\mu_s(c_1,dz_2)+\alpha^{(1)}_p(c_1,c_2+b_2)=0$. From Figure \ref{fig:a}, we have
\begin{align*}
  \bar{\mu}(A)&=\int_{c_1}^{P_1}\left(\int_{s_1(z_1)}^{c_2+b_2}\mu(z_1,dz_2)+\mu_s(z_1,c_2+b_2)\right)+\int_{s_1(z_1)}^{c_2+b_2}\mu_s(z_1,dz_2)\\
  &=-\int_{c_1}^{P_1}\alpha_s^{(1)}(dz_1,\{c_2+b_2\})-\alpha_p^{(1)}(c_1,c_2+b_2)\\
  &=-\bar{\alpha}^{(1)}([c_1,P_1]\times\{c_2+b_2\})
\end{align*}
which is what we claimed. By the same argument, we have $\int_Az_1\,\bar{\mu}(dz)=-\int_{[c_1,P_1]}z_1\bar{\alpha}^{(1)}(dz_1,c_2+b_2)$. In a similar way, we can show that $\bar{\mu}(B)=-\bar{\alpha}^{(2)}(\{c_1+b_1\}\times[c_2,Q_2])$, and $\int_Bz_1\,\bar{\mu}(dz)=-\int_{[c_2,Q_2]}z_2\bar{\alpha}^{(2)}(c_1+b_1,dz_2)$. By Proposition \ref{prop:cvx}, the conditions $\frac{\partial R}{\partial t_{a_i}}=\frac{\partial R}{\partial a_i}=0$ can be replaced by $\bar{\alpha}^{(i)}\succeq_{cvx}0$.

Thus the first order conditions to find the local maximum of $R$ are as follows: $\bar{\alpha}^{(1)}\succeq_{cvx}0$, $\bar{\alpha}^{(2)}\succeq_{cvx}0$, and $\bar{\mu}(W)=0$. In addition, we also impose the condition $P_1+P_2=Q_1+Q_2$, because the line separating the regions with allocations $(0,0)$ and $(1,1)$ should be a line of slope $-135^\circ$ in order to have a utility function $u=\nabla q$. Observe that these are exactly the conditions that we imposed in our original dual problem, to reduce it to solving (\ref{eqn:pa1-pa2}) and (\ref{eqn:big-W}) simultaneously.
\end{proof}

This shows that if the structure of the mechanism is given to be as in Figure \ref{fig:a}, the first order conditions to find the local maximum of the expected revenue is the same as the conditions we derived using the dual method. This is not unexpected, because the method of constructing the canonical partition in Section \ref{sec:zero} was designed by \citet{DDT13} for optimal mechanisms that have the allocations and payments $(q,t)$ defined according to (\ref{eqn:q-full}). We also expect the first order conditions to be the same as those we have derived using the dual method, for all possible structures of the mechanism.

In addition to the limitation that the first order conditions derived in the proof of Proposition \ref{prop:opt-specific} can only guarantee local optimality, another limitation is that the local optimality is specific to the structure in Figure \ref{fig:a}. Given that there are more structures possible with at most $4$ menu items (see Figure \ref{fig:anomaly}) beyond the structures depicted in Figure \ref{fig:gen-structure}, the computation of best mechanisms for a given structure does not give a proof of its global optimality over all structures.

\section*{References}

\bibliographystyle{elsarticle-harv} \bibliography{two_item_uniform_cases}

\end{document}